\documentclass[reqno]{amsart}
\usepackage{amsmath,amsfonts,amssymb,amsthm}
\usepackage[utf8x]{inputenc}  
\usepackage[breaklinks=true]{hyperref}
\usepackage{float}
\usepackage{color}
\usepackage[british]{datetime2}
\usepackage{todonotes} 
\usepackage{bm} 

\usepackage{amssymb}
\usepackage{graphicx}
\usepackage[cmtip,all]{xy}
\usepackage{verbatim}
\usepackage{tikz-cd}
\usepackage{mathtools}
\usepackage{graphicx}
\usepackage{mathrsfs}
\usepackage{booktabs}
\usepackage{upgreek}

\usepackage{braket}
\usepackage{tensor}

\usepackage[T2A,T1]{fontenc}
\newcommand{\Be}{\mbox{\usefont{T2A}{\rmdefault}{m}{n}\CYRB}}
\newcommand{\be}{\mbox{\usefont{T2A}{\rmdefault}{m}{n}\cyrb}}

\usepackage{tikz-cd}

\usepackage{outlines}
\usepackage{enumitem}
\setlist{leftmargin=\parindent} 

\usepackage[normalem]{ulem}

\newcommand{\p}{\mathfrak{p}} 
\renewcommand{\t}{\mathfrak{t}} 
\newcommand{\cN}{\mathcal{N}} 
\newcommand{\M}{\mathbb{M}} 

\newcommand{\D}{\mbox{\sffamily{D}}} 
\newcommand{\PP}{\mbox{\sffamily{P}}} 

\newcommand{\Lag}{L} 
\renewcommand{\L}{L} 
\newcommand{\Lie}{\mathcal{L}} 
\newcommand{\Y}{\mathcal{Y}} 
\newcommand{\pd}[2][]{\frac{\partial #1}{\partial #2}} 
\newcommand{\ev}{\mathpzc{ev}} 
\newcommand{\Diff}{\operatorname{Diff}} 
\newcommand{\ber}{\operatorname{\be}} 
\newcommand{\Ber}{\Be} 
\newcommand{\bersheaf}{\Be} 
\newcommand{\A}{\mathcal{A}} 
\newcommand{\dR}{\mathpzc{dR}} 
\newcommand{\Sp}{\mathpzc{Sp}} 
\newcommand{\pr}{\mathrm{pr}} 

\renewcommand{\c}{\mathpzc{c}} 
\renewcommand{\dR}{\mathpzc{dR}} 
\renewcommand{\Sp}{\mathpzc{Sp}} 

\newcommand{\fib}{\mathpzc{F}} 

\newcommand{\Oh}{\mathcal{O}} 
\newcommand{\Cal}{\mathcal{C}} 
\newcommand{\J}{\mathcal{J}} 
\newcommand{\HOM}{\operatorname{\mathcal{H}om}} 
\renewcommand{\S}{\mathcal{S}} 
\newcommand{\T}{\mathcal{T}} 
\newcommand{\PV}{\mathcal{PV}} 
\newcommand{\Der}{\operatorname{Der}} 

\newcommand{\Geo}{\mathscr{F}} 
\newcommand{\Sup}{\mathscr{E}} 
\newcommand{\Rheo}{\mathscr{R}} 
\newcommand{\Conv}{\mathscr{C}} 
\newcommand{\GeoComp}{\mathcal{F}} 
\newcommand{\SupComp}{\mathcal{E}} 
\renewcommand{\Gauge}{\mathscr{G}} 
\newcommand{\GaugeComp}{\mathcal{G}} 

\newcommand{\scH}{\mathscr{H}} 

\newcommand{\scal}{R} 
\newcommand{\tens}{G} 
\newcommand{\curv}{\Omega} 

\newcommand{\rheo}[1]{\bm{#1}} 
\newcommand{\conv}[1]{\bm{#1}} 
\newcommand{\comp}[1]{{#1}_0} 

\newcommand{\ip}[1]{\langle #1 \rangle} 

\newcommand{\sugra}{\mathfrak{sugra}} 

\newcommand{\susy}{\mathpzc{susy}}
\newcommand{\geom}{\mathpzc{geom}}
\newcommand{\compo}{\mathpzc{comp}}

\newcommand{\red}{\mathpzc{red}} 
\newcommand{\can}{\mathpzc{can}} 
\renewcommand{\flat}{\mathpzc{flat}} 
\newcommand{\cl}{\mathpzc{cl}}

\newcommand{\trun}{\tau} 
\newcommand{\shift}{\sigma} 

\newcommand{\maps}{\colon}    
\newcommand{\R}{{\mathbb R}}  
\newcommand{\Z}{{\mathbb Z}}  

\newcommand{\SO}{{\rm SO}}     
\newcommand{\SL}{{\rm SL}}     
\newcommand{\Spin}{{\rm Spin}} 

\newcommand{\so}{{\mathfrak{so}}}  

\newcommand{\Hom}{\operatorname{Hom}} 

\newcommand{\inclusion}{\hookrightarrow}
\newcommand{\iso}{\cong} 
\newcommand{\isoto}{\stackrel{\sim}{\longrightarrow}} 

\newcommand{\vol}{\operatorname{vol}} 

\newcommand{\define}[1]{{\bf \boldmath{#1}}}
\newcommand{\arxiv}[1]{\href{http://arxiv.org/abs/#1}{arXiv:{#1}}}

\theoremstyle{plain}
\newtheorem{thm}{Theorem}[section]
\newtheorem{cor}[thm]{Corollary}

\newtheorem{prop}[thm]{Proposition}
\newtheorem{claim}[thm]{Claim}

\newtheorem{convention}[thm]{Convention}

\newtheorem{theorem}[thm]{Theorem}
\newtheorem{lemma}[thm]{Lemma}
\newtheorem{corollary}[thm]{Corollary}

\theoremstyle{definition}
\newtheorem{rem}[thm]{Remark}
\newtheorem{defn}[thm]{Definition}

\newtheorem{definition}[thm]{Definition}
\newtheorem{example}[thm]{Example}

\theoremstyle{remark}
\newtheorem{remark}[thm]{Remark}

\DeclareMathAlphabet{\mathpzc}{OT1}{pzc}{m}{it}

\numberwithin{equation}{section}

\makeatletter
\DeclareRobustCommand{\cev}[1]{%
  \mathpalette\do@cev{#1}%
}
\newcommand{\do@cev}[2]{%
  \fix@cev{#1}{+}%
  \reflectbox{$\m@th#1\vec{\reflectbox{$\fix@cev{#1}{-}\m@th#1#2\fix@cev{#1}{+}$}}$}%
  \fix@cev{#1}{-}%
}
\newcommand{\fix@cev}[2]{%
  \ifx#1\displaystyle
    \mkern#23mu
  \else
    \ifx#1\textstyle
      \mkern#23mu
    \else
      \ifx#1\scriptstyle
        \mkern#22mu
      \else
        \mkern#22mu
      \fi
    \fi
  \fi
}

\makeatother

\usepackage[top=3.2cm, bottom=3.2cm, left=3.2cm, right=3.2cm]{geometry}

\newcommand{\slantone}[2]{{\raisebox{.1em}{$#1$}\left/\raisebox{-.1em}{$#2$}\right.}}
\newcommand*{\defeq}{\mathrel{\vcenter{\baselineskip0.5ex \lineskiplimit0pt
                     \hbox{\scriptsize.}\hbox{\scriptsize.}}}%
                     =}

\newcommand\asim{\mathrel{%
  \ooalign{\raise0.1ex\hbox{$\sim$}\cr\hidewidth\raise-0.8ex\hbox{\scalebox{0.9}{$\scriptstyle{x}$}}\hidewidth\cr}}}

\newcommand{\mani}{\ensuremath{\mathpzc{X}}}
\newcommand{\para}{\ensuremath{\mathpzc{S}}}

\newcommand{\stsheaf}{\ensuremath{\mathcal{O}_{\mathpzc{M}}}}

\newcommand{\beq}{\begin{equation}}
\newcommand{\eeq}{\end{equation}}
\newcommand{\bear}{\begin{eqnarray}}
\newcommand{\eear}{\end{eqnarray}}

\title[Poincar\'e duality and supergravity]{Poincar\'e duality and supergravity}

\author{Konstantin Eder}
\address{Institute for Quantum Gravity (IQG), Friedrich-Alexander-Universit\"at  Erlangen-N\"urnberg}
\curraddr{Staudtstrasse 7 / B2, 91058 Erlangen (Germany)}
\email{konstantin.eder@fau.de}

\author{John Huerta}
\address{Faculty of Physics, Ludwig Maximilian University}
\curraddr{Theresianstrasse 37, 80333, Munich (Germany)}
\email{john.huerta@tecnico.ulisboa.pt}

\author{Simone Noja}
\address{Dipartimento di Matematica, Universit\'a degli Studi di Bari `Aldo Moro'}
\curraddr{via Edoardo Orabona 4, 75125, Bari (Italy)}
\email{simone.noja@uniba.it}

\begin{document}
\maketitle

\begin{abstract}

We study relative differential and integral forms on families of supermanifolds and their cohomology. 
We prove a relative Poincar\'e--Verdier duality and show that it
relates the cohomology of differential and integral forms, admitting a concrete geometric realization
via Berezin fiber integration. We further introduce the Poincar\'e--dual integral form associated to
an embedded even family and prove that it satisfies the correct localization property. We then apply these results to
supergravity, focusing on the $3d$ case. In this setting, we show that relative Poincar\'e
duality provides the natural framework for encoding the data needed to relate a superspace
formulation to the physical spacetime, thereby yielding a rigorous definition of picture changing
operators used in the physics literature. Building on this, after a careful analysis of the space
of fields and the relevant constraints, we prove that the component, superspace, and geometric
formulation of the theory are all equivalent. Finally, under suitable hypotheses, we argue that our construction
illustrates a general principle governing the mathematical formulation of classical
field theories on supermanifolds.

\end{abstract}

\tableofcontents

\section{Introduction}

In de Rham theory, Poincar\'e duality admits a beautiful geometric incarnation: wedge product and
integration furnish a pairing between ordinary differential forms and compactly supported forms,
pairing degree $p$ with degree $n-p$ on an $n$-dimensional oriented manifold, and descending to a
perfect pairing on cohomology \cite{BottTu}. We will call this pairing, the \emph{wedge--and--integrate pairing}, Poincar\'e pairing: 
one may regard this  as the basic mechanism by which geometry produces duality.

The purpose of this paper is twofold: first, to develop a supergeometric analogue of the classical
wedge--and--integrate mechanism underlying Poincar\'e duality, formulated in the relative setting of
\emph{families} of supermanifolds; and second, to demonstrate the concrete effectiveness of this
formalism by applying it to a physical theory, supergravity.
Both aspects are essential: the abstract derived duality framework provides the correct notion of
integration and localization in superspace, while the supergravity case shows that these structures
are not merely formal, but supply a rigorous and usable toolkit for treating physical theories on
superspaces.

On a supermanifold $\mani$, the de Rham complex is qualitatively different from the
classical one \cite{BerLei, Kostant, Manin, Noja}. Besides anticommuting one--forms such as $dx$ coming from even directions, there are
commuting one--forms such as $d\theta$ coming from odd directions; consequently the de Rham complex
is not bounded above, and in particular there is no ``top degree'' form. From the classical point of
view, this is precisely the obstruction to integrating differential forms on $\mani$ in a way that sees
the odd directions. The correct integration theory instead involves \emph{integral forms}, and their complex, 
which we call \emph{Spencer complex} after $\mathcal{D}$-module theory \cite{BerLei, Penkov}. 

Conceptually, integral forms arise as the twisted dual (by the Berezinian line bundle) of the de Rham complex \cite{BerLei, CNR, CGN, Lie, Noja, Pol1}. 
From the perspective of pairing, the guiding
principle is that wedge--and--integrate pairing is restored as a pairing between de Rham cohomology
of differential forms and compactly supported Spencer cohomology of integral forms \cite{BerLei, Manin, Noja}. 

A second layer appears when one takes seriously the needs of physics. Supersymmetric theories are
often written ``on a single superspace'', but computations implicitly assume an unlimited supply of
odd coefficients (supersymmetry parameters, odd moduli, auxiliary odd variables in variations, \emph{et cetera}, see for example \cite{CGV, DFClassical, DFSuper, DonWit, DonOtt1, DonOtt2, Eder, FKP}).
A clean way to encode this is to work with a \emph{family} of supermanifolds
$\varphi \maps \mani\to \para$, allowing arbitrary base change $\para'\to \para$ \cite{DFClassical, Eder, GioSati, SSR, Kessler}; in this sense the base $\para$ is not an
inessential decoration, but a device for \emph{probing the odd directions of the super moduli space}.
If one were to rigidly restrict to a single supermanifold over a point, then pullback to physical
spacetime (an ordinary manifold) forces odd coefficients to vanish and the fermionic degrees of
freedom disappear \cite{EderJMP, Hack:2015vna}. Thus, the family viewpoint is necessary already for well-definedness.

Working over a base also forces a conceptual shift in what it means to ``restrict to spacetime''.
In general, the ordinary manifold underlying a supermanifold -- its \emph{reduced space} $\mani_{\red}$ --  
is not a relative object over $\para$, and one must instead work with
underlying \emph{even families} $\varphi^\ev  \maps \mani^{\ev} \to \para $ and embeddings
$\iota \maps \mani^{\ev} \hookrightarrow \mani$
over $\para$ \cite{Eder, Kessler}. In this setting, for a trivial family $\pr_\para  \maps \fib \times \para \to \para$, the distinguished embedding
constructed from the canonical embedding $\iota_{\mathpzc{can}}  \maps \fib_{\mathpzc{red}} \hookrightarrow \fib$ of the reduced fiber manifold inside the fiber supermanifold,
plays a distinguished role: it is the geometric avatar of the physical
prescription ``set the odd coordinates to zero'' to reduce to ordinary spacetime (here the reduced fiber manifold $\fib_{\mathpzc{red}}$), but now carried out in a context where odd
parameters can survive through the base.

The first part of the paper develops a uniform theory of differential and integral forms in families,
together with their cohomology and dualities, in such a way that the absolute theory over a point is
recovered as a special case \cite{Noja}. We define relative de Rham forms and relative
integral forms, study their
cohomology with and without support conditions (and develop the relevant Cartan calculus).

Main technical results are the Poincar\'e lemmas (compactly supported and not) 
for differential and integral forms, which are indispensable for a clean cohomological integration theory 
and fills a gap in the existing literature \cite{CNR, Manin, Noja}. 
At the conceptual level, we formulate and prove a \emph{Poincar\'e--Verdier duality} for
families of supermanifolds. The natural home of this statement is that of derived categories \cite{KashSha}: we identify
the appropriate dualizing complex for a family and express duality as a derived equivalence
relating the relative de Rham complex and the relative Spencer complex. Importantly, this is not a
purely formal exercise: we also give a concrete geometric realization in terms of \emph{Berezin
fiber integration}, so that the derived adjunctions become explicit operations on relative forms.

A particularly useful concrete output is a supergeometric theory of \emph{Poincar\'e duals} for even
families of supermanifolds. To an embedding $\iota \maps \mani^{\ev}\hookrightarrow \mani$ over $\para$ we attach a
canonical closed integral $0$--form $Y_\iota$ in the Spencer complex characterized by the
localization identity
\[
  \int_{\mani^{\ev}/\para} \iota^*\omega \;=\; \int_{\mani/\para} Y_\iota\cdot \omega,
\]
for de Rham forms $\omega$ of the appropriate degree. Thus $Y_\iota$ is the Poincar\'e dual of
$\iota$ in a precise supergeometric sense: it packages the dependence on an embedded spacetime
family into a cohomological datum on the ambient superspace family.

The second part applies the theory to the mathematical foundations of supergravity, focusing on
three--dimensional ($3d$, $\mathcal{N}=1$) supergravity as a nontrivial test case \cite{CCG, Merkulov, RRvN, ZP}. Here the role of our duality
theory becomes sharply visible.

In the geometric formulation of supergravity \cite{AD, CDF, Castellani:2018zey, Cortes, FrancoisRavera, nLab}, Lagrangian densities are naturally formulated using differential forms, 
but an \emph{action} requires integration over superspace. Since differential forms
on a supermanifold do not furnish integrable top forms, one needs integral forms. In the physics
literature, the missing ingredient is supplied by \emph{picture changing operators} (PCOs) \cite{Belo, CGN, Carlo, FMS}. One of
the main conceptual novelties of this work is that relative Poincar\'e duality in families provides a
natural and intrinsic definition of PCOs: a PCO is precisely a closed integral $0$--form which is a
Poincar\'e dual to a choice of spacetime embedding inside superspace. In other words, the data needed
to relate superspace to physical spacetime is not an ad hoc analytic gadget, but a geometric and
cohomological choice governed by the same duality mechanism that underlies integration theory on supermanifolds.

We develop this point in a way that is faithful both to geometry and to the physics practice of
switching among seemingly different formulations of the same theory. We treat the three standard
approaches to supergravity \cite{CCG}: the \emph{component} formulation (where fields live on spacetime and supersymmetry is not manifest),
the \emph{superspace} formulation (where fields live on superspace and supersymmetry is geometric),
and the \emph{geometric} (or rheonomic) formulation (which mediates between the two: fields live on
superspace, but correspond one--to--one with component fields via rheonomic parametrizations).
A key structural feature is that both geometric and superspace formalisms, \emph{i.e.}, the formalisms which actually rely on supermanifolds,
organize fields into commutative ``triangles'' relating superfields, constrained superfields, and component fields.\footnote{Though, in the
superspace approach one must in general work up to gauge equivalence, naturally leading from sheaves
to stacks, while the geometric approach is more rigid.}

Within our framework, we give a mathematically precise action principle for geometric supergravity
by constructing the PCO associated with the ``canonical'' spacetime embedding, 
and defining the action by integrating the geometric Lagrangian against this
PCO. This produces an integral over superspace which is \emph{provably equal} to the usual component
action. We then analyze the superspace action and introduce a notion of \emph{supersymmetric} PCO,
designed to make supersymmetry manifest at the level of the action \cite{Carlo}. The comparison between the
geometric and superspace actions becomes a problem in cohomology: it reduces to
showing that two candidate PCOs differ by an exact term, i.e.\ are cohomologous. This is exactly the
sort of statement our duality theory is designed to formulate cleanly. In doing so, we also isolate
which parts of the comparison are genuinely subtle, namely we compare the space of fields and their constraints (with special attention to the elusive 
\emph{rheonomic constraints} \cite{AD, CDF, Castellani:2018zey, nLab}) in different formulation of the theories, and we explain where certain arguments in the
existing literature fall short. Finally, we extract from the $3d$ case a general principle: under
suitable hypotheses, mathematically well--defined supersymmetric actions in superspace are governed
by relative Poincar\'e duality and the cohomology class of the chosen PCO. \\
We stress that our framework reconciles mathematical rigor with physical consistency in supergravity
by combining three essential ingredients: the intrinsic supergeometry of differential and integral
forms (and Berezin integration), the family viewpoint needed to capture fermionic degrees of freedom,
and a careful treatment of the resulting spaces of fields. Finally, we wish to stress once again that supersymmetric field 
theories in physics continue to serve as a profound source of inspiration and a driving force for progress in pure mathematics

\medskip

\noindent\textbf{Structure of the paper.}
Section~2 introduces families of supermanifolds and the basic constructions we use, emphasizing base
change and underlying even families.
Section~3 defines relative differential and integral forms together with their de Rham and Spencer complexes,
and develops their algebraic calculus.
Section~4 studies their cohomology in the absolute setting, including in particular their Poincar\'e lemmas for compact
support and not.
Section~5 returns to families: here, after computing the dualizing complex of a family of supermanifolds, 
we prove a Poincar\'e--Verdier duality and a related notion of relative Poincar\'e duality relating de Rham and Spencer 
cohomology, and give a concrete realization via Berezin fiber integration; we
then construct Poincar\'e dual integral forms attached to embeddings of even families.
Finally, Section~6 applies the theory to $3d$ supergravity, defining PCOs geometrically and using
them to relate the component, geometric, and superspace formulations. Finally, Appendix~A develops 
Cartan calculus for integral forms, while Appendix~B gives a concrete derivation for compactly supported Poincar\'e
lemma for integral forms.

\vspace{.3cm}

\noindent
{\bf Acknowledgments.} The authors wish to thank Antonio Grassi, Roberto Catenacci and Leonardo Castellani. John Huerta
would like to thank Carlo Alberto Cremonini and Ingmar Saberi, as well as Rita Fioresi, Dimitry Leites and Andrea Santi for
discussions on supergeometry. Simone Noja wishes to thank Sergio
Cacciatori (for asking some questions that ultimately led to this paper), Sasha Voronov, Alberto Arabia, and 
acknowledges correspondence with Ugo Bruzzo, Daniel
Hern\'andez Ruip\'erez.

Part of this work was done while the authors visited the \emph{Simons Center for
  Geometry and Physics} (SCGP) and the \emph{Galileo Galileo Galilei Institute}
(GGI): the authors would like to collectively thank these institutions and the
organizers of the respective workshops (S. Grushevksy, A. Polishchuk \& S.
Cacciatori and R. Donagi \& A. Grassi) for their hospitality and the fantastic
working conditions. 

The work of John Huerta was supported by FCT/Portugal through project
UIDB/04459/2020. The work of Simone Noja was supported by the Deutsche
Forschungsgemeinschaft (DFG, German Research Foundation) Projektnummer 517493862
(Homologische Algebra der Supersymmetrie: Lokalit\"at, Unitarit\"at, Dualit\"at).

\section{Main definitions: families of supermanifolds}
\label{sec:maindefs}

\noindent
We now introduce the main objects we work with throughout this paper: supermanifolds
and families of supermanifolds. Briefly, a `supermanifold' is an ordinary manifold
equipped with a sheaf of supercommutative algebras. That sheaf turns out not to have
enough elements to support the superfield formalism in physics, so we need to work in
families: a `family of supermanifolds' is a fiber bundle of supermanifolds.
The base supermanifold in this fiber bundle provides us with additional odd
coefficients for our superfields, and base change allows us to obtain more as needed.

For a general introduction to supermanifolds, we suggest the reader refer to the
classic book by Manin \cite{Manin}. For references that focus on families, we
recommend the more recent work of Eder \cite{Eder} and Kessler \cite{Kessler}.

\begin{definition}[Superspaces and Morphisms]
  A \define{superspace} is a locally ringed space $\mani \defeq (|\mani|,
  \mathcal{O}_{\mani})$ whose structure sheaf is given by a sheaf of
  $\mathbb{Z}_2$-graded supercommutative rings $\mathcal{O}_{\mani} \defeq
  \mathcal{O}_{\mani, 0} \oplus \mathcal{O}_{\mani, 1}$. Given two superspaces
  $\mathpzc{X}$ and $\mathpzc{Y}$, a \define{morphism of superspaces} $\varphi \maps
  \mathpzc{X} \to \mathpzc{Y}$ is a pair $\varphi \defeq (|\varphi|, \varphi^\sharp)$ where
  \begin{enumerate}[leftmargin=*]
  \item $|\varphi| \maps |\mathpzc{X}| \to |\mathpzc{Y}|$ is a continuous map
    of topological spaces;
  \item $\varphi^\sharp \maps \mathcal{O}_{\mathpzc{Y}} \to |\varphi|_*
    \mathcal{O}_\mathpzc{X}$ is a morphism of sheaves of $\Z_2$-graded
    supercommutative rings, having the properties that it preserves the
    $\Z_2$-grading and the unique maximal ideal of each stalk, i.e. $\varphi^\sharp_x
    (\mathfrak{m}_{|\varphi| (x)}) \subseteq \mathfrak{m}_x$ for all $x \in |\mathpzc{X}|.$
  \end{enumerate}
\end{definition}

\noindent
Given a superspace $\mani$, the nilpotent sections in $\mathcal{O}_\mani$ define a sheaf of ideals.

\begin{definition}[Nilpotent Sheaf]
  Let $\mani$ be a superspace. The sheaf of ideals $\mathcal{J}_\mani  \defeq \langle \mathcal{O}_{\mani, 1} \rangle =
  \mathcal{O}_{\mani,1} + (\mathcal{O}_{\mani, 1})^2$, generated by all the odd
  sections in $\mathcal{O}_\mani = \mathcal{O}_{\mani, 0 } \oplus \mathcal{O}_{\mani,
  1}$ is called the \define{nilpotent sheaf} of $\mani$.
\end{definition}

\noindent
This definition is slightly misleading for a general superspace, which can have even
nilpotent elements not generated by the odd elements. However, for a supermanifold,
$\J_\mani$ is precisely the nilpotent ideal of $\Oh_\mani$. The supermanifold case is
our exclusive concern here, and we turn to the definition now.

\begin{definition}[Smooth / Complex / Algebraic Supermanifold and Morphisms]
  \label{supermanifold}
  A \define{smooth / complex / algebraic supermanifold} is a superspace $\mani$ such that
  \begin{enumerate}[leftmargin=*]
  \item $\mani_{\mathpzc{red}} \defeq (|\mani|, \mathcal{O}_{\mani} /
    \mathcal{J}_{\mani})$ is a smooth / complex manifold or an algebraic variety;
  \item the quotient sheaf $\mathcal{J}_{\mani} / \mathcal{J}^2_\mani$ is a
    locally-free sheaf of $\mathcal{O}_{\mani} / \mathcal{J}_\mani$-modules and the
    structure sheaf $\mathcal{O}_{\mani}$ is locally isomorphic to the sheaf of
    exterior algebras $\wedge^\bullet \left( \mathcal{J}_{\mani} /
      \mathcal{J}^2_\mani \right)$ over $\mathcal{O}_{\mani} / \mathcal{J}_\mani$.
  \end{enumerate}
  We call $\mani_{\mathpzc{red}}$ the \define{reduced space} of the supermanifold
  $\mani$ and we denote $\mathcal{O}_{\mani_{\mathpzc{red}}} \defeq
  \mathcal{O}_{\mani} / \mathcal{J}_{\mani}$ and $\mathcal{F}_\mani \defeq
  \mathcal{J}_{\mani} / \mathcal{J}_{\mani}^2 $ for short. Also, we say that the
  supermanifold $\mani$ has \define{dimension} $m|n$ if $m = \dim
  \mani_{\mathpzc{red}}$ and $n = \mbox{rank}\, \mathcal{J}_\mani /
  \mathcal{J}^2_\mani$. Finally, a \define{morphism of supermanifolds} is a morphism
  of superspaces. We will also call a morphism a \define{map}, and call an
  isomorphism of supermanifolds a \define{diffeomorphism} in the case of real supermanifolds.
\end{definition}

\begin{remark}
  Note that the \define{sheaf of regular functions}
  $\mathcal{O}_{\mani_{\mathpzc{red}}}$ on $\mathcal{X}_{\mathpzc{red}}$ is the sheaf
  of smooth, holomorphic or regular functions, depending on whether $\mani_{\mathpzc{red}}$ is
  a smooth manifold, a complex manifold, or an algebraic variety, and likewise $\mathcal{F}_\mani$ is the sheaf of
  sections of a smooth, holomorphic, or algebraic vector bundle. The sheaf
  $\mathcal{O}_\mani$ is filtered by powers of its nilpotent ideal $\mathcal{J}_\mani$:
  \begin{eqnarray}
    \label{filt}
    \mathcal{O}_\mani \supset \mathcal{J}_\mani \supset \mathcal{J}_\mani^2 \supset  \ldots \supset \mathcal{J}_\mani^{n} \supset \mathcal{J}_\mani^{n+1} = 0.
  \end{eqnarray}
  This allows us to define the associated $\mathbb{Z}$-graded sheaf of
  supercommutative algebras $\mbox{Gr}^{\bullet} \mathcal{O}_\mani$ by taking
  successive quotients:
  \begin{eqnarray}
    \label{grsheaf}
    \mbox{{Gr}}^{\bullet} \mathcal{O}_\mani \defeq
    \bigoplus_{i=0}^n \mbox{{Gr}}^{i} \stsheaf = \mathcal{O}_{\mani_{\mathpzc{red}}}
    \oplus \slantone{\mathcal{J}_\mani}{\mathcal{J}^2_{\mani}} \oplus \ldots \oplus
    \slantone{\mathcal{J}_\mani^{n-1}}{\mathcal{J}^n_{\mani}} \oplus
    {\mathcal{J}^n_\mani},
  \end{eqnarray}
  where $\mbox{{Gr}}^{i} \stsheaf \defeq
  \slantone{\mathcal{J}_\mani^i}{\mathcal{J}_\mani^{i+1}}$. It is not hard to check
  that $\mbox{Gr}^\bullet \mathcal{O}_\mani \iso \bigwedge^\bullet \mathcal{F}_\mani$
  as sheaves of $\Z$-graded algebras over $\mathcal{O}_{\mani_{\mathpzc{red}}}$. This
  implies that the pair of \emph{classical} objects $(\mani_{\mathpzc{red}},
  \mathcal{F}_\mani)$, namely a manifold and a vector bundle, defines a local model
  for $\mani$ in the sense that $\mani$ is locally isomorphic to its associated
  graded supermanifold $\mbox{Gr}\, \mani \defeq (|\mani|, \mbox{Gr}^\bullet
  \mathcal{O}_\mani)$, by definition.
\end{remark}

Definition \ref{supermanifold} implies that any supermanifold comes endowed with a
short exact sequence of sheaves that encodes the relation between the supermanifold
itself and its reduced space:
\begin{eqnarray}
  \label{ses}
  \xymatrix{
  0 \ar[r] & \mathcal{J}_\mani \ar[r] & \mathcal{O}_{\mani} \ar[r]^{\iota^*} & \mathcal{O}_{\mani_{\mathpzc{red}}} \ar[r] & 0.
                                                                                                                                 }
\end{eqnarray}
The embedding $\iota \maps \mani_{\mathpzc{red}} \to \mani$ corresponding to the
morphism of sheaves $\iota^* \maps \mathcal{O}_\mani \to \iota_*
\mathcal{O}_{\mani_{\mathpzc{red}}}$ is canonical, and corresponds to modding out all
the nilpotent sections from the structure sheaf $\mathcal{O}_{\mani}$. Intuitively
speaking, we can view $\mani_\mathpzc{red}$ as the subspace defined by setting all nilpotent
sections to zero.

From here on, unless stated otherwise, all supermanifolds are understood to be real
smooth supermanifolds. We note that Definition \ref{supermanifold} allows us to
introduce local coordinates: a real smooth supermanifold of dimension $m|n$ is
locally diffeomorphic to $\R^{m|n} = (\R^m, \Cal^\infty_{\R^m} \otimes_\R
\wedge^\bullet_\R(\theta_1, \ldots, \theta_n))$, where $\Cal^\infty_{\R^m}$ denotes
the sheaf of smooth functions on $\R^m$, and $\wedge^\bullet_\R(\theta^1, \ldots,
\theta^n)$ denotes the real exterior algebra with generators $\theta^1, \ldots,
\theta^n$. The standard coordinate functions $x^1, \ldots, x^m$ on $\R^m$ are called
the \define{even coordinates} on $\R^{m|n}$, while the generators $\theta^1, \ldots,
\theta^n$ of the exterior algebra are called the \define{odd coordinates}. We will
often denote such a coordinate system by $x^a|\theta^\alpha$, the $|$ indicating the
division into even and odd just as it does with the dimension $m|n$. A supermanifold
that is diffeomorphic to an open submanifold of $\R^{m|n}$ is called a
\define{domain}; the existence of coordinates is equivalent to saying that every
supermanifold can be covered with domains.

We now come to our main object of study---instead of working with a single
supermanifold, we will consider families of them in the following sense.

\begin{definition}[Family of Supermanifolds]
  \label{deffam}
  Let $\mani$ and $\para$ be smooth supermanifolds such that $ |\para | $ is
  connected and $\dim_0 \para \leq \dim_0 \mani$, $\dim_1 \para \leq \dim_1 \mani$. A \define{family of supermanifolds} is
  a fiber bundle $\varphi \maps \mani \to \para$. \\
  This means that
  $\varphi \maps \mani \to \para$ is a smooth map of supermanifolds with the property
  that for each point $b \in |\para|$ there exist a domain $B \subseteq \para$
  containing $b$ and a diffeomorphism
  $\Phi \maps \varphi^{-1} (B) \rightarrow \mathpzc{F} \times B $, with
  $\mathpzc{F} $ the \define{fiber supermanifold} over $b$ such that the following
  commutes
  \[ \begin{tikzcd}
      \varphi^{-1} (B) \ar[rd, "\tilde \varphi", swap] \ar[rr, "\sim"] &&  \mathpzc{F} \times B  \ar[ld, "p_2"]  \\
      & B 
    \end{tikzcd} \] where $\tilde \varphi \defeq \varphi |_{\varphi^{-1} (B)}.$ In
  particular, a \define{trivial family of supermanifolds} is a trivial bundle,
  $\varphi \maps \mathpzc{F} \times \para \to \para.$

  We call $\para$ the \define{parametrizing supermanifold} or \define{base} of the
  family and $B\subseteq \para$ a \define{trivializing domain} or
  \define{trivializing open set} for $\varphi$. We will say that
the family has \define{relative dimension} $m|n$ if
$\dim \mathpzc{F} = \dim \mani - \dim \para = m|n$. In the following, we will often write
  the data of a family of supermanifolds $\varphi \maps \mani \to \para $ as
  $\mani_{/\mathpzc{S}}$ for short.
\end{definition}

\noindent
Note since we choose $\para$ to be connected, we can unambiguously speak of the fiber
$\mathpzc{F}$ of the family $\varphi \maps \mani \to \para$, as all fibers are
diffeomorphic to one another -- in the following we will also assume that the fiber
is \emph{connected}. More generally, unless otherwise stated, we will also assume
that both $\mani$ and $\para$ are \emph{orientable}, see also section
\ref{sec:fiberint}.

Notice that the above definition of a family of supermanifolds differs from the one
given in \cite{Kessler}, where the morphism $\varphi \maps \mani \to \para$ is only
assumed to be a surjective submersion, while definition \ref{deffam} can be rephrased
by saying that a family is a locally trivial surjective submersion. On the other
hand, if the surjective submersion is also proper, then it is a family of
supermanifolds in the sense of definition \ref{deffam}, by a super analog of the
Ehresmann fibration theorem\footnote{Several results of the present paper remain true
  in the more general case of family as surjective submersions as in
  \cite{Kessler}.}.

\begin{remark}[Families as morphisms of locally ringed spaces]\label{rem:family-ringed}
A family of supermanifolds $\varphi \maps \mani\to \para$ can be regarded as a morphism of locally ringed
spaces
\(
(|\varphi|,\varphi^\sharp) \maps (|\mani|,\mathcal O_\mani)\longrightarrow (|\para|,\mathcal O_\para).
\)
Here $|\varphi| \maps |\mani|\to |\para|$ is the underlying continuous map, and
\(
\varphi^\sharp \maps \mathcal O_\para \longrightarrow |\varphi|_*\mathcal O_\mani
\)
is a morphism of sheaves of supercommutative algebras satisfying the locality condition mentioned above (\emph{i.e.}\ for every
$x\in |\mani|$, the induced map on stalks
\(
\varphi^\sharp_x \maps \mathcal O_{\para,|\varphi|(x)} \longrightarrow \mathcal O_{\mani,x}
\)
sends the maximal ideal into the maximal ideal, i.e.\ $\varphi^\sharp_x(\mathfrak m_{|\varphi|(x)})\subseteq \mathfrak m_x$).
Equivalently, $\varphi^\sharp$ may be viewed as a morphism of sheaves on $|\mani|$,
\[
\varphi^\sharp \maps |\varphi|^{-1}\mathcal O_\para \longrightarrow \mathcal O_\mani,
\]
where we have used the adjunction $|\varphi|^{-1}\dashv |\varphi|_*$; in the sequel we will often
identify $|\varphi|^{-1}\mathcal O_\para$ with its image in $\mathcal O_\mani$ when convenient.

In our smooth setting, the assumption that $\varphi$ is a (surjective) submersion means that locally
$\varphi$ trivializes as a projection $F\times B\to B$. On such a trivializing open $U\simeq F\times B$,
$\varphi^\sharp$ is given by pullback along the projection,
\[
\varphi^\sharp(f)=f\circ \mbox{pr}_B,
\qquad f\in \mathcal O_B,
\]
and hence it is injective. In particular, one may regard $|\varphi|^{-1}\mathcal O_\para$ as a subsheaf of
$\mathcal O_\mani$ (though it is generally not fine).
\end{remark}

Given two families of supermanifolds, a morphism relating them is defined as follows.
\begin{definition}[Morphisms of Families of Supermanifolds] \label{defmor} Let $\mani_{/ \para}$ and $\mani'_{/ \para'}$ be families of supermanifolds over $\para$ and $\para'$ respectively. Then a morphism $f \maps \mani_{/ \para} \to \mani'_{/ \para'}$ of families of supermanifolds is a commutative square
\begin{eqnarray} \label{basec}
\xymatrix{ \mani \ar[r]^{\tilde f} \ar[d]_{\varphi_{\mani}} & \mani' \ar[d]^{\varphi_{\mani'}} \\
\para \ar[r]^{b} & \para',
}
\end{eqnarray}
where $\tilde f $ and $b$ are morphisms of supermanifolds. If in particular $b = id_\para$, then $f \maps \mani_{/ \para} \to \mani'_{/ \para} $ is called morphism of supermanifolds over $\para$, and it is given by a commutative triangle
\begin{eqnarray} \label{mapfam}
\xymatrix{ \mani \ar[rr]^{\tilde \varphi} \ar[dr]_{\varphi_{\mani}} & & \mani' \ar[dl]^{\varphi_{\mani'}} \\
& \para.
}
\end{eqnarray}
The set of all such maps is denoted with $\Hom_{\para} (\mani , \mani')$. The subset of $\Hom_{\para} (\mani , \mani')$ of invertible maps is denoted by $\mbox{Diff}_\para (\mani, \mani').$
\end{definition}

\noindent
Families of supermanifolds and their morphisms define a category.
\begin{definition}[Category of Relative Supermanifolds] We call $\mathbf{SMan}_{/}$ the category whose objects are families of supermanifolds in the sense of Definition \ref{deffam} and whose morphisms are morphisms of families of supermanifolds in the sense of Definition \ref{defmor} and we call it the category of relative supermanifolds. In particular, we define $\mathbf{SMan}_{/\para}$ the subcategory of $\mathbf{SMan}_{/}$ whose objects are families of supermanifolds over $\para$ and whose morphisms are morphisms of supermanifolds over $\para$ and we call it the category of $\para$-relative supermanifolds.
\end{definition}

\begin{remark}[Base Change and Families of Supermanifolds]
  Base change exists in the category of relative
  supermanifolds \cite{Eder, Kessler}. In particular, given a family $\varphi \maps \mani
  \to \para$ and a morphism of supermanifolds $b \maps \para' \to
  \para $, the fiber product $b^* \mani \defeq \mani \times_{\para} \para'$
  is a family of supermanifolds over $\para'$,
\begin{eqnarray}
\xymatrix{
\mani \times_{\para} \mathpzc{S}' \ar[d]_{p_\mani} \ar[r]^{\quad p_{\mathpzc{S}'}} & \mathpzc{S}' \ar[d]^{b}\\
\mani \ar[r]^{\varphi} & \para.
}
\end{eqnarray}
We call this \define{pullback} or \define{base change} along $b \maps \para' \to \para.$ Moreover, given a morphism of families of supermanifolds $f \maps \mani \to \mathpzc{Y} $ over $\para$ as in \eqref{mapfam}, then the map $b^* f \defeq (f, id_{\para'}) \maps b^* \mani \to b^* \mathpzc{Y}$ is a map of families of supermanifolds over $\para'$, \emph{i.e.} the following diagram commutes
\begin{eqnarray}
\xymatrix@R=10pt{
\mani \times_{\para} \para' \ar@/^/[rrrr]^{(\varphi, id_{\para^{\prime}})} \ar[drr]^{p_{\para'}} \ar@/_/[ddr]_{p_{\mani}} & & & &  \mathpzc{Y} \times_{\para} \para' \ar[dll]_{p_{\para'}} \ar@/^/[ddl]^{p_\mathpzc{Y}} \\
& & \para' \ar[dd] |!{[d];[r]}\hole & & \\
& \mani \ar[dr]_{b_\mani} \ar[rr]_{\varphi \qquad }  & & \mathpzc{Y} \ar[dl]^{b_{\mathpzc{Y}}} \\
& & \para.
}
\end{eqnarray}
This says that the base change via $b \maps \para' \to \para$ defines a covariant functor $b^* \maps \mathbf{SMan}_{/\para} \to \mathbf{SMan}_{/ \para'}$.

Finally notice that the property for a family to be trivial is stable under base change, \emph{i.e.} if $ \mani_{/\para}$ is a trivial family over $\para$, then $b^* \mani_{/\para'}$ is a trivial family over $\para'$ for any $b \maps \para' \to \para$. Indeed, if $\mani_{/\para}$ is trivial, then $\mani_{/ \para} \iso \fib \times \para$ for some supermanifold $\fib$. It follows that
\begin{align}
(\fib \times \para) \times_{\para} \para' \iso \fib \times (\para \times_\para \para') \iso \fib \times \para',
\end{align}
having used associativity of the fiber product and that $\para \times_\para
\para' \iso \para' $ for any $\para'$ over $\para$. See Eder's
thesis \cite{Eder} for more details.

As advocated by Deligne and Morgan \cite{DM} and more recently by Kessler
\cite{Kessler}, we will only consider \emph{geometric constructions}, i.e.,
constructions which are compatible with \emph{any} base change.
\end{remark}

\begin{remark}[Reduced Space and Relative Supermanifolds] Whereas when
  dealing with a single supermanifold, the notion of reduced space is fundamental, 
  it must be observed that it is no longer a meaningful construction
  when dealing with families of supermanifolds. \\
  Indeed,
  let $\varphi \maps \mani \to \para$ be a family of supermanifolds, with $\para$
  a supermanifold of dimension $m|n \geq 0|1$. While is it true that
  $\varphi_{\mani_{\mathpzc{red}}} \maps \mani_{\mathpzc{red}} \to
  \para_{\mathpzc{red}}$ is a family over $\para_{\mathpzc{red}},$ the map
  $\mani_{\mathpzc{red}} \to \para$ does
  not define a family of supermanifolds over $\para$. In particular, this implies
  that for a family $\mani_{/ \para},$ the natural embedding of the reduced space
  $\iota \maps \mani_{\mathpzc{red}} \hookrightarrow \mani$ inside the supermanifold
  $\mani$ is \emph{not} a map of supermanifolds over $\para$. To make up for this,
  Kessler introduced the notion of underlying even families, see \cite{Kessler}, which
  substitutes the notion of reduced space in the case that the base supermanifold has odd
  coordinates. Due to its importance in what follows, we will recall the definition here. 
\end{remark}

\begin{definition}[Underlying even manifolds] \label{def:uem} Let $\mani_{/\para} $ be a family of supermanifolds of relative dimension $m|n$ over $\para$. An \define{underlying even manifold} for the family $\mani_{/\para}$ is a pair $(\mani^{\mathpzc{ev}}_{/\para}, \iota_{\mathpzc{ev}})$ where $\mani_{/ \para}^\mathpzc{ev} $ is a family of supermanifolds of relative dimension $m|0$ over $\para$ and $\iota_{\mathpzc{ev}} \maps \mani_{/ \para}^{\mathpzc{ev}} \hookrightarrow \mani_{/\para}$ is an embedding of families of supermanifolds over $\para,$ such that it is the identity on the underlying topological space.\\
Equivalently, $\iota_{\mathpzc{ev}}$ is a morphism of locally ringed spaces over $\para$ whose pullback
$\iota_{\mathpzc{ev}}^\sharp$ identifies $\mathcal O_{\mani_{\mathpzc{ev}}}$ with a subsheaf of
$\mathcal O_\mani$ that contains $\varphi^{-1}(\mathcal O_\para)$ and has no odd generators in the fibers.
\end{definition}

\begin{example}[Underlying even submanifold for $\R^{m|n}_{/\para}$]

  The basic example is given considering the trivial family
  $\R^{m|n} \times \para$. One needs to define an embedding
  $\iota_{\mathpzc{ev}} \maps \R^{m|0} \times \para \to \R^{m|n} \times \para$.
  To this purpose, let $y^a$ be the standard coordinates on $\R^{m|0}$, let
  $X^A = (x^a | \theta^\alpha)$ be the standard coordinates on $\R^{m|n}$ and
  $L^C = (\ell^c | \lambda^\gamma)$ some local coordinates on $\para$. Then an embedding over $\para$ is specified by its pullback on coordinate functions
  \begin{align}
    & \iota_{\mathpzc{ev}}^\sharp x^a = f^a (y, \ell | \lambda) = y^a + \sum_{\underline \gamma} \lambda^{\underline \gamma} f_{\underline \gamma}^a (y, \ell), \\
    & \iota_{\mathpzc{ev}}^\sharp \theta^\alpha = f^\alpha (y, \ell | \lambda).
  \end{align}
  It follows that, provided that $\para \neq \para_{\mathpzc{red}}$, there are
  \emph{infinitely many} underlying even manifolds for a family $\mani_{/ \para}$,
  whilst over a point $\R^{0|0}$ there is just one (canonical) reduced space
  $\mani_{\mathpzc{red}}$. 
\end{example}
On the other hand, let us define a change of coordinates $\xi \in \mbox{Diff}_\para (\R^{m|0}\times \para) $ on the underlying even manifolds $\R^{m|0} \times \para$ (source reparametrization) by declaring the new even coordinates $\tilde y^a$ on $\mathbb R^{m|0}\times \para$ to be
\begin{eqnarray}
\tilde y^a \ \defeq\ f^a(y,\ell\,|\,\lambda).
\end{eqnarray}
Further, define a diffeomorphism $\Xi \in \mbox{Diff}_{\para} (\R^{m|n} \times \para) $ on $\R^{m|n} \times \para$ (target reparametrization) by keeping $x^a$ fixed and shifting the odd coordinates:
\begin{eqnarray}
\tilde x^a \defeq x^a, \qquad \tilde \theta^\alpha \defeq - f^\alpha (x, \ell | \lambda) + \theta^\alpha.
\end{eqnarray}
Then one checks directly on generators that
\begin{align}
(\Xi\circ \iota_{\mathpzc{ev}} \circ \xi)^\sharp(\tilde x^a)=\tilde y^a,\qquad
(\Xi\circ \iota_{\mathpzc{ev}} \circ \xi)^\sharp(\tilde\theta^\alpha)=0.
\end{align}
We will call this composite map the \define{canonical embedding} of the underlying even family, and we will denote it by $\iota_{\can} \maps \mathbb{R}^{m|0}_{/\para} \to \mathbb{R}^{m|n}_{/\para}$.
\begin{eqnarray}
\xymatrix{
\R^{m|0} \times \para \ar@/_1.5pc/[rrr]^{\iota_{\can}} \ar[r]^{\xi } & \R^{m|0} \times \para \ar[r]^{\iota_{\mathpzc{ev}}} & \R^{m|n} \times \para \ar[r]^{\Xi} & \R^{m|n} \times \para.
}
\end{eqnarray}
It is characterized by the conditions
\begin{eqnarray} \label{cancoord}
\iota^\sharp_{\can} x^a = y^a, \qquad \iota^\sharp_{\can} \theta^\alpha = 0,
\end{eqnarray}
More in general, we give the following definition.
\begin{definition}[Canonical even manifold of a trivial family]\label{def:ican-trivial}
Let \(
\varphi  \maps \mathpzc{F}\times \para \to S.
\)
be a trivial family of supermanifolds. Fix a splitting of $\fib$, \emph{i.e.}\ an identification
$\mathcal O_\fib \cong \wedge^\bullet \mathcal{E}^\vee$ for some vector bundle $\mathcal{E}\to \fib_{\red}$.
Let
\(
\mathrm{aug}_\fib \maps O_\fib \longrightarrow \mathcal O_{\fib_{\red}}
\)
be the corresponding augmentation (projection to degree $0$).
The \define{canonical even manifold} of the trivial family is the pair $(\fib_{\mathpzc{red}} \times \para_{/\para}, \iota_{\can})$, where the closed embedding over $\para$
\begin{equation}
\iota_{\can}:\;\fib_{\red}\times \para \hookrightarrow \fib\times \para
\end{equation}
defined on the underlying topological spaces by the identity and on structure sheaves by
\[
\iota_{\can}^\sharp \defeq \mathrm{aug}_\fib \boxtimes \mathrm{id}_{\mathcal O_\para}:
\ \mathcal O_\fib \boxtimes \mathcal O_\para \longrightarrow
\mathcal O_{\fib_{\red}}\boxtimes \mathcal O_\para.
\]
Equivalently, $\iota_{\can}$ is the closed embedding cut out by the ideal sheaf generated by the odd
functions on the fiber $\fib$ (\emph{i.e.}\ the kernel of $\mathrm{aug}_\fib$), tensored with $\mathcal O_\para$.
\end{definition}

For a non-trivial family $\varphi:\mani\to \para$ there is, in general, no intrinsically preferred
choice of an underlying even manifold: it is a choice, and different choices are related by relative diffeomorphisms over $\para$
(\emph{i.e.}\ they form a $\Diff_\para$-torsor).

\section{Differential and integral forms on families of supermanifolds}
\label{sec:forms}

\noindent
We now turn to the definition of differential and integral forms and their complexes
for families of supermanifolds\footnote{Though, the given definitions will make sense in any geometric category.}.
To begin, for a given supermanifold $\mani$, we define the tangent sheaf of $\mani$ as the sheaf of derivations of the structure sheaf $\Oh_X$
\begin{equation}
  \T_\mani \defeq \Der_{\mathbb{R}}(\Oh_\mani, \Oh_\mani) .
\end{equation}
As usual, a global section of this sheaf is called a \define{vector field} on
$\mani$. Given a family of supermanifolds $\varphi \maps \mani \to \para$, one has a surjective morphism of sheaves of
$\mathcal{O}_{\mani}$-modules $d \varphi \maps \mathcal{T}_{\mani} \to \varphi^*
\mathcal{T}_{\para}$. This paves the way for our next definition.

\begin{definition}[Relative Tangent Sheaf]
  Let $\varphi \maps \mani \to \para$ be a family of supermanifolds. We define the \define{relative tangent sheaf}
  of the family $\mani_{/\para}$ as 
    \begin{equation}
  \mathcal{T}_{\mani / \para } \defeq \Der_{\varphi^{-1}(\mathcal O_\para)}(\mathcal O_\mani,\mathcal O_\mani).
  \end{equation} 
  A section of $\T_{\mani / \para}$ is called a \define{vertical vector field} or simply a
  \define{relative vector field} on the family.
\end{definition}
\noindent The relative tangent sheaf $\mathcal{T}_{\mani / \para}$ is a sheaf of $\mathcal{O}_{\mani}$-modules.
For every open $U\subseteq \mani$ one has
\begin{equation}
\mathcal{T}_{\mani/ \para}(U)=\{\,D\in \Der_{\mathbb{R}}(\mathcal O_\mani(U),\mathcal O_\mani(U)) \mid
D(\varphi^\sharp(f))=0\ \text{for all } f\in \mathcal O_\para(V),\ \varphi(U)\subseteq V\,\}.
\end{equation}

\noindent The previous algebraic definition has a geometric counterpart. Indeed, the relative tangent sheaf is the kernel of the
  surjective morphism of sheaves $d\varphi \maps \mathcal{T}_{\mani} \to \varphi^*
  \mathcal{T}_{\para}$. In particular, there is a short exact sequence of sheaves of
  $\mathcal{O}_\mani$-modules as follows
  \begin{eqnarray} \label{rbes} \xymatrix{ 0 \ar[r] & \mathcal{T}_{\mani / \para}
    \ar[r] & \mathcal{T}_{\mani} \ar[r] & \varphi^* \mathcal{T}_{\para} \ar[r] & 0. }
  \end{eqnarray}

\noindent If $\varphi:\mani \to \para$ is family with typical fiber $\mathpzc{F}$, then on any open
$B\subseteq |\para|$ admitting a trivialization $\Phi  \maps \varphi^{-1}(B)\xrightarrow{\sim}\mathpzc{F}\times B$
(over $B$) we may rewrite the local splitting as
\begin{equation}
\mathcal{T}_{\mani / \para}\big|_{\varphi^{-1}(B)}
\;\cong\;
\Phi^*\!\big(\mathrm{pr}_{\mathpzc{F}}^{*}\mathcal{T}_{\mathpzc{F}}\big)
\;=\;
\Phi^*\!\big(\mathcal{T}_{\mathpzc{F}}\boxtimes \mathcal O_B\big).
\end{equation}
In the following, for the sake of notationn, we will slightly abuse notation
and simply write
\begin{equation}
\mathcal{T}_{\mani/ \para}\big|_{\varphi^{-1}(B)} \;\cong\; T_{\fib}\boxtimes \mathcal O_B.
\end{equation}
This isomorphism is the one induced by the trivialization $\Phi$; different choices of $\Phi$ are related on overlaps
by the natural action of the transition maps in $\Diff(\fib)$ on $\mathcal{T}_{\fib}$.

  Accordingly, relative vector fields can be thought of as fields along the fibers of $\varphi \maps \mani
  \to \para$, which are parametrized by the variables which are
  ``transverse'' to the fibers:
  with abuse of notation, forgetting about the
  pull-backs, we will write for a local section of $\mathcal{T}_{\mani / \para}$
  \begin{eqnarray}
    X = \sum_{a = 1}^m f^a( X^A, s^L) \frac{\partial}{\partial x^a} + \sum_{\alpha = 1}^n g^\alpha (X^A, s^L) \frac{\partial}{\partial \theta^\alpha},
  \end{eqnarray}
  where $f^a, g^\alpha \in \mathcal{O}_{\mani} |_{\varphi^{-1} (B)}$, and where $X^A = x^a
  |\theta^\alpha$ are relative local coordinates and $s^L$ are (even and odd) local
  coordinates on the base $\para.$

Dualizing the tangent sheaves of $\mani $ and $\para$ one obtains the sheaves of 1-forms on $\mani$ and
$\para$, respectively, that we call $\Omega^1_\mani =
\HOM_{\Oh_\mani}(\T_\mani, \Oh_\mani)$ and $\Omega^1_{\para} =
\HOM_{\Oh_\para}(\T_\para, \Oh_\para)$. The sheaf of relative 1-forms,
can be introduced similarly to the relative tangent bundle.
\begin{definition}[Sheaf of Relative 1-Forms]
  Let $\varphi \maps \mani \to \para$ be a family of supermanifolds and let
  $\mathcal{T}_{\mani / \para}$ its relative tangent sheaf.
  The \define{relative cotangent sheaf} of the family $\mani_{/\para}$ is the dual
  $\Oh_\mani$-module, $\Omega^1_{\mani / \para} \defeq
  \mathcal{H}om_{\mathcal{O}_\mani} (\mathcal{T}_{\mani / \para},
  \mathcal{O}_\mani)$.
\end{definition}

\noindent
Alternatively, the geometric avatar of the previous definition comes by defining the relative cotangent sheaf as the cokernel sheaf of the map $d\varphi^\vee \maps \varphi^* \Omega^1_{\para} \to \Omega^1_{\mani}$, so that one has the dual of the short exact sequence \eqref{rbes},
\begin{eqnarray} \label{relses}
\xymatrix{
0 \ar[r] & \varphi^* \Omega^1_{\para} \ar[r] & \Omega^1_{\mani} \ar[r] & \Omega^1_{\mani / \para} \ar[r] & 0.
}
\end{eqnarray}
The sheaf $\Omega^1_{\mani / \para}$ is a locally-free sheaf of relative rank $m|n$ if the
  family is of relative dimension $m|n$. Locally, on an open set $\varphi^{-1}(B) \iso \fib \times B$ one has
  \begin{eqnarray}
    \label{split_omega1}
  \Omega^1_{\mani / \para} |_{\varphi^{-1}(B)} \iso  \Omega^1_\mathpzc{F} \boxtimes \mathcal{O}_B.
  \end{eqnarray}
  As with vector fields, equation \eqref{split_omega1} implies that relative 1-forms can be thought of as 1-forms along the fibers of $\varphi \maps \mani \to \para$, parametrized by the variables which are transverse to the fibers, and we will write a local section of $\Omega^1_{\mani / \para}$ as follows
\begin{eqnarray}
\omega = \sum_{a = 1}^m dx^a f_a (X^A , s^L) +  \sum_{\alpha = 1}^n d\theta^\alpha g_\alpha(X^A, s^L),
\end{eqnarray}
where again $f_a$ and $g_\alpha$ are local sections of $\mathcal{O}_{\mani }|_{\varphi^{-1}(B)}$, with $X^A = x^a | \theta^a$ relative local coordinates for $\mani$ and $s^A$ local coordinates for $\para.$

Now that we have $\mathcal{T}_{\mani / \para} $ and $\Omega^1_{\mani / \para}$ in hand, we can define a duality pairing between forms and fields as
\begin{eqnarray}
  \label{eq:pairing}
\xymatrix{
( - , - ) \maps \mathcal{T}_{\mani / \para} \otimes_{\mathcal{O}_{\mani}} \Omega^1_{\mani / \para} \ar[r] & \mathcal{O}_{\mani},
}
\end{eqnarray}
where one has $(fX, g\omega) = (-1)^{|g||X|} fg (X, \omega)$, for $f,g \in \mathcal{O}_{\mani}.$

\begin{remark}[Differential forms and polyvector fields]
  \label{rem:signs}
  Given a sheaf $\mathcal{O}_\mani$-modules $\mathcal{E},$ in the following we will
  often consider its sheaves of exterior and symmetric (super)algebras, which will be
  denoted by the symbols $\bigwedge^\bullet \mathcal{E}$ and $\bigvee^\bullet
  \mathcal{E}$ for any $p \geq 0$. In the specific case $\mathcal{E}$ is the sheaf of
  relative 1-forms or vector fields, we write
  \begin{eqnarray}
    \textstyle \Omega^p_{\mani /\para } \defeq \bigwedge^p \Omega^1_{\mani / \para}, \qquad \PV^{-p}_{\mani / \para} \defeq \bigwedge^p \mathcal{T}_{\mani / \para},
  \end{eqnarray}
  for any $p \geq 0$, and we call their sections \define{$p$-forms} and
  \define{$p$-vector fields}, respectively. More generally, a section of
  $\PV^\bullet_{\mani / \para}$ is called a \define{polyvector field}.

  Both $\Omega^\bullet_{\mani / \para}$ and $\PV^\bullet_{\mani / \para}$ are sheaves
  of $\mathcal{O}_{\mani}$-modules and carry a $\mathbb{Z} \times \mathbb{Z}_2$
  grading. We call the $\Z$-grading the \define{cohomological degree} and the
  $\Z_2$-grading the \define{parity}. Both degrees contribute to the sign rule
  independently, as advocated by Deligne and Morgan \cite[p. 62, Point of View
  I]{DM}. For example, if $\alpha$ and $\beta$ are differential forms of degree $p$
  and $q$ respectively, then we have
  \begin{equation}
    \alpha \wedge \beta = (-1)^{pq + |\alpha| |\beta|} \beta \wedge \alpha ,
  \end{equation}
  where $|\alpha|, |\beta| \in \Z_2$ are the parities, independent of the
  cohomological degrees $p$ and $q$. Moreover, in relative local coordinates, coordinate vector fields and
  1-forms are assigned the following degrees
  $\deg_{\mathbb{Z} \times \mathbb{Z}_2} \cdot = ( \deg_{\mathbb{Z}} \cdot, \deg_{\mathbb{Z}_2} \cdot )$
  in the $\mathbb{Z} \times \mathbb{Z}_2$ grading:
  \begin{eqnarray}
    \begin{tabular}{lrrrrrr}
      \toprule
      &  $x^a$ & $\theta^\alpha$  & $\partial_{x^a}$  & $\partial_{\theta^\alpha}$ & $dx^a$ & $d\theta^\alpha$ \\
      \midrule
      $\deg_{\mathbb{Z} \times \mathbb{Z}_2}$ & $(0,0)$ & $(0,1)$ & $(-1,0)$ & $(-1,1)$ & $(1,0)$ & $(1,1)$ \\
      \bottomrule
    \end{tabular}
  \end{eqnarray}
  This choice leads to the following commutation relations
  \begin{align}
    & dx^a \wedge dx^b = - dx^b \wedge dx^a, \qquad d\theta^\alpha \wedge d\theta^\beta = d\theta^\beta \wedge d\theta^\alpha, \qquad dx^a \wedge d\theta^\alpha = - d \theta^\alpha \wedge dx^a,
  \end{align}
  and similarly for vector fields. It follows that if $\mani_{/ \para}$ is of
  relative dimension $m|n$ with $n\geq 1$, the sheaves of
  $\mathcal{O}_{\mani}$-modules $\Omega^p_{\mani / \para}$ and
  $\PV^p_{\mani / \para}$ are non-zero for all $p\geq 0$. Also, notice that using
  this convention, the pairing \eqref{eq:pairing} is indeed of degree
  0.\footnote{There is a different definition due to Bernstein and Leites for the
    exterior algebra of a supermodule or a sheaf of modules, see \cite{BerLei}. In
    particular, they defined
    $\wedge^\bullet \mathcal{E} \defeq \vee^\bullet \Pi \mathcal{E}$, where $\Pi$ is
    the parity changing functor. The two definitions, the one given by Deligne that
    we employed in the present paper and the one by Bernstein and Leites introduced
    above, are equivalent - in the sense that there is an equivalence of categories
    and an isomorphism of functors, as shown by Deligne--Morgan \cite[p. 64]{DM}.
    Nonetheless, the attentive reader might see that they lead to different sign
    conventions, therefore a great deal of attention must be paid.}
\end{remark}

\begin{remark}[Exterior Derivative] The \define{relative exterior derivative} $d \maps \mathcal{O}_\mani \to \Omega^1_{\mani / \para}$ is defined given an action on relative vector fields. In particular, using the pairing \eqref{eq:pairing} above, one defines
\begin{eqnarray}
( X, d_{\mani / \para} f ) \defeq X(f),
\end{eqnarray}
where $f$ is a section of $\mathcal{O}_\mani$ and $X$ is a section of $\mathcal{T}_{\mani / \para}.$ In relative local coordinates, $X^A \defeq x^a | \theta^\alpha$, the exterior derivative has the following expression
\begin{eqnarray}
d_{\mani / \para} f = \sum_A \left (dx^A \otimes \frac{\partial}{\partial_{X^A}} \right ) f = \sum_{a = 1}^m dx^a \frac{\partial f (X, s)}{\partial x^a} + \sum_{\alpha = 1}^n d\theta^\alpha \frac{\partial f (X, s)}{\partial \theta^\alpha}.
\end{eqnarray}
According to the above definitions, $d_{\mani / \para}$ is of
degree $(1,0)$ in the $\mathbb{Z}\times \mathbb{Z}_2$ grading. The relative exterior derivative can be continued to a $\varphi^{-1}(\mathcal{O}_\para)$-linear derivation $d_{\mani / \para} \maps \Omega^{k}_{\mani / \para} \to \Omega^{k+1}_{\mani / \para},$ which is easily checked to square to zero.  This allows us to give the following definition.
\end{remark}

\begin{definition}[Relative de Rham Complex] Let $\mani_{/\para}$ be a family of supermanifolds. We call the pair $(\Omega^\bullet_{\mani / \para}, d_{\mani / \para})$ the \define{relative de Rham complex} of the family $\mani_{ / \para}$ or de Rham complex of $\mani_{/ \para}$ for short,
\begin{eqnarray} \xymatrix{
0 \ar[r] & \mathcal{O}_{\mani} \ar[r]^{d_{\mani / \para}} & \Omega^1_{\mani / \para} \ar[r]^{d_{\mani / \para}} & \Omega^2_{\mani / \para} \ar[r]^{d_{\mani / \para}} & \ldots \ar[r]^{d_{\mani / \para}} & \Omega^p_{\mani / \para} \ar[r]^{d_{\mani / \para}} & \ldots .
}
\end{eqnarray}
\end{definition}

\noindent
The pair $(\Omega^\bullet_{\mani / \para}, d_{\mani / \para})$ defines a sheaf of
differential graded superalgebras (DGsA).

\begin{remark}[Berezinian Sheaf and Integral Forms]
  \label{rem:ber}
  As observed above, the (relative) de Rham complex is \emph{not} bounded from above.
  In particular, there is no tensor density to integrate over in the de Rham complex
  of a supermanifold. To define a meaningful integration theory, one needs to
  introduce a second complex, which is -- in a sense that we will made precise shortly
  -- dual to the de Rham complex. This new complex is widely known as the complex of
  \emph{integral forms}. 
    Its construction is based on one of the landmark notions in
  supergeometry: the Berezinian sheaf, which serves as the analogue of the
  determinant sheaf in ordinary geometry. The literature features several related
  constructions of the Berezinian sheaf as a natural sheaf on a supermanifold, see
  for example the article of Hern\'andez--Mu\~noz \cite{HM}, which dates
  back to the early days of supergeometry, as well as more recent work \cite{BHP, DM,
    Kessler, NojaRe}. In the case of families of supermanifolds, the Berezinian
  arises from the sheaf of relative 1-forms.
\end{remark}

\begin{definition}[Relative Berezinian Sheaf] Let $\mani_{ / \para}$ be a family of supermanifolds of relative dimension $m|n$. We define the \define{relative Berezinian sheaf} $\bersheaf_{\mani / \para}$ of $\mani_{/ \para}$ as
\begin{eqnarray}
\bersheaf_{\mani / \para} \defeq \mathcal{E}xt^m_{\vee^\bullet (\Omega^1_{\mani / \para} )^\vee } (\Omega^0_{\mani /\para}, \vee^\bullet (\Omega^1_{\mani / \para} )^\vee ).
\end{eqnarray}
\end{definition}

\noindent
The relative Berezinian sheaf is a locally-free sheaf of $\Oh_\mani$-modules of rank $\delta_{n \mbox{\scriptsize{mod}2},0} | \delta_{n\mbox{\scriptsize{mod}2},1}$ \cite{Noja, NojaRe}. A representative of the above cohomology class in a relative coordinate chart of relative dimension $m|n$ is given by $\ber_{\mani / \para} = [\partial_{\theta^1} \ldots \partial_{\theta^n} \otimes dx^1\ldots dx^m].$ Notice that the parity of $\bersheaf_{\mani / \para}$ is only related to the derivations $\partial_{\theta}$'s in the representative of $\ber.$

On the other hand, we assign $\bersheaf_{\mani / \para}$ cohomological degree $m$, so that one has $\deg_{\mathbb{Z} \times \mathbb{Z}_2} \ber_{\mani / \para} = (m, n\, \mbox{mod} 2)$. This is in agreement with the convention used in \cite{DM} -- also, given the definition via $\mathcal{E}{xt}$-group, the assignment is natural from a derived point of view.

The Berezinian carries a natural \emph{right} action of vector fields in $\mathcal{T}_{\mani / \para}$, via its Lie derivative \cite{Manin, Noja, Penkov}. Given any vector fields $X \in \mathcal{T}_{\mani / \para}$, we will denote this action as
\begin{align}
\ber_{\mani / \para} \cdot X.
\end{align} In particular, this action makes the Berezinian sheaf into a sheaf of \emph{right} $\mathcal{D}_{\mani / \para}$-modules\footnote{Analogously, up to a sign, the Lie derivative endows the Berezinian with a flat right connection.}. Building on $\bersheaf_{\mani / \para} $, we give the following definition.

\begin{definition} Let $\mani_{ / \para}$ be a family of supermanifolds of relative
  dimension $m|n$. We define \define{relative integral $p$-forms} on $\mani_{/\para}$
  to be the sections of the tensor product of sheaves given by
\begin{eqnarray}
\bersheaf^{m-p}_{\mani / \para} \defeq \bersheaf_{\mani / \para} \otimes_{\mathcal{O}_{\mani}} \PV^{-p}_{\mani / \para}.
\end{eqnarray}
\end{definition}

\noindent
The parity of an integral form is assigned in the obvious way, taking into account the parity of $\ber_{\mani / \para} \in \bersheaf_{\mani / \para}$. On the other hand, the cohomological degree takes into account the previous convention that $p$-vector fields sit in degree $-p$ and the Berezinian sheaf sits in degree $m$ if $\mani_{/\para}$ is of relative dimension $m|n$. In other words, one has that if $\sigma \in \bersheaf_{\mani / \para}^{m-p}$, then $\deg_{\mathbb{Z}} (\sigma) = m - p,$ for any $p \geq 0.$ In particular, observe that integral forms are of negative degree whenever $p>m.$

Relative integral forms are defined as sheaves of $\Oh_\mani$-modules. While they are obviously not sheaves of algebras due to the presence of the Berezinian sheaf in their definition (indeed, multiplying two integral forms does not in general yield an integral form), an obvious action of polyvector fields $\PV^{\bullet}_{\mani / \para}$ and, more importantly, forms $\Omega_{\mani / \para}^\bullet$ can be given on them. We will denote the action of forms in $\Omega^\bullet_{\mani / \para}$ on $\Be^\bullet_{\mani / \para}$ by writing $
\sigma \cdot \omega$
for any integral form $\sigma$ and differential form $\omega.$

A differential can be defined for relative integral forms, though its definition is not straightforward, see \cite{CNR, Manin, Noja}, and originally given only in local coordinates \cite{BerLei, Penkov}. Here we provide an intrinsic definition.  
Let $\varphi \maps \mani \to \para$ be a family of supermanifolds of relative even dimension $m$. For $\ber \in \bersheaf_{\mani / \para}$ and homogeneous vector fields $Y_1,\dots,Y_k \in \mathcal{T}_{\mani / \para}$, let us consider the integral form $\ber \otimes Y_1\wedge\cdots\wedge Y_k \in \bersheaf_{\mani / \para}^{m-k}$, and define the operator $\delta_{\mani / \para} \maps \bersheaf^{m-k}_{\mani / \para} \to \bersheaf^{m-k+1}_{\mani / \para} $ as
\begin{align}
&\delta_{\mani / \para}\bigl(\ber \otimes Y_1\wedge\cdots\wedge Y_k\bigr)
=
\sum_{i=1}^k (-1)^{i+1}\,
(-1)^{|Y_i|\,(|\ber|+\sum_{j<i}|Y_j|)}
\,(\ber \cdot Y_i)\otimes Y_1\wedge\cdots\widehat{Y}_i \cdots\wedge Y_k
\label{eq:delta} \nonumber\\
&\quad+\sum_{1\le i<j\le k} (-1)^{i+j}\,
(-1)^{|Y_i|\sum_{l<i}|Y_l| + |Y_j|\sum_{l<j,l\neq i}|Y_l|}
\,\ber\otimes [Y_i,Y_j]\wedge Y_1\wedge\cdots\widehat{Y}_i \cdots \widehat{Y}_j \cdots \wedge Y_k.
\end{align}
This definition is particularly useful, as it makes clear that there exists a Cartan calculus on integral forms, as we shall discuss in Appendix \ref{app:Cartan}. Notice that according to the given definition also $\delta_{\mani/ \para}$ is of degree $(1,0)$ in the $\mathbb{Z}\times \mathbb{Z}_2$-grading, and, crucially, a direct calculation shows that $\delta_{\mani / \para}$ squares to zero, hence justifying the following definition.

\begin{definition} Let $\mani_{/\para}$ be a family of supermanifolds over $\para$.
  We call the pair $(\bersheaf^\bullet_{\mani / \para}, \delta_{\mani / \para})$ the
  \define{relative Spencer complex} of the family $\mani_{ / \para}$ or Spencer
  complex of $\mani_{/ \para}$ for short,
\begin{eqnarray} \xymatrix{
\ldots \ar[r]^{\delta_{\mani / \para} \qquad \;\; } & \ber_{\mani / \para} \otimes_{\mathcal{O}_{\mani}} \wedge^k\mathcal{T}_{\mani / \para} \ar[r]^{\qquad \quad \delta_{\mani / \para}} & \ldots \ar[r]^{\delta_{\mani / \para}\qquad \;} & \ber_{\mani / \para} \otimes_{\mathcal{O}_{\mani }} \mathcal{T}_{\mani / \para} \ar[r]^{\qquad \delta_{\mani / \para}} & \ber_{\mani / \para} \ar[r] & 0 .
}
\end{eqnarray}
\end{definition}

\noindent
It is important to note that also $\delta_{\mani / \para}$ is $\varphi^{-1} (\mathcal{O}_\para)$-linear: notice though, that in contrast to the relative de Rham complex, the pair $(\bersheaf^\bullet_{\mani / \para}, \delta_{\mani / \para})$ does not define a sheaf of differential graded superalgebras, but a sheaf of differential graded $\mathcal{O}_\mani$-supermodules instead. \\
Further, the $\Omega^\bullet_{\mani / \para}$-action on $\Be^\bullet_{\mani / \para}$ induced by contracting polyfields with differential forms is compatible with the differential $\delta_{\mani / \para}$, more precisely $\delta_{\mani / \para}$ is a derivator with respect to the $\Omega^\bullet_{\mani / \para}$-action
\begin{equation}
\delta_{\mani / \para}(\sigma \cdot \omega) = \delta_{\mani / \para} \sigma \cdot \omega +(-1)^{|\sigma|} \sigma \cdot d_{\mani / \para} \omega, 
\end{equation} 
for any integral form $\sigma$ and differential form $\omega$. More generally, this makes the relative Spencer complex into a sheaf of differential graded $\Omega^\bullet_{\mani / \para}$-modules. 
Finally, it is worth comparing the relative de Rham and Spencer complexes from the point of view of their cohomological degree.
\begin{eqnarray}
\xymatrix@R=10pt{
& 0 \ar[r] & \Omega^0_{\mani /\para}  \ar[r]\ar@{.}[d] & \Omega^1_{\mani /\para} \ar[r]\ar@{.}[d] & \cdots \ar[r]  & \Omega^m_{\mani /\para} \ar[r]\ar@{.}[d] & \Omega^{m+1}_{\mani /\para} \ar[r] & \ldots \\
\ldots \ar[r] & \bersheaf^{-1}_{\mani /\para} \ar[r] & \bersheaf^0_{\mani /\para} \ar[r] & \bersheaf^1_{\mani /\para} \ar[r] & \cdots \ar[r] & \bersheaf^m_{\mani /\para} \ar[r] & 0.
\save "1,3"."2,6"*[F]\frm{}
\restore
}
\end{eqnarray}
In the given degree convention, where $\bersheaf_{\mani / \para}$ is in degree $m$, the framed part of the above diagram is where both the sheaves of differential and integral forms are non-zero. We will see that this is useful, since the cohomology of these complexes can only be non-zero---in fact isomorphic---inside this highlighted frame.

\section{Cohomology of differential and integral forms}
\label{sec:cohomology}

\subsection{Cohomology and Poincar\'e lemmas}

In this section, we take a step back to study the cohomology of complexes of differential and integral forms on a single supermanifold $\mani$\footnote{Of course, a single supermanifold is nothing but a ``family'' over a point, that is a collapse map $c \maps \mani \to \R^{0|0}$, so that the definitions of the (absolute) de Rham $\Omega^\bullet_\mani$ and Spencer complex $\bersheaf^\bullet_\mani$ can be easily inferred from the previous section, or just read out from the literature -- see for example \cite{Manin} and the recent \cite{Noja} or \cite{Pol1} for the reader interested in the algebraic setting}. We will then leverage on these result to study cohomology of families of supermanifolds.
\begin{definition}[de Rham and Spencer Cohomology] Let $\mani$ be a supermanifold and let $\Omega^\bullet_{\mani}$ and $\bersheaf^\bullet_\mani$ be their de Rham and Spencer complex respectively. We define \define{de Rham cohomology} of $\mani$ to be the hypercohomology of the de Rham complex $\Omega^\bullet_\mani$, and we write
\begin{eqnarray}
H_{\mathpzc{dR}}^\bullet (\mani) \defeq \mathbb{H}^\bullet (\mani, \Omega^\bullet_\mani).
\end{eqnarray}
Likewise, we define the \define{Spencer cohomology} of $\mani$ to be the hypercohomology of the Spencer complex $\bersheaf^\bullet_\mani$, and we write
\begin{eqnarray}
H_{\mathpzc{Sp}}^\bullet (\mani) \defeq \mathbb{H}^\bullet (\mani, \bersheaf^\bullet_\mani).
\end{eqnarray}
\end{definition}

In the smooth category, since $\Omega^k_\mani$ and $\Be^k_\mani$ are fine sheaves on $|\mani|$, hypercohomology 
is computed by global sections. Hence, these definitions amount to computing $H_{\mathpzc{dR}}^\bullet
(\mani) = H^\bullet (\Gamma (\mani, \Omega^\bullet )$ and $H_{\mathpzc{Sp}}^\bullet (\mani) =
H^\bullet (\Gamma (\mani, \bersheaf^\bullet))$.

As it turns out, the de Rham and Spencer cohomology of $\mani$ compute the cohomology
of the reduced space $\mani_{\mathpzc{red}}$ underlying $\mani$, \cite{CNR, DM, Noja,
  Penkov, Pol1}. More precisely, the relevant Poincar\'e lemmas and their
consequences are collected in the following theorem.

\begin{theorem}[Poincar\'e Lemmas \& Cohomology]
  \label{thm:PL}
  The Poincar\'e lemmas for de Rham and Spencer cohomology read
\begin{eqnarray}
H^k_{\mathpzc{dR}} (\R^{m|n}) \iso H^k_{\mathpzc{Sp}} (\R^{m|n})\iso  \left \{ \begin{array}{lll}
\R & &  k = 0\\
0 & & k \neq 0.
\end{array}
\right.
\end{eqnarray}
In particular, $H^0_{\mathpzc{dR}} (\R^{m|n})$ is generated by $1 \in \Omega^0(\R^{m|n})$, and $H^0_{\mathpzc{Sp}} (\R^{m|n})$ is generated by
\begin{equation}
  \label{eq:generator}
  Y \defeq \ber^\c \otimes \, \theta^1 \theta^2 \cdots \theta^n \pd{x^1} \wedge
  \pd{x^2} \wedge \cdots \wedge \pd{x^m} \in \Ber^0(\R^{m|n}),
\end{equation}
where $\ber^c \in \Ber^m(\R^{m|n})$ is the coordinate Berezinian, $\ber^c = [\partial_{\theta^1} \ldots \partial_{\theta^n} \otimes dx^1\ldots dx^m]$.

More generally, let $\mani$ be a smooth supermanifold. The de Rham and Spencer
complexes are quasi-isomorphic and they define the same object in the derived
category of sheaves of $\mathbb{R}$-modules $\mbox{\sffamily{\emph{D}}}(\mani, \mathbb{R})$,
\begin{align}
  \mbox{\sffamily{\emph{D}}}(\mani) \owns (\Omega^\bullet_{\mani}, d) \cong (\R_\mani, 0) \cong (\bersheaf^\bullet_\mani, \delta)
\end{align}
In particular, for any $k$, one has
\begin{eqnarray}
  H^k_{\mathpzc{dR}} (\mani) \iso H^k ({\mani}, \R_{\mani}) \iso H^k_{\mathpzc{Sp}} (\mani).
\end{eqnarray}
In other words, both the de Rham and Spencer complex compute the cohomology of the
underlying topological space $|\mani|$, with real coefficients.
\end{theorem}
Before we go on, it is very important to stress the role of the non-zero class in the Poincar\'e lemma for Spencer cohomology, the element $Y \in \Be^0 (\mathbb{R}^{m|n}).$ First off, it is easy to see 
that $Y$ is globally-defined, closed and not exact, on every (connected) supermanifold $\mani$ -- not only the model supermanifold $\mathbb{R}^{m|n}$ --, \emph{i.e.}\ $Y \in H^0_{\Sp} (\mani)$, see \cite{Carlo, Noja}. Pairing 
a differential form with $Y$ concretely realizes the quasi-isomorphism between the de Rham and Spencer complex. More precisely, we have the following result.
\begin{theorem}[$Y$ as a quasi-iso]\label{thm:pairing-with-Y-global}
Let $\mani$ be a smooth supermanifold of dimension $m|n$.
Let $\mathcal{Y}  \maps \Omega^\bullet_\mani \to \Be^\bullet_\mani$ be the morphism of complexes defined by pairing a differential form with $Y$,
\begin{equation}
\xymatrix@R=1.5pt{
\mathcal{Y}  \maps \Omega^\bullet_\mani \ar[r] & \Be^\bullet_\mani \\
\omega \ar@{|->}[r] & Y \cdot \omega
}
\end{equation}
Then $\mathcal{Y}$ is a cochain map, \emph{i.e.} $\delta \circ \mathcal{Y} = (-1)^{|\mathcal{Y}|} \mathcal{Y} \circ d$. In particular, it is a quasi-isomorphism.
\end{theorem}
\begin{proof}
Let $U\subseteq \mani$ be open and $\omega\in \Omega_X^k(U)$.
Since $Y$ is a global section, we may restrict it to $U$ and compute using the fact that $\delta$ is a derivator, 
\[
\delta\big(\mathcal{Y}(\omega)\big)
=\delta\big(Y\cdot \omega\big)
=\delta(Y)\cdot \omega + (-1)^{|Y|}\,Y\cdot d\omega = (-1)^{|Y|}\,Y\cdot d\omega = (-1)^{|\mathcal{Y}|} \mathcal{Y} (d\omega),
\]
having used that $\delta (Y) = 0.$
That $\mathcal{Y}$ is a quasi-isomorphism is a local statement on $\mani$. Indeed, a morphism of complexes of sheaves is a quasi-isomorphism
if and only if it is a quasi-isomorphism on all stalks. To see this, fix $x\in |\mani|$ and choose a coordinate neighborhood $U\ni x$ with $U\simeq \mathbb{R}^{m|n}$.
By restriction, $\mathcal{Y}|_U$ is the model map
\(
(\Omega^\bullet(\mathbb{R}^{m|n}),d)\rightarrow (\Be^{\bullet}(\mathbb{R}^{m|n}),\delta) 
\) still given by 
\( \omega\longmapsto Y\cdot \omega,
\)
where $Y$ restricts to the standard local representative as above (up to an overall unit).
It follows from Theorem \ref{thm:PL} that on $\mathbb{R}^{m|n}$ this map is a quasi-isomorphism: both complexes have cohomology concentrated
in degree $0$ and $\mathcal{Y}(1)=Y$ identifies the corresponding generators. Therefore $\mathcal{Y}_x$ is a quasi-isomorphism of stalk complexes for every $x$, hence $\mathcal{Y}$ is a
quasi-isomorphism of complexes of sheaves on $\mani$.
\end{proof}
Before we go on, we observe in passing that the sign in the cochain map could also be absorbed by redefining $\mathcal{Y} (\omega) = (-1)^{m|\omega|} Y \cdot \omega$ for $\omega$ a differential form of degree $k$.

\subsection{Compactly supported cohomology and Poincar\'e lemmas}
\label{sec:compactsupport}

In this section we will introduce the compactly supported variant of the ordinary de
Rham and Spencer cohomology, and we will compute the related Poincar\'e lemmas.
\begin{definition}[Compactly Supported de Rham and Spencer Cohomology] Let $\mani$ be a supermanifold and let $\Omega^\bullet_{\mani}$ and $\bersheaf^\bullet_\mani$ be their de Rham and Spencer complex respectively. We define compactly supported de Rham cohomology of $\mani$ to be the compactly supported hypercohomology of the de Rham complex $\Omega^\bullet_\mani$, and we write
\begin{eqnarray}
H_{\mathpzc{dR}, \mathpzc{c}}^\bullet (\mani) \defeq \mathbb{H}^\bullet_{\mathpzc{c}} (X, \Omega^\bullet_\mani).
\end{eqnarray}
Likewise, we define the compactly supported Spencer cohomology of $\mani$ to be the compactly supported hypercohomology of the Spencer complex $\bersheaf^\bullet_\mani$, and we write
\begin{eqnarray}
H_{\mathpzc{Sp}, \mathpzc{c}}^\bullet (\mani) \defeq \mathbb{H}^\bullet_{\mathpzc{c}} (\mani, \bersheaf^\bullet_\mani).
\end{eqnarray}
\end{definition}

\noindent
These amount to computing the cohomologies $H_{\mathpzc{dR}, \mathpzc{c}}^\bullet (\mani) = H^k (\Gamma_{\mathpzc{c}} (\mani, \Omega^\bullet_{\mani}))$, and $H_{\mathpzc{Sp}, \mathpzc{c}}^\bullet (\mani) = H^k (\Gamma_{\mathpzc{c}} (\mani, \bersheaf^\bullet_{\mani} ))$.

As we discussed, the Poincar\'e lemmas for ordinary de Rham and Spencer cohomology are well-established in the literature, the same is not true---to our best knowledge---for their compactly supported counterparts. We remedy this here, by leveraging on our previous result, namely Theorem \ref{thm:pairing-with-Y-global}. Indeed, one can put to good use the quasi-isomorphism $\mathcal{Y}$ upon observing the following.
\begin{lemma}\label{prop:compact-support}
Let $X$ be a smooth supermanifold. The cochain map
$\mathcal{Y}:\Omega_\mani^\bullet\to \Be_\mani^{\bullet}$ given by $\mathcal{Y}(\omega)=Y\cdot\omega$
induces a quasi-isomorphism on compactly supported sections:
\[
\Gamma_\mathpzc{c}(\mani,\Omega_\mani^\bullet)\xrightarrow{\ \simeq\ }\Gamma_{\mathpzc{c}}(\mani,\Be_\mani^{\bullet}).
\]
Consequently,
\[
H^\bullet_{\mathpzc{dR,c}}(\mani)\cong H^\bullet_{\mathpzc{c}} (\mani , \mathbb{R}_\mani) \cong H^\bullet_{\mathpzc{Sp,c}}(\mani).
\]
\end{lemma}

\begin{proof}
Both $\Omega_\mani^k$ and $\Be_\mani^{k}$ are fine sheaves on $|\mani|$. Since fine sheaves are $\mathpzc{c}$-soft, the functor $\Gamma_\mathpzc{c}(\mani,-)$ is exact, hence it preserves quasi-isomorphism
between complexes of fine sheaves. In particular, applying $\Gamma_\mathpzc{c}(\mani,-)$ to the quasi-isomorphism $\mathcal{Y}$ yields a quasi-isomorphism of complexes.
\end{proof}
This allows us to prove the following theorem.

\begin{theorem}[Compactly Supported Poincar\'e Lemmas \& Cohomology] \label{thm:compsupp}
  The Poincar\'e lemmas for the compactly supported de Rham and Spencer cohomology
  read
  \begin{eqnarray}
    H^k_{\mathpzc{dR}, \mathpzc{c}} (\R^{m|n}) \iso H^k_{\mathpzc{Sp}, \mathpzc{c}} (\R^{m|n})\iso  \left \{ \begin{array}{lll}
      \R & &  k = m\\
      0 & & k \neq m.
    \end{array}
    \right.
  \end{eqnarray}
  In particular, $H^m_{\mathpzc{dR}, \mathpzc{c}} (\R^{m|n})$ is generated by
  $\mathpzc{B}_m dx_1 \ldots dx_m \in \Omega^m_\mathpzc{c}(\R^{m|n})$, where $\mathpzc{B}_m$ is a bump function that integrates to 1 on $\R^m$ and
  $H^m_{\mathpzc{Sp}, \mathpzc{c}} (\R^{m|n})$ is generated by
  \begin{equation}
  \ber^c \otimes \, \theta_1 \ldots \theta_n\mathpzc{B}_m \in \Ber^m_\mathpzc{c}(\R^{m|n}),
  \end{equation}
  where $\ber^c \in \Ber^m(\R^{m|n})$ is the coordinate Berezinian,
  $\ber^c = [\partial_{\theta^1} \ldots \partial_{\theta^n} \otimes dx^1\ldots dx^m]$.

  More generally, let $\mani$ be a smooth supermanifold. Then, for any $k$, one has
  \begin{eqnarray}
    H^k_{\mathpzc{dR}, \mathpzc{c}} (\mani) \iso H^k_\mathpzc{c} (|{\mani}|, \R_\mani) \iso H^k_{\mathpzc{Sp}, \mathpzc{c}} (\mani).
  \end{eqnarray}
  In other words, the compactly supported de Rham and Spencer complexes both compute
  the compactly supported cohomology of the underlying topological space $|\mani|$
  with real coefficients.
\end{theorem}
\begin{proof}
The compactly supported Poincar\'e lemma for differential forms can be obtained
exactly as in the classic textbook of Bott--Tu \cite{BottTu} upon paying some extra
care with signs: namely, on $\R^{m|n}$, the compactly supported cohomology agrees with that of $\R^m$, and the representative mirrors the
classical one. Next, Lemma \ref{prop:compact-support} yields the results for the compactly supported cohomology for integral forms. In particular, the only non-zero representative 
in $H^m_{\mathpzc{Sp,c}} (\mani)$ is obtained by applying $\mathcal{Y}$ to the only non-zero representative in $H^m_{\mathpzc{dR,c}} (\mani)$. Finally, the last part follows again from Lemma \ref{prop:compact-support}.
 \end{proof}
We note that the compactly supported Poincar\'e lemma proved above could also have been established without using the quasi-isomorphism $\mathcal{Y}$, by means of a direct computation exhibiting an explicit homotopy for the complex. Since this calculation is interesting and nontrivial, we present it in Appendix \ref{app:csintegral}.

Finally, before we move back to families of supermanifolds we stress that the representative $\ber^c \otimes \, \theta_1 \ldots \theta_n\mathpzc{B}_m \in H^m_{\mathpzc{Sp,c}} (\mani)$ is the key to define the notion of Berezin integral, as we shall see.



\section{Families of supermanifolds and cohomology}
\label{sec:families}

In this section we go back to families of supermanifolds $\varphi \maps \mani \rightarrow \para$. Namely, we will be concerned with relative differential and integral forms together with their cohomology. 
\begin{definition}[Relative de Rham and Spencer Cohomology] \label{relcoho} Let $\varphi\maps\mani \to \para$ be a family of supermanifolds and let $\Omega^\bullet_{\mani / \para}$ and $\bersheaf^\bullet_{\mani / \para}$ be the de relative Rham and Spencer complex respectively. We define the \define{relative de Rham cohomology} of the family $\varphi \maps \mani \rightarrow \para$ to be the hypercohomology of $\mathbf{R}\varphi_* \Omega^\bullet_{\mani / \para}$, and we write
\begin{eqnarray}
{H}_{\mathpzc{dR}}^\bullet (\mani{/\para}) \defeq \mathbb{H}^\bullet (\mani , \Omega^\bullet_{\mani / \para} ). 
\end{eqnarray}
Likewise, we define the \define{relative Spencer cohomology} of the family $\varphi \maps \mani \rightarrow \para$ to be the hypercohomology of the complex $\mathbf{R} \varphi_* \bersheaf^\bullet_{\mani / \para}$, and we write
\begin{eqnarray}
{H}_{\mathpzc{Sp}}^\bullet (\mani{/ \para}) \defeq \mathbb{H}^\bullet (\mani , \Be^\bullet_{\mani / \para} ). 
\end{eqnarray}
The same definitions can be given for the compactly supported case, simply considering compactly supported hypercohomology. Accordingly, we will denote the compactly supported relative de Rham and Spencer cohomology by ${H}_{\mathpzc{dR}, \mathpzc{c}}^\bullet (\mani{/\para})$ and ${H}_{\mathpzc{Sp}, \mathpzc{c}}^\bullet (\mani{/\para})$ respectively\footnote{These reduce to the usual definitions when $\para$ is a point.}.
\end{definition}
With reference to the previous definition, we observe that differentials $d_{\mani / \para}$ and $\delta_{\mani / \para}$ are $\varphi^{-1}(\mathcal{O}_\para)$-linear, and hence both the relative de Rham and relative Spencer complex define objects in the related derived category, \emph{i.e.}
\begin{equation}
\Omega^\bullet_{\mani / \para} \in \D^+({\varphi}^{-1} (\mathcal{O}_{\para})) \quad \mbox{and} \quad \Be^\bullet_{\mani / \para} \in \D^+({\varphi}^{-1} (\mathcal{O}_{\para})),
  \end{equation}
  where by $\D^+ (\varphi^{-1} (\mathcal{O}_S))$ we denote the derived category of (bounded below) $\varphi^{-1} ( \mathcal{O}_\para)$-modules and hence (note that this implies \emph{a fortiori} that they are also objects in the derived category of Abelian sheaves $\D^+ (\mani, \mathbb{R})$ on $\mani$ -- hence the previous definition Definition \ref{relcoho} is well-given).\\
  Applying the (derived) direct image $\mathbf{R}\varphi_{*} $ functor, we get objects in the derived category of $\mathcal{O}_\para$-modules, \emph{i.e.}
  \begin{align}
\mathbf{R}\varphi_* \Omega^\bullet_{\mani / \para} \in \D^+(\mathcal{O}_{\para}) \quad &\mbox{and} \quad \mathbf{R}\varphi_{*}\Be^\bullet_{\mani / \para} \in \D^+(\mathcal{O}_{\para}), 
  \end{align}    
and similarly applying the (derived) exceptional direct image (or direct image with compact support) $\mathbf{R}\varphi_{!}$ functor. Taking hypercohomology, one has the following identification in the derived category of $\mathbb{R}$-vector spaces $\D^+(\mathbb{R})$
\begin{equation}
\mathbb{H}^\bullet (\para, \mathbf{R} \varphi_* \bersheaf^\bullet_{\mani / \para}) \cong  \mathbb{H}^\bullet (\mani , \Omega^\bullet_{\mani / \para} ),  \qquad \mathbb{H}^\bullet (\para, \mathbf{R} \varphi_* \bersheaf^\bullet_{\mani / \para}) = \mathbb{H}^\bullet (\mani , \Be^\bullet_{\mani / \para}),
\end{equation}
and similarly for the compactly supported case. Note that since all of the terms of $\Omega^\bullet_{\mani/\para}$ and $\Be^\bullet_{\mani /\para}$ are fine sheaves in
the smooth category, these hypercohomology groups are computed by global sections, \emph{i.e.}
\begin{equation}
\mathbb{H}^\bullet (\mani , \Omega^\bullet_{\mani / \para} ) = H^\bullet (\Gamma (\mani, \Omega^\bullet_{\mani / \para})), \qquad \mathbb{H}^\bullet (\mani , \Be^\bullet_{\mani / \para} ) = H^\bullet (\Gamma (\mani, \Be^\bullet_{\mani / \para})),\end{equation}
and similarly for the compactly supported case.

Our first step will be to prove a relative version of the Poincar\'e lemma.  
 \begin{lemma}[Poincar\'e lemma for families of supermanifolds] \label{rel-poinc}
  Let $\varphi  \maps \mani \to \para$ be a family of supermanifolds. The inclusion
  $    \varphi^{-1}(\Oh_\para) \inclusion \Omega^\bullet_{\mani/\para}
  $  is a quasi-isomorphism. \\
  Hence, there is a
  canonical isomorphisms in $\mbox{\sffamily{\emph{D}}} (\varphi^{-1} (\mathcal{O}_\para))$
  \begin{equation}
\Omega^\bullet_{\mani
  / \para}\iso  \varphi^{-1} (\mathcal{O}_\para) \iso  \Be^\bullet_{\mani / \para}.
  \end{equation}
  \end{lemma}
\begin{proof}
Since $\varphi:\mani \to \para$ is a (locally trivial) submersion, for every point $x\in \mani$ there exists an
open neighborhood $U\ni x$ and an open set $B\subseteq S$ with $\varphi(U)\subseteq B$ such that
$\varphi|_U$ trivializes as the projection
\[
U \;\cong\; F\times B \xrightarrow{\ \mathrm{pr}_2\ } B,
\]
where $W$ is an open domain in the fiber $F \subset \fib$. Shrinking $F$ if necessary, we may assume that $F$ is
contractible. On such a trivializing open set, with obvious notation, we have a natural identification of complexes
\[
\Omega^\bullet_{\mani / \para}(U)\;\cong\;\Omega^\bullet(F) \widehat{\otimes}_{\mathbb{R}}\mathcal O_\para(B),
\]
where the differential is $d_{\mani/ \para}=d_F\otimes \mathrm{id}$, \emph{i.e.}\ the differential acts only along the fiber
directions\footnote{Here $\widehat{\otimes}$ denotes the completed (Fr\'echet) tensor product; since $C^\infty$-spaces
are nuclear, the completion is canonical and will not play any role in the quasi-isomorphism
arguments below.}. Under this identification, the canonical inclusion
\[
\varphi^{-1}(\mathcal O_\para)(U)=\mathcal O_\para(B)\longrightarrow \Omega^\bullet_{\mani / \para}(U)
\]
corresponds to the map $\mathcal O_\para(B)\to \Omega^\bullet(F)\otimes \mathcal O_\para(B)$ given by
$f\mapsto 1\otimes f$. Since $W$ is contractible, the Poincar\'e lemma, Theorem \ref{thm:PL}, yields a quasi-isomorphism
$\mathbb{R} \cong \Omega^\bullet(F),$ hence after tensoring with $\mathcal O_\para(B)$ we obtain a quasi-isomorphism
\[
\mathcal O_\para(B)\;\cong\;\Omega^\bullet(F)\widehat{\otimes}_{\mathbb{R}}\mathcal O_\para(B)
\;\cong\;\Omega^\bullet_{\mani/\para}(U).
\]
Therefore the inclusion $\varphi^{-1}(\mathcal O_\para)\hookrightarrow \Omega^\bullet_{\mani/\para}$ is a
quasi-isomorphism on a basis of the topology of $\mani$, hence it is a quasi-isomorphism of complexes of
sheaves. This yields the claimed canonical isomorphism
$\varphi^{-1}(\mathcal O_S)\simeq \Omega^\bullet_{X/S}$ in $\mbox{\sffamily{{D}}} (\varphi^{-1} (\mathcal{O}_\para))$.\\
For the second isomorphism, it suffices to show that $\Omega^\bullet_{\mani /\para}\cong \Be^\bullet_{\mani /\para}$. For this, in the same setting as above, one has
\[
\Be^\bullet_{\mani/\para}(U)\;\cong\;\Be^\bullet(F)\widehat{\otimes}_{\mathbb{R}}\mathcal O_\para(B).
\]
On the fiber, we have $F$ we have the absolute quasi-isomorphism 
\(
\mathcal{Y}  \maps \Omega^\bullet(F)\xrightarrow{\ \iso\ }\Be^\bullet(F),
\)
given by pairing with the cocycle $Y$ on $F$ as in Theorem \ref{thm:pairing-with-Y-global}.
Tensoring with $\mathcal O_\para(B)$ yields a quasi-isomorphism
\[
\Omega^\bullet_{\mani/\para}(U)\;\xrightarrow{\ \cong\ }\;\Be^\bullet_{\mani/\para}(U).
\]
Since the previous construction is compatible with restrictions, we obtain a quasi-isomorphism of
complexes of sheaves on $U$. As $\mani$ is covered by such trivializing open sets, it follows that
$\Omega^\bullet_{\mani/\para}\cong \Be^\bullet_{\mani/\para}$ as complexes of sheaves on $\mani$, completing the proof.
\end{proof} 
Upon observing that $\mathbf{R}\varphi_*$ and $\mathbf{R}\varphi_!$ send quasi-isomorphisms to quasi-isomorphisms (since they are derived functors), the previous theorem has the following immediate consequence. 
\begin{theorem}[Relative quasi-isomorphism and hypercohomology] Let $\varphi  \maps \mani \to \para$ be a family of supermanifolds. The following are canonical
isomorphisms in $\mbox{\sffamily{\emph{D}}} (\mathcal O_\para)$
\begin{equation}
\mathbf{R}\varphi_*\Omega^\bullet_{\mani/\para}\;\cong\;\mathbf{R}\varphi_*\varphi^{-1}\mathcal O_\para
\;\cong\;\mathbf{R}\varphi_*\Be^\bullet_{\mani/\para},
\qquad
\mathbf{R}\varphi_!\Omega^\bullet_{\mani/\para}\;\cong\;\mathbf{R}\varphi_!\Be^\bullet_{\mani/\para}.
\end{equation}
In particular, the relative de~Rham and Spencer cohomologies (with or without compact supports)
are canonically isomorphic, \emph{i.e.}\
\begin{equation}
H^\bullet_{\dR}(\mani/\para)\cong H^\bullet_{\Sp}(\mani/\para),
\qquad
H^\bullet_{\dR\mathpzc{,c}}(\mani/\para)\cong H^\bullet_{\Sp\mathpzc{,c}}(\mani/\para).
\end{equation}

\end{theorem}

In the rest of this section, we will exploit the powerful formalism of Grothendieck--Verdier duality in order to establish a Poincar\'e duality-like theorem for families of supermanifolds. We will then interpret this result through fiber integration and, in turn, unveil its deep connection to the formalism of picture changing operators used in theoretical physics. This leads to the first rigorous mathematical formulation of the picture changing operators' framework, providing a crucial bridge between supergeometry and physics.

\subsection{Grothendieck--Verdier duality and the dualizing complex}
\label{sec:gv}

Although the Grothendieck--Verdier duality framework extends to a much broader
context\footnote{For a comprehensive treatment of the formalism in its most general form, we refer the reader to thorough treatise \cite{KashSha}.}, here we will only briefly review it in the specific setting of families of
(smooth) supermanifolds $\varphi \maps \mani \to \para$, viewed as morphisms of locally ringed spaces as detailed in Remark \ref{rem:family-ringed}. 

We will systematically work with sheaves of $\varphi^{-1}(\mathcal O_\para)$-modules on $\mani$ (rather than $\mathcal O_\mani$-modules),
since the complexes of interest (\emph{e.g.}\ $\Omega^\bullet_{\mani/\para}$ and $\Be^\bullet_{\mani/\para}$) have
$\varphi^{-1}(\mathcal O_\para)$-linear differentials, as underlined above. We let $\mathbf{Mod}(\varphi^{-1}(\mathcal O_\para))$ be the abelian category of $\varphi^{-1}(\mathcal O_\para)$-modules on
$\mani$, and let $\mathbf{Mod}(\mathcal O_\para)$ be the abelian category of $\mathcal O_\para$-modules on $\para$.
We write
\(
\D^+\!\bigl(\varphi^{-1}(\mathcal O_\para)\bigr),\;
\D^-\!\bigl(\varphi^{-1}(\mathcal O_\para)\bigr),\;
\D^b\!\bigl(\varphi^{-1}(\mathcal O_\para)\bigr)
\)
for the derived categories of bounded below, bounded above, and bounded complexes, respectively
(and similarly $\D^\pm(\mathcal O_\para)$, $\D^b(\mathcal O_\para)$).
Unless otherwise stated, all derived functors below are taken in these categories, and expressions
are understood only under the relevant boundedness hypotheses.

Grothendieck--Verdier duality asserts that $\mathbf{R}\varphi_!:\D^+\!\bigl(\varphi^{-1}(\mathcal O_\para)\bigr)\rightarrow \D^+(\mathcal O_\para)
$ admits a right adjoint
\[
\varphi^!:\D^+(\mathcal O_\para)\longrightarrow \D^+\!\bigl(\varphi^{-1}(\mathcal O_\para)\bigr),
\]
called the \emph{exceptional inverse image}, characterized by a canonical isomorphism
\begin{equation}\label{eq:GV-adjunction-rel}
\mathbf{R}\mathrm{Hom}_{\mathcal O_S}\!\bigl(R\varphi_!(\mathcal M^\bullet),\,\mathcal N^\bullet\bigr)
\;\cong\;
\mathbf{R}\mathrm{Hom}_{\varphi^{-1}\mathcal O_S}\!\bigl(\mathcal M^\bullet,\,\varphi^!(\mathcal N^\bullet)\bigr),
\end{equation}
functorial in $\mathcal M^\bullet\in D^-\!\bigl(\varphi^{-1}(\mathcal O_\para)\bigr)$ and
$\mathcal N^\bullet\in D^+(\mathcal O_\para)$.
Equivalently, \eqref{eq:GV-adjunction-rel} says that $\varphi^!$ is a right adjoint of $R\varphi_!$.

The complex obtained from applying the exceptional inverse image functor $\varphi^!$ to the structure sheaf $\Oh_\para$ is so important to bear its own name.
\begin{definition}[Dualizing Complex] Let $\varphi\maps \mani \rightarrow \para$ be a family of supermanifolds. We call $\varphi^{!} (\Oh_\para) \in \D^+ (\varphi^{-1} (\Oh_\para))$ the \define{dualizing complex} of the family $\varphi \maps \mani \rightarrow \para$ and we denote it by $\omega^\bullet_{\mani / \para}$.
\end{definition}
Specializing \eqref{eq:GV-adjunction-rel} to $\mathcal N^\bullet=\mathcal O_S$ yields
\begin{equation}\label{eq:dualizing-adjunction-rel}
\mathbf{R}\mathrm{Hom}_{\mathcal O_\para}\!\bigl(\mathbf{R}\varphi_!(\mathcal M^\bullet),\,\mathcal O_\para\bigr)
\;\cong\;
\mathbf{R}\mathrm{Hom}_{\varphi^{-1}(\mathcal O_\para)}\!\bigl(\mathcal M^\bullet,\,\omega^\bullet_{\mani/\para}\bigr),
\end{equation}
for $\mathcal M^\bullet\in \D^-\!\bigl(\varphi^{-1}(\mathcal O_\para) \bigr).$

The following lemma provides an explicit expression for the dualizing complex $\omega^\bullet_{\mani / \para}.$

\begin{theorem}[Dualizing complex and relative Berezinian]\label{prop:omega-is-ber-relative}
Let $\varphi \maps \mani\to \para$ be a family of supermanifolds of relative dimension $m|n$. Then there is a canonical isomorphism in $\mbox{\sffamily{\emph{D}}} (\varphi^{-1} (\mathcal{O}_\para))$
\begin{equation}
\omega^\bullet_{\mani/\para}\;\cong\;\Be_{\mani/\para}[m].
\end{equation}
\end{theorem}

\begin{proof}
The statement is local on $\mani$. Fix $x\in \mani$ and choose a trivializing neighborhood
$U\subset \mani$ and an open $B\subset |\para|$ with $\varphi(U)\subseteq B$ such that
$U\;\cong\;F\times B,$ for $F$ an open subset of the typical fiber, and where $\varphi|_U=\mathrm{pr}_B$, the projection onto $B$.
Shrinking $U$ if necessary, we may assume that $F$
is contractible. It suffices to construct on each such $U$ a canonical isomorphism
$\omega^\bullet_{\mani/\para}|_U\simeq \Ber_{\mani/\para}|_U[m]$ compatible with restriction; these local
isomorphisms then glue.

On $U\simeq F\times B$, for $\varphi=\mathrm{pr}_B$, Grothendieck--Verdier duality gives the adjunction
\begin{equation}
\mathbf{R}\mathrm{Hom}_{\mathcal O_B}\!\bigl(\mathbf{R}\mathrm{pr}_{B\,!}(\mathcal M^\bullet),\mathcal O_B\bigr)
\cong
\mathbf{R}\mathrm{Hom}_{\mathrm{pr}_B^{-1}(\mathcal O_B)}\!\bigl(\mathcal M^\bullet,\mathrm{pr}_B^!(\mathcal O_B)\bigr),
\end{equation}
for $\mathcal M^\bullet\in \D^-\!\bigl(\mathrm{pr}_B^{-1}(\mathcal O_B)\bigr)$.
Taking $\mathcal M^\bullet=\mathrm{pr}_B^{-1}\mathcal O_B$ and using $\mathrm{pr}_B^!\mathcal O_B=\omega^\bullet_{X/S}|_U$,
we obtain in $\D(\mathcal O_B)$ an identification
\begin{equation}\label{eq:omega-as-dual}
\Gamma\!\bigl(U,\omega^\bullet_{X/S}\bigr)\;\cong\;
\mathbf{R}\mathrm{Hom}_{\mathcal O_B}\!\bigl(\mathbf{R}\Gamma_\mathpzc{c}(U/B,\mathrm{pr}_B^{-1}(\mathcal O_B)),\,\mathcal O_B(B)\bigr),
\end{equation}
where $\Gamma_\mathpzc{c}(U/B,-)$ denotes sections with compact support in the $F$-direction (fiberwise compact
support). In the smooth category, the relative Poincar\'e lemma identifies
$\mathrm{pr}_B^{-1}\mathcal O_B\simeq \Omega^\bullet_{U/B}$, hence $\mathbf{R}\Gamma_\mathpzc{c}(U/B,\mathrm{pr}_B^{-1}(\mathcal O_B))$ is
computed by the complex of compactly supported relative forms $\Omega^\bullet_{\mathpzc{c}}(F)\widehat{\otimes}\mathcal O_B(B)$.
Thus \eqref{eq:omega-as-dual} becomes
\begin{equation}\label{eq:omega-as-RHom}
\Gamma\!\bigl(U,\omega^\bullet_{\mani / \para}\bigr)\;\cong\;
\mathbf{R}\mathrm{Hom}_{\mathbb R}\!\bigl(\Omega^\bullet_{\mathpzc{c}}(F),\mathbb R\bigr)\;\widehat{\otimes}_{\mathbb R}\;\mathcal O_B(B)
\;\cong\;
\mathbf{R}\mathrm{Hom}_{\mathbb R}\!\bigl(\Omega^\bullet_{\mathpzc{c}}(F),\mathcal O_B(B)\bigr).
\end{equation}

Since $\Be^m(W)=\Ber(W)$, Berezin integration provides a canonical pairing,\footnote{This will be discussed in details in the next section.}
$
\Ber(F)\otimes_{\mathbb R}\Omega^0_\mathpzc{c}(F)\rightarrow \mathbb R,
$
given by 
\begin{equation}
(\be,f)\longmapsto \int_F \be f,
\end{equation}
where $\be f$ denotes multiplication of a Berezinian top integral form by a compactly supported function. 
Using this pairing, define a morphism of complexes
\begin{equation}\label{eq:Psi}
\Psi_F:\ \Ber_F[m]\longrightarrow \mathbf{R}\mathrm{Hom}_{\mathbb R}\!\bigl(\Omega^\bullet_\mathpzc{c}(F),\mathbb R\bigr)
\end{equation}
by sending $\be[m]$ to the cochain $\Psi_F(\be[m])$ which is zero on $\Omega^k_\mathpzc{c}(F)$ for $k\neq 0$
and on $\Omega^0_\mathpzc{c}(F)$ is given by $f\mapsto \int_F \be f$.
This is a cochain map, and it follows from compactly supported Poincar\'e lemma that it is a quasi-isomorphism.

By the previous discussion, on $U\cong F\times B$,
one has
$\Gamma\!\bigl(U,\omega^\bullet_{\mani/\para}\bigr)\cong
\mathbf{R}\mathrm{Hom}_{\mathbb{R}}\!\bigl(\Omega^\bullet_\mathpzc{c}(F),\,\mathcal O_B(B)\bigr),
$ so that, the canonical morphism $\Psi_F$ induces a canonical quasi-isomorphism
\begin{equation}
\mathbf{R}\mathrm{Hom}_{\mathbb{R}}\!\bigl(\Omega^\bullet_\mathpzc{c}(F),\,\mathcal O_B(B)\bigr)\;\cong\;
\Ber_F[m]\;\widehat{\otimes}_{\mathbb R}\;\mathcal O_B(B),
\end{equation}
and the right-hand side identifies with $\Gamma\!\bigl(U,\Ber_{\mani/\para}[m]\bigr)$ under the trivialization
$U\simeq F\times B$.
Therefore, on $U$ we obtain an isomorphism in $\D(\mathcal O_B)$
\begin{equation}
\Gamma\!\bigl(U,\omega^\bullet_{\mani/\para}\bigr)\;\cong\;\Gamma\!\bigl(U,\Ber_{\mani/\para}[m]\bigr).
\end{equation}
Equivalently, as complexes of sheaves of $\varphi^{-1}(\mathcal O_\para)$-modules on $U$, one has that in $\D \bigl(\varphi^{-1}\mathcal O_\para|_U\bigr)$
\begin{equation}
\omega^\bullet_{\mani/\para}\big|_U\;\cong\;\Ber_{\mani/\para}\big|_U[m].
\end{equation}
Finally, these local identifications are compatible with restriction to smaller trivializations and
with changes of fiber coordinates, hence they glue on
overlaps and yield a global isomorphism
$
\omega^\bullet_{\mani/\para}\cong\Ber_{\mani/\para}[m]$
in $\D \bigl(\varphi^{-1}(\mathcal O_\para)\bigr)$, as claimed. \end{proof}

The next result lists a series of consequences of the previous theorem.

\begin{lemma}[Triviality of the dualizing complex / relative Berezinian]\label{lem:dualizing-vs-berezinian}
Let $\varphi  \maps \mani \to \para$ be a family of supermanifolds of relative dimension $m|n$, and
let $\omega^\bullet_{\mani/\para}$ be the dualizing complex.
Then the following are equivalent:
\begin{enumerate}
\item there is an isomorphism $\omega^\bullet_{\mani/\para}\cong \varphi^{-1}(\mathcal O_\para)[m]$ in
$\mbox{\sffamily{\emph{D}}} (\varphi^{-1} (\mathcal{O}_\para))$;
\item the relative Berezinian line bundle $\Be_{\mani / \para }$ is trivial as a $\varphi^{-1}\mathcal (O_\para)$-module;
\item there exists a nowhere-vanishing relative Berezinian
$\hat \be \in \Gamma(\mani,\Be_{\mani/\para})$.
\end{enumerate}
\end{lemma}

\begin{proof} By Theorem \ref{prop:omega-is-ber-relative}, the dualizing complex is canonically identified as
$\omega^\bullet_{\mani/\para}\cong\Be_{\mani/\para}[m]$
in $\D(\varphi^{-1}\mathcal O_\para).$ Therefore $\omega^\bullet_{\mani/\para}\cong \varphi^{-1}(\mathcal O_\para)[m]$ holds if and only if
$\Be_{\mani/\para}\cong \varphi^{-1}(\mathcal O_\para)$, which in turn holds if and only if $\Ber_{\mani/\para}$
admits a nowhere-vanishing global section.
\end{proof}

\noindent
As a word of comment on the previous results, we stress that Theorem \ref{prop:omega-is-ber-relative} and the following Lemma \ref{lem:dualizing-vs-berezinian} hold in the same form also for more general morphisms (\emph{e.g.}\ surjective submersion). On the other hand, we stress that the orientability of the individual fibers is not sufficient to get results of the form of Lemma \ref{lem:dualizing-vs-berezinian}. Indeed, what is needed is a set of orientations that vary continuously (locally-constantly) along the base -- this is what is given by the nowhere-vanishing section of the Berezinian. In fact, one has the following sufficient condition.

\begin{corollary}\label{cor:orientable-implies-trivial-dualizing}
If both $\mani$ and $\para$ are orientable, then there is an isomorphism in $\mbox{\sffamily{\emph{D}}} (\varphi^{-1} (\mathcal{O}_\para))$
\begin{equation}
\omega^\bullet_{\mani/\para}\;\cong\;\varphi^{-1}(\mathcal O_\para)[m].
\end{equation}
\end{corollary}
\begin{proof}
Orientability of $\mani$ and $\para$ implies that $\Be_\mani$ and $\Be_\para$ are trivial line bundles.
On the other hand, for a smooth family one has the standard factorization $
\Be_\mani \cong \Be_{\mani/\para} \otimes \varphi^*(\Be_\para).$ Hence $\Be_{\mani/\para}$ is trivial, so Lemma~\ref{lem:dualizing-vs-berezinian} applies.
\end{proof}

\noindent
Finally, we stress that while $\Omega^\bullet_{\mani / \para} \cong \varphi^{-1} (\mathcal{O}_\para) \cong  \bersheaf^\bullet_{\mani}$ were canonical identifications, the previous one is not.

\subsection{Relative Poincar\'e duality}
\label{sec:relpoinc}

The results of the previous section enable one to establish the following theorem, which can be interpreted as a Poincar\'e-Verdier duality for families of supermanifolds.

\begin{theorem}[Poincar\'e--Verdier Duality] \label{PV} Let $\varphi \maps \mani \to \para$ be a family of supermanifolds of relative dimension $m|n$.  
There is a canonical isomorphism in $\mbox{\sffamily{\emph{D}}}(\mathcal{O}_\para)$
\begin{equation}\label{eq:PV-general-final}
\mathbf{R}\mathcal H \mbox{\emph{om}}_{\mathcal O_\para}\!\bigl(\mathbf{R}\varphi_!\,\varphi^{-1}\mathcal O_\para, \mathcal O_\para\bigr)
\;\cong\;
\mathbf{R}\varphi_*\omega^\bullet_{X/S},
\end{equation}
where $\omega^\bullet_{\mani / \para} \cong \Be_{\mani / \para}[m]$ is the dualizing complex. In particular, if $\mani$ and $\para$ are orientable, upon choosing a relative orientation, the previous isomorphism takes the form 
\begin{equation}\label{eq:PV-orientable-final}
\mathbf{R}\mathcal H \mbox{\emph{om}}_{\mathcal O_\para}\!\bigl(\mathbf{R}\varphi_!\,\Be^\bullet_{\mani/\para}, \mathcal O_\para\bigr)
\;\cong\;
\mathbf{R}\varphi_*\Omega^\bullet_{\mani/\para}[m].
\end{equation}
\end{theorem}

\begin{proof}
The sheafified form of Grothendieck--Verdier duality yields a canonical isomorphism in $\D(\mathcal O_\para)$
\begin{equation}\label{eq:GV-sheafified-used}
\mathbf{R}\mathcal H \mbox{om}_{\mathcal O_\para}\!\bigl(\mathbf{R}\varphi_!(\mathcal M^\bullet), \mathcal N^\bullet\bigr)
\;\cong\;
\mathbf{R}\varphi_*\,\mathbf{R}\mathcal H \mbox{om}_{\varphi^{-1}(\mathcal O_\para)}\!\bigl(\mathcal M^\bullet, \varphi^!\mathcal N^\bullet\bigr),
\end{equation}
functorial in $\mathcal M^\bullet\in \D^-\!\bigl(\varphi^{-1}(\mathcal O_\para)\bigr)$ and
$\mathcal N^\bullet\in \D^+(\mathcal O_\para)$.
Now, take $\mathcal N^\bullet=\mathcal O_\para$ and $\mathcal M^\bullet=\varphi^{-1}(\mathcal O_\para)$. Since
\[
\mathbf{R}\mathcal H om_{\varphi^{-1}\mathcal O_\para}\!\bigl(\varphi^{-1}(\mathcal O_\para), \varphi^!\mathcal O_\para\bigr)
\;\cong \;
\varphi^!\mathcal O_\para \;=\; \omega^\bullet_{\mani/\para},
\]
\eqref{eq:GV-sheafified-used} yields equation \eqref{eq:PV-general-final}, as claimed.

Assume now that both $\mani$ and $\para$ are orientable. Then, by Corollary \ref{cor:orientable-implies-trivial-dualizing}, we have an isomorphism
$\omega^\bullet_{\mani/\para}\cong\varphi^{-1}\mathcal O_\para[m]$ in $\D(\varphi^{-1}(\mathcal O_\para ))$.
Moreover, the relative Poincar\'e lemmas give quasi-isomorphisms
$
 \Omega^\bullet_{\mani/\para} \cong  \varphi^{-1}\mathcal O_\para \cong \Be^\bullet_{\mani/\para}.
 $
Substituting these into \eqref{eq:PV-general-final} (and using functoriality of $\mathbf{R}\varphi_!$ and
$\mathbf{R}\mathcal H \mbox{om}$) yields \eqref{eq:PV-orientable-final}.
\end{proof}

We remark that, analogously, for $\mani$ and $\para$ orientable, we could also obtain the isomorphism 
\begin{equation}\label{eq:PV-orientable-final2}
\mathbf{R}\mathcal H \mbox{{om}}_{\mathcal O_\para}\!\bigl(\mathbf{R}\varphi_!\,\Omega^\bullet_{\mani/\para}, \mathcal O_\para\bigr)
\;\cong\;
\mathbf{R}\varphi_*\Be^\bullet_{\mani/\para}[m].
\end{equation}

The effect of cohomology on the previous quasi-isomorphisms is studied in the next theorem, which in fact yields a relative notion of Poincar\'e duality, and clearly relate relative de Rham and Spencer cohomologies in the orientable case.

\begin{theorem}[Relative Poincar\'e duality]\label{thm:relative-PD-sheaf}
Let $\varphi  \maps \mani \to \para$ be a smooth family of supermanifolds of relative dimension $m|n$.
There are canonical isomorphisms of $\mathcal O_\para$-modules (equivalently,
of local systems) for all $i\in\mathbb Z$:
\begin{equation}\label{eq:relPD-sheaves}
\mathbf{R}^{i}\varphi_*\bigl(\varphi^{-1}(\mathcal O_\para ) \bigr)
\;\cong\;
\mathcal{H}\mbox{\emph{om}}_{\mathcal{O}_\para} \left (\mathbf{R}^{m-i}\varphi_!\Be_{\mani/\para}, \mathcal{O}_\para \right).
\end{equation}
Equivalently, there is a canonical perfect pairing of sheaves
\begin{equation}
\mathbf{R}^{i}\varphi_*\bigl(\varphi^{-1}(\mathcal O_\para) \bigr) \otimes_{\mathcal O_S}
\mathbf{R}^{m-i}\varphi_!\Be_{\mani/\para} \longrightarrow \mathcal O_S.
\end{equation}
In particular, if both $\mani$ and $\para$ are orientable, upon choosing a relative orientation, the previous isomorphism takes the form
\begin{equation} \label{eq:relPD-sheaves-orient}
\mathbf{R}^{m-i}\varphi_*\Omega^\bullet_{\mani / \para} \, \cong \, \mathcal{H}\mbox{\emph{om}}_{\mathcal{O}_\para}\bigl (\mathbf{R}^{i}\varphi_!\Be^\bullet_{\mani / \para}, \mathcal{O}_\para \bigr ).
\end{equation}
\end{theorem}

\begin{proof}
Take cohomology sheaves on both sides of
\ref{eq:PV-general-final}. Since $\omega^\bullet_{\mani / \para} \cong \Be_{\mani/\para}[m]$, we have
\begin{equation}
\mathcal H^{i}\!\bigl(\mathbf{R}\varphi_*\Be_{\mani/\para}[m]\bigr)\;\cong\;\mathbf{R}^{i+m}\varphi_*\Be_{\mani/\para}.
\end{equation}
On the other hand, for any $\mathcal{A}\in \D^+(\mathcal O_\para)$,
\begin{equation}
\mathcal H^{i}\!\bigl(\mathbf{R}\mathcal H \mbox{om}_{\mathcal O_\para}(\mathcal{A},\mathcal O_\para)\bigr)
\;\cong\;
\mathcal E \mbox{xt}^{-i}_{\mathcal O_\para}(\mathcal{A},\mathcal O_\para).
\end{equation}
Applying this to $\mathcal{A}=\mathbf{R}\varphi_!\,\varphi^{-1}\mathcal O_\para$ gives
\begin{equation}
\mathcal H^{i}\!\bigl(\mathbf{R}\mathcal H \mbox{om}_{\mathcal O_\para}\bigl(\mathbf{R}\varphi_!\varphi^{-1}\mathcal O_\para,\mathcal O_\para\bigr)\bigr)
\;\cong\;
\mathcal E \mbox{xt}^{-i}_{\mathcal O_\para}\!\bigl(\mathbf{R}\varphi_!\varphi^{-1}(\mathcal O_\para),\mathcal O_\para\bigr).
\end{equation}
In our smooth setting, the cohomology sheaves $\mathbf{R}^j\varphi_!\varphi^{-1}(\mathcal O_\para)$ are locally free
of finite rank (local systems), hence the derived $\mathcal E \mbox{xt}$ reduces to $\mathcal{H}\mbox{om}$ on
cohomology sheaves:
\begin{equation}
\mathcal E \mbox{xt}^{q}_{\mathcal O_\para}\!\bigl(\mathbf{R}\varphi_!\varphi^{-1}(\mathcal O_\para),\mathcal O_\para\bigr)
\;\cong\;
\mathcal H \mbox{om}_{\mathcal O_\para}\!\bigl(\mathcal H^{-q}(\mathbf{R}\varphi_!\varphi^{-1}(\mathcal O_\para)),\mathcal O_\para\bigr)
\;=\;
\mathcal{H}\mbox{om}_{\mathcal{O}_\para}\left (\mathbf{R}^{-q}\varphi_!\varphi^{-1}(\mathcal O_\para), \mathcal{O}_\para\right).
\end{equation}
Putting $q=-i$ we obtain
\begin{equation}
\mathcal H^{i}\!\bigl(\mathbf{R}\mathcal H om_{\mathcal O_\para}\bigl(R\varphi_!\varphi^{-1}(\mathcal O_\para),\mathcal O_\para\bigr)\bigr)
\;\cong\;
\mathcal{H}\mbox{om}_{\mathcal{O}_\para} \left (\mathbf{R}^{i}\varphi_!\varphi^{-1}(\mathcal O_\para), \mathcal{O}_\para \right ).
\end{equation}
Comparing cohomology sheaves in \eqref{eq:PV-general-final} therefore yields, for every $i$,
\begin{equation}
\mathcal{H}\mbox{om}_{\mathcal{O}_\para}\left (\mathbf{R}^{i}\varphi_!\varphi^{-1}(\mathcal O_\para), \mathcal{O}_\para \right )\;\cong\;\mathbf{R}^{m-i}\varphi_*\Be_{\mani/\para},
\end{equation}
which is equivalent to \eqref{eq:relPD-sheaves} upon re-indexing. The pairing formulation is the evaluation pairing associated with the duality isomorphism.

Finally, in the orientable case choose a trivialization $\Be_{\mani/\para}\cong \varphi^{-1}(\mathcal O_\para)$ and resolve via relative Poincar\'e lemmas to get equation \ref{eq:relPD-sheaves-orient}.
\end{proof}

If one further assumes that the base supermanifold $\para$ is simply connnected, then relative Poincar\'e duality becomes a neat statement, relating de Rham and Spencer cohomology of the fiber supermanifold.

\begin{lemma}[Relative Poincar\'e lemma -- simply connected base]\label{cor:relPD-simply-connected}
Let $\varphi  \maps \mani \to \para$ be a family of supermanifolds of relative dimension $m|n$ such that $\mani$ and $\para$ are orientable and $\para$ is simply connected.
Then, upon choosing a relative orientation, one has that 
 \begin{equation} \label{eq:relPD-simplycon}
 H^{m-i}_{\mathpzc{dR}} (\fib) \otimes \mathcal{O}_\para \; \iso \; \mbox{\emph{Hom}}_{\mathbb{R}} ( H^{i}_{\mathpzc{Sp}, \mathpzc{c}} (\fib), \mathbb{R} ) \otimes \mathcal{O}_\para.
\end{equation}
The same holds true for a trivial family $\varphi  \maps \fib \times \para \to \para$ without assuming simply-connectedness of $\para$.
\end{lemma}
\begin{proof}  First, a trivializing open set $B\subset |\para|$, we identify $\varphi^{-1}(B)$ 
with $\fib \times B$, for $\fib$ the typical fiber. One has 
$\Gamma(\fib \times B,\Omega^\bullet_{\mani/\para})
\cong
\Omega^\bullet(\fib)\,\widehat{\otimes}_{\mathbb R}\,\mathcal O_\para(B),
$
with differential $d_\fib \otimes \mathrm{id}$.
Taking cohomology yields
\begin{equation}
H^i\Gamma(\fib \times B,\Omega^\bullet_{\mani/\para})
\;\cong\;
H^i_{\mathpzc{dR}}(\fib)\otimes_{\mathbb R}\mathcal O_\para(B).
\end{equation}
Finally, by definition of derived direct image,
\begin{equation}
\bigl(R^i\varphi_*\Omega^\bullet_{\mani/\para}\bigr)(B)
\;=\;
\mathbb H^i\!\bigl(\varphi^{-1}(B),\Omega^\bullet_{\mani/\para}\bigr)
\;\cong\;
H^i\Gamma\!\bigl(\varphi^{-1}(B),\Omega^\bullet_{\mani/\para}\bigr),
\end{equation}
where the last identification uses that the terms of $\Omega^\bullet_{X/S}$ are fine sheaves in the
smooth category. Consequently, for each $i$ one has a canonical isomorphism
\begin{equation}
\bigl(\mathbf{R}^i\varphi_*\Omega^\bullet_{\mani/\para}\bigr)\big|_{B}
\;\cong\;
H^i_{\mathpzc{dR}}(F)\otimes_{\mathbb R}\mathcal O_\para\big|_{B}.
\end{equation}
Then, the sheaf $\mathbf{R}^i\varphi_*\Omega^\bullet_{\mani/\para}$ is a
locally free $\mathcal O_\para$-module with locally constant transition functions, \emph{i.e.}\ it is a (Gauss-Manin) local system
of $\mathbb R$-vector spaces tensored with $\mathcal O_\para$. On a simply-connected base, every 
local system is constant, hence one has the global identification 
\begin{equation}
\mathbf{R}^i\varphi_*\Omega^\bullet_{\mani/\para}\;\cong\;H^i_{\mathpzc{dR}}(\fib)\otimes_{\mathbb R}\mathcal O_\para.
\end{equation}
Identifying also the sheaf at the right hand side of \eqref{eq:relPD-sheaves-orient} in the very same fashion, one concludes the verification.

For trivial families, all local systems are constant without further assumptions, as there can be no non-trivial monodromy.
\end{proof}

 Finally, note that for a collapse map $\varphi  \maps \mani \to \mathbb{R}^{0|0}$ to a point, the previous reduces to the usual Poincar\'e duality on a single supermanifold, as stated in \cite{Manin, Noja}.
\begin{corollary}[Poincar\'e Duality] \label{cor:PDpoint} Let $\mani$ be an orientable supermanifold of dimension $m|n$.  Then one has
\begin{eqnarray}
H^i_{\mathpzc{dR}} (\mani) \iso \mathrm{Hom}_{\mathbb{R}} (H^{m-i}_{\mathpzc{Sp}, \mathpzc{c}} (\mani) , \mathbb{R}). 
\end{eqnarray}
\end{corollary}

\subsection{Berezin fiber integration}
\label{sec:fiberint}

The duality results of the previous section admit a clear and concrete interpretation as a supergeometric analog of a (relative) wedge-and-integration pairing, once a notion of Berezin fiber integration is introduced.

As we shall see, in the supergeometric setting, the wedge-and-integration pairing --- referred to here as the Poincar\'e pairing --- must be specifically realized as a pairing between differential and integral forms; otherwise, a meaningful notion of integration cannot be defined.
Before proceeding, and for future reference, we recall some fundamental aspects of ordinary Berezin integration --- see \cite{Noja} and chapter 8 of \cite{Kessler} for further details.

\subsubsection*{\textbf{Interlude on Berezin integration}}

We recall that, given a connected and orientable supermanifold $\mani$ of dimension $m|n$, denoting with $\bersheaf^m_{\mathpzc{c}} (\mani) \defeq \Gamma_{\mathpzc{c}} (\mani, \bersheaf_\mani)$ the $\R$-module of the compactly supported sections of the Berezinian sheaf, then the Berezin integral is the $\mathbb{R}$-linear map
\begin{eqnarray} \label{BerezinInt1}
\int_\mani \maps \bersheaf^m_{\mathpzc{c}} (\mani) \longrightarrow \R,
\end{eqnarray}
which is defined in local coordinates $X^A \defeq \{ x_i | \theta^\alpha \}$ with $i = 1\ldots,m$, $\alpha = 1, \ldots, n$ on a domain $U$ by
\begin{eqnarray} \label{BerezinU}
\int_U \ber_\mathpzc{c} \defeq \int_{U_{\mathpzc{red}}} dx_{1, \mathpzc{red}}\ldots dx_{m, \mathpzc{red}} f^{1\ldots 1}_\mathpzc{c} (x_{1, \mathpzc{red}} \ldots x_{m, \mathpzc{red}}),
\end{eqnarray}
where $ \be_\mathpzc{c} (x, \theta) = \sum_{\underline
\epsilon} \ber^c \theta^1_{\epsilon_1} \ldots \theta^n_{\epsilon_n} f^{\underline
\epsilon}_\mathpzc{c} (x_1, \ldots, x_m)$ for $\epsilon_j =\{0, 1\}$ and $\ber^c$ a
constant Berezinian. Equation \eqref{BerezinU} is then extended via a partition of
unity to the whole supermanifold. 
More invariantly, it is possible to introduce Berezin integration by observing the
following:
\begin{enumerate}[leftmargin=+13pt]
\item the $\R$-module of compactly supported sections of the Berezinian decomposes as
\begin{eqnarray}
\bersheaf^{m}_{\mathpzc{c}} (\mani) = \mathcal{J}^m_\mani \bersheaf^{m}_\mathpzc{c} (\mani) + \delta (\bersheaf^{m-1}_{\mathpzc{c}} (\mani)),
\end{eqnarray}
where $\mathcal{J}_\mani$ are the nilpotent sections in $\mathcal{O}_\mani;$
\item the isomorphism of {sheaves} of $\mathcal{O}_{\mani_{\mathpzc{red}}}$-modules $\phi \maps \mathcal{J}^m_\mani \bersheaf_\mani \to \Omega^m_{\mani_{\mathpzc{red}}}$ specializes to the previous decomposition to give in the intersection
 \begin{eqnarray}
\mathcal{J}^m_\mani \bersheaf^{m}_\mathpzc{c} (\mani) \cap \delta (\bersheaf^{m-1}_{\mathpzc{c}} (\mani)) \iso d (\Omega^m_{\mathpzc{c}} (\mani_{\mathpzc{red}})).
\end{eqnarray}
\end{enumerate}
One can then look at the Berezin integral as the map given by the composition,
\begin{eqnarray} \label{compomap}
\xymatrix{
{\bersheaf^m_\mathpzc{c} (\mani)}/{\delta (\bersheaf^{m-1}_{\mathpzc{c}} (\mani))} \ar[r]^{\hat{\phi}\;\;} & {\Omega_{\mathpzc{c}}^m} (\mani_{\mathpzc{red}})/{d(\Omega^{m-1}_{\mathpzc{c}} (\mani_{\mathpzc{red}}))} \ar[r]^{\qquad \qquad \int} & \R,
}\end{eqnarray}
where the first map $\hat \phi $ is in fact an isomorphism and the second map is the ordinary integration. It follows that the Berezin integral yields a linear isomorphism
\begin{eqnarray} \label{BerezinInt}
\int_\mani \maps H^m_{\mathpzc{Sp}, \mathpzc{c}} (\mani) \stackrel{\cong}{\longrightarrow} \R,
\end{eqnarray}
provided that $\mani$ is connected and orientable---see for example the recent \cite{Noja} for details.

In order to generalize the above construction to families of supermanifolds and introduce a related notion of \emph{fiber} integration, we first
recall that the Grothendieck--Verdier we discussed in the previous section can be seen as an adjunction, namely
\begin{equation}
\mathbf{R}\varphi_!:\D^+\!\bigl(\varphi^{-1}(\mathcal O_\para ) \bigr) \rightleftarrows \D^+(\mathcal O_\para ):\ \varphi^!.
\end{equation}
Further, we stress that in this section, we will always assume $\mani$ and $\para$ to be orientable, and we choose a relative orientation, \emph{i.e.}\ a trivialization
$\Be_{\mani /\para }\cong \varphi^{-1}(\mathcal O_\para)$, hence $\omega^\bullet_{\mani/\para} \cong \Be^m_{\mani / \para}$ is such that $\omega_{\mani / \para}^\bullet \cong \varphi^{-1}(\mathcal O_\para)[m]$, and in turn
-- resolving by relative Poincar\'e lemma -- we can write $\omega^\bullet_{\mani / \para} \cong \Be^\bullet_{\mani / \para} [m].$ Finally, we will always take the typical fiber supermanifold $\fib$ to be connected. \\

\begin{definition}[Berezin fiber integration]\label{def:fiber-integration} Let $\varphi  \maps \mani \to \para$ be a family of supermanifolds of
relative dimension $m|n$. The \define{fiber integration} morphism is the counit of the adjunction $\mathbf{R}\varphi_! \dashv \varphi^!$
evaluated at $\mathcal O_\para$:
\begin{equation} \label{fibintgen}
\int_{\mani/\para}: \mathbf{R}\varphi_!\,\omega^\bullet_{\mani/\para}\longrightarrow \mathcal O_\para
\end{equation}
in $ \mbox{\sffamily{{D}}}(\mathcal{O}_\para)$, where $\omega^\bullet_{\mani / \para} \cong \Be^m_{\mani / \para}[m]$. \\
For $\mani $ and $\para$ orientable, transporting $\int_{\mani / \para}$ along the identification $\omega^\bullet_{\mani / \para} \cong \Be^\bullet_{\mani / \para}[m]$, 
yields a (non-canonical) morphism in $ \mbox{\sffamily{{D}}}(\mathcal{O}_\para)$, which we call the \define{Berezin fiber integration}
\begin{equation} \label{berfibint}
\int_{\mani / \para}:  \mathbf{R}\varphi_!\,\Be^\bullet_{\mani / \para}[m]\longrightarrow \mathcal O_\para.
\end{equation}
\end{definition}

\begin{remark}[Local formula for fiber integration]\label{rem:fiber-integral-local}
On a local trivialization $U\simeq F \times B$ with $\varphi|_U=\mathrm{pr}_B$, the morphism
$\int_{\mani / \para}$ restricts to the usual Berezin integral along the typical fiber $\fib$ on compact supports, tensored with
$\mathcal O_B$. 
\end{remark}

The previous definition is related to the ordinary Berezin integration via the following.

\begin{lemma}[Induced morphism on cohomology sheaves]\label{lem:int-on-sheaves}
Let $\varphi  \maps \mani \to \para $ be a family of supermanifolds of relative dimension $m|n$. The morphism $\int_{\mani / \para}$ induces on cohomology sheaves
a canonical map
\begin{equation}
\mathcal H^0\!\bigl(\mathbf{R}\varphi_!\Be^\bullet_{\mani/\para}[m]\bigr)=\mathbf{R}^{m}\varphi_!\Be^\bullet_{\mani/\para}
\ \longrightarrow\ \mathcal O_\para.
\end{equation}
For every point $|S|$, the induced map on stalks identifies with the (absolute) Berezin
integration isomorphism of Equation \ref{BerezinInt} on the fiber $\fib$:
\begin{equation}
\int_\fib : H^m_{\mathpzc{Sp,c}}(\fib)\ \xrightarrow{\ \cong\ }\ \mathbb R.
\end{equation}
\end{lemma}

\begin{proof}
Applying $\mathcal H^0(-)$ to $\int_{\mani / \para}$ yields the displayed morphism
$\mathbf{R}^m\varphi_!\Be^\bullet_{\mani/\para}\to \mathcal O_\para$.
To identify the stalk at a point $s \in |\para|$, choose a trivializing neighborhood $B\ni s$ so that
$\varphi^{-1}(B)\simeq \fib \times B$.
By Remark~\ref{rem:fiber-integral-local}, the restriction of $\int_{\mani / \para}$ to $B$ is given by
fiberwise Berezin integration, hence the stalk map is precisely the (absolute) Berezin integral on $\fib$,
which is an isomorphism.
\end{proof}

Now, recalling that $\Be^\bullet_{\mani / \para}$ is a $\Omega^\bullet_{\mani / \para}$-module, we can give the following definition.

\begin{definition}[Relative Poincar\'e pairing]\label{def:relative-Poincare-pairing-Be-Omega}
Let $\varphi  \maps \mani \to \para $ be a family of supermanifolds of relative dimension $m|n$. For each $i$ we define the \define{relative Poincar\'e pairing}
as the composite morphism of sheaves on $\para$
\begin{equation}
\PP_{\mani / \para} : \mathbf{R}^i\varphi_*\Omega^\bullet_{\mani/\para} \otimes_{\mathcal O_\para} \mathbf{R}^{m-i}\varphi_!\Be^\bullet_{\mani/\para}
\ \xrightarrow{ \cdot   }\ \mathbf{R}^m\varphi_!\Be^\bullet_{\mani/\para}
\ \xrightarrow{\ \mathcal H^0(\int_{\mani / \para})\ }\ \mathcal O_\para,
\end{equation}
where the first arrow is induced by the $\Omega^\bullet_{\mani/\para}$-module structure on $\Be^\bullet_{\mani/\para}$, and the second arrow is the map of
Lemma~\ref{lem:int-on-sheaves}.
\end{definition}

\begin{remark}[Local expression of the relative Poincar\'e pairing]\label{rem:local-relative-pairing}
Let $U\subset \mani$ be a trivializing open set for $\varphi  \maps \mani \to \para$, so that
$U \cong F\times B$ and $ \varphi|_{U}=\mbox{pr}_{B}$,
with $F$ an open domain in the fiber and $B\subset \para$ open.
On $U$ we have the usual identifications
$
\Omega^\bullet_{\mani/\para}(U)\;\cong\;\Omega^\bullet(F)\,\widehat{\otimes}_{\mathbb R}\,\mathcal O_B(B),
$ and $
\Be^\bullet_{X/S}(U)\;\cong\;\Be^\bullet(F)\,\widehat{\otimes}_{\mathbb R}\,\mathcal O_B(B),
$
where the differentials act only along the $F$-factor and $\widehat{\otimes}$ denotes the completed tensor product.
Then the relative Poincar\'e pairing $\PP_{\mani / \para}$  is given locally on $B$ by the following formula: if $[\omega]$ is represented by
$
\omega\in \Omega^i(F)\,\widehat{\otimes}_{\mathbb R}\,\mathcal O_B(B),
$
and $[\sigma]$ is represented by
$
\alpha\in \Be^{m-i}_c(F)\,\widehat{\otimes}_{\mathbb R}\,\mathcal O_B(B),
$
then
\begin{equation}\label{eq:local-pairing}
\PP_{\mani / \para} ( [\omega],[\sigma ]) |_{U}
\;=\;
\int_{F} \sigma \cdot \omega.
\end{equation}
By construction $\sigma \cdot\omega\in \Be^m_\mathpzc{c}(F)$, hence the integral is well-defined.
\end{remark}

The relative Poincar\'e pairing via Berezin fiber integration is precisely what appears in Grothendieck--Verdier duality. More precisely, the following holds.

\begin{theorem}\label{prop:pairing-is-GV-orientable}
Let $\varphi  \maps \mani \to \para$ be a family of supermanifolds of relative dimension $m|n$ such that $\mani$ and $\para$ are orientable.
Then the relative Poincar\'e pairing $\mbox{\sffamily{\emph{P}}}_{\mani / \para}$ of Definition~\ref{def:relative-Poincare-pairing-Be-Omega} coincides with the
pairing induced by Grothendieck--Verdier duality, \emph{i.e.}\ it identifies
\begin{equation}
\mathbf{R}^i\varphi_*\Omega^\bullet_{\mani/\para}
\ \cong\
\mathcal{H}\mbox{\emph{om}}_{\mathcal{O}_\para} \bigl(\mathbf{R}^{m-i}\varphi_!\Be^\bullet_{\mani/\para}, \mathcal{O}_\para\bigr ),
\end{equation}
and hence realizes the relative Poincar\'e duality isomorphisms on cohomology sheaves.
\end{theorem}

\begin{proof}
By Poincar\'e--Verdier duality in the orientable case, one has an isomorphism in $\D(\mathcal O_\para)$
\[
\mathbf{R}\mathcal H \mathrm{om}_{\mathcal O_\para}\!\bigl(\mathbf{R}\varphi_!\,\Be^\bullet_{\mani/\para}, \mathcal O_\para\bigr)
\;\cong\;
\mathbf{R}\varphi_*\Omega^\bullet_{\mani/\para}[m].
\]
Taking cohomology sheaves yields the claimed isomorphism. The corresponding evaluation pairing is obtained by composing
the natural product map with proper supports
\begin{equation}
\mathbf{R}^i\varphi_*\Omega^\bullet_{\mani/\para} \otimes_{\mathcal O_\para} \mathbf{R}^{m-i}\varphi_!\Be^{\bullet}_{\mani/\para}
\longrightarrow
\mathbf{R}^m\varphi_!\Be^\bullet_{\mani/\para},
\end{equation}
induced by the $\Omega^\bullet_{\mani/\para}$-module structure on $\Be^\bullet_{\mani/\para}$, with
the fiber integration map $\mathbf{R}^m\varphi_!\Be^\bullet_{\mani/\para}\to \mathcal O_\para$ coming from the counit of
the adjunction. This is precisely the pairing $\PP_{\mani / \para}.$ 
\end{proof}

Finally, it is worth specializing the above constructions over a simply connected base supermanifold $\para$, or a trivial family.

\begin{corollary}[Relative Poincar\'e pairing -- simply connected base]\label{cor:pairing-simply-connected-orientable}
Let $\varphi  \maps \mani \to \para$ be a family of supermanifolds and let $\para$ be simply connected.
Then the relative Poincar\'e pairing is the $\mathcal O_\para$-linear extension of the classical
fiberwise pairing 
\begin{equation} \label{PDpairF}
\xymatrix@R=1.5pt{
H^i_{\mathpzc{dR}}(\fib)\ \otimes\ H^{m-i}_{\mathpzc{Sp,c}}(\fib)\ \ar[r] &\ \mathbb R \\
([\omega] , [\sigma])\ar@{|->}[r] & \int_\fib \sigma \cdot \omega,
}
\end{equation}
where $\sigma \cdot \omega \in \Be^m_{\mathpzc{c}}(F)$ is integrated via (absolute) Berezin integral.\\
The same holds true for a trivial family $\varphi  \maps \fib \times \para \to \para$ without assuming simply-connectedness of $\para$.
\end{corollary}

\begin{proof}
Follows from Lemma \ref{cor:relPD-simply-connected} together with Definition~\ref{def:relative-Poincare-pairing-Be-Omega} and Remark \ref{rem:local-relative-pairing}.
\end{proof}

Notice that the above fiberwise pairing of Equation \eqref{PDpairF} is exactly the Poincar\'e pairing behind the (absolute) Poincar\'e duality of Corollary \ref{PDpairF}.

\subsection{Trivial families and even submanifolds}\label{subsec:trivial-families-even-submanifolds}

In subsection~\ref{sec:fiberint} we introduced Berezin fiber integration for a
family $\varphi:X\to S$ as the counit of the Grothendieck--Verdier adjunction
$R\varphi_!\dashv \varphi^!$, and we showed that, after identifying the dualizing complex, 
this counit may be viewed as a morphism \eqref{berfibint}.
In particular, fiber integration produces $\mathcal O_\para$-linear functionals on compactly supported
integral forms along the fibers, and it is the intrinsic avatar of ``integrating out the
odd directions'' in families.

The goal of the present section is to reinterpret this construction in a geometric way for
\emph{trivial families}. 
Specifically, let $\mani_{/ \para} = \fib \times \para_{/ \para}$ be a trivial family
with oriented, connected fiber of dimension $m|n$, and let $\iota \maps
\mani^\ev_{/ \para} \inclusion \mani_{/ \para}$ be an even subfamily. While we cannot
integrate an $m$-form $\omega \in \Gamma_{\mathpzc{c}} ( \mani, \Omega_{\mani / \para} )$ over
$\mani$, we can integrate its pullback $\iota^* \omega$ over $\mani^\ev_{/ \para}$.
Poincar\'e duality then gives us an integral form $Y_\iota \in \Ber^0_{\mani /
\para}(\mani)$, which is \emph{dual to the even submanifold}, in the sense that, for all $\omega$:
\begin{equation}
  \int_{\mani^\ev / \para} \iota^* \omega = \int_{\mani / \para} Y_\iota \cdot \omega .
\end{equation}
We will prove that the cohomology class $[Y_\iota]$ is independent of $\iota$ -- and indeed the map $\omega \mapsto \int_{\mani^\ev / \para} \iota^\ast \omega$ is fiber integration in disguise. 

\subsubsection*{\textbf{Interlude on projections}} Before we begin, let us introduce some additional notions related to underlying even manifolds, as defined in Definition \ref{def:uem}. 
\begin{defn}
  Let $\iota \maps \mani^\ev_{/ \para} \inclusion \mani_{/
  \para}$ be an even manifold, a map $p \maps \mani_{/ \para} \to \mani^\ev_{/ \para}$ is a
  \define{projection} with respect to $\iota$ if $p \circ \iota = 1_{\mani^\ev}$. If
  $\iota$ admits a projection, we say $\mani_{/ \para}$ is \define{projected} with
  respect to $\iota$.
\end{defn}

Although an even manifold is a choice, for a trivial family
$\mani_{/\para} = \fib \times \para_{/\para}$ there is a canonical choice, as given in Definition \ref{def:ican-trivial}, which we now recall
\begin{equation}
  \mani^\ev_{/\para} = \fib_\red \times \para_{/\para}, \quad \iota_\can \maps
  \fib_\red \times \para_{/\para} \inclusion \fib \times \para_{/\para} .
\end{equation}
Here, $\iota_\can = \iota_\red \times 1_\para$, where
$\iota_\red \maps \fib_\red \inclusion \fib$ is the inclusion of the reduced
submanifold of the fiber $\fib$. With slight abuse of notation, we call $\iota_\can$ the {canonical even
  manifold} of a trivial family.

\begin{prop}
  \label{prop:trivproj}
  Any trivial family $\mani_{/\para} = \fib \times \para_{/\para}$ is projected with respect to the canonical even manifold $\iota_\can \maps \fib_\red \times \para_{/\para} \inclusion \fib \times \para_{/\para}$.
\end{prop}

\begin{proof}
  Because $\fib$ is a smooth supermanifold, Batchelor's theorem \cite{Bat} says we
  can choose a projection $p_\red \maps \fib \to \fib_\red$ of the reduced
  submanifold of the fiber, \emph{i.e.}, a map such that $p_\red \circ \iota_\red =
  1_{\fib_\red}$. Thus $p = p_\red \times 1_\para$ is a projection with respect
  to $\iota_\can$.
\end{proof}

\noindent
Note that $p$, like $\iota$, is necessarily the identity on the underlying
topological space. We will need to know how projection with respect to a given
embedding composes with a different embedding.

\begin{lemma}
  \label{prop:changeemb}
  Suppose $p$ is a projection of $\mani_{/\para}$ with respect to the even
  manifold $\iota \maps \mani^\ev_{/\para} \inclusion \mani_{/\para}$, and let
  $\iota' \maps \mani^\ev_{/\para} \inclusion \mani_{/\para}$ be another choice of
  embedding. Then $\xi = p  \circ \iota' \maps \mani^\ev_{/\para} \to
  \mani^\ev_{/\para}$ is a diffeomorphism which is the identity on the underlying
  topological space.
\end{lemma}

\begin{proof}
  The map $\xi = p \circ \iota'$ is the identity on the underlying topological space
  because $p$ and $\iota'$ are. In particular, $\xi_\red$ is a diffeomorphism. Thus,
  to check that $\xi$ is a diffeomorphism, it suffices to check that it is a local
  diffeomorphism, see \cite{Leites}.
  Indeed, let $x^a$ be relative local coordinates for $\mani^\ev_{/\para}$. Then
  because $\xi$ is the identity on the underlying topological space, it has the form
  $\xi^* \maps x^a \mapsto x^a + f^a$ for $f^a$ even, nilpotent elements, hence $\xi$ has 
  a local inverse by the inverse. 
  \end{proof}

\subsubsection*{\textbf{Independence of embedding}}
\label{sec:independence}

Having introduced projections, we will now
prove that the integral of a closed, compactly support $m$-form over an even
manifold of a trivial family is independent of the choice of embedding of that
even manifold. First, we define this integration more precisely.

\begin{definition}[Integration on $\mani^{\ev}$]
  Let $\mani_{/\para} = \fib \times \para_{/\para}$ be a trivial family of
  supermanifolds of relative dimension $m|n$ such that $\fib$ is connected and oriented, and let $(\mani_{/ \para}^\mathpzc{ev},
  \iota)$ be an underlying even submanifold.\footnote{Note that we may take $\mani^\ev/\para =
  \fib_\red \times \para / \para$, where $\fib_\red$ is the reduced submanifold of
  the fiber, but we take the embedding $\iota$ to be arbitrary.} Then, for any $\omega \in
  \Gamma_\mathpzc{c} (\mani, \Omega^m_{\mani / \para} ) = \Omega^m_{\mathpzc{c}}(\mathpzc{F})
  \hat \otimes \mathcal{O}_\para (\para)$, the integral
 \begin{equation}
   \int_{\mani{/ \para}^\mathpzc{ev}}{\iota^*\omega}\in\mathcal{O}_\para(\para)
 \end{equation}
 is defined as the composition of the pull-back and the integration of an ordinary
 top-degree differential form on the smooth manifold $\mathpzc{F}_\red$.
 \begin{equation}
   \xymatrix@R=1.5pt{
     \Omega^m_{\mathpzc{c}}(\mathpzc{F})
  \hat \otimes \mathcal{O}_\para (\para)  \ar[rr]^{\iota^*} && \Omega_{\mathpzc{c}}^m (\fib_{\red}) \hat \otimes \mathcal{O}_\para(\para)  \ar[rr]^{\quad \int_{ \mathpzc{F}_\red} \otimes 1} && \mathcal{O}_\para (\para).
   }
 \end{equation}
 In particular, this induces a morphism of $\mathcal{O}_\para (\para)$-modules in
 cohomology
 \begin{equation}
   \label{eq:Phi}
   \xymatrix@R=1.5pt{
     \Phi_\iota : H^m_{\mathpzc{dR}, \mathpzc{c}} (\mathpzc{F}) \otimes \mathcal{O}_\para (\para) \ar[r]  &\mathcal{O}_\para (\para)  \\
     [\omega] \ar@{|->}[r] & \displaystyle \int_{\mani^{\mathpzc{ev}}{/\para}} \iota^* \omega.
   }
 \end{equation}
\end{definition}
Note the map $\Phi_\iota$ can be seen as coming from Poincar\'e--Verdier duality, or -- more precisely -- fiber integration. Indeed, consider equation \eqref{fibintgen} in the definition of fiber integration: resolving $\omega^\bullet_{\mani / \para}$ with relative differential forms, one get a morphism in $\D(\mathcal O_\para)$
\begin{equation}
\int_{\mani / \para } : \mathbf{R}\varphi_!\Omega^\bullet_{\mani/\para}[m]\longrightarrow \mathcal O_\para.
\end{equation}
Hence, taking cohomology, a morphism of sheaves of $\mathcal{O}_\para$-modules
$
\mathbf{R}^m \varphi_!\Omega^\bullet_{\mani/\para} \rightarrow \mathcal O_\para.
$
In the trivial case $\varphi \maps \fib \times \para \to \para$, one has a canonical identification
$\mathbf{R}^m\varphi_!\Omega^\bullet_{\mani / \para}\cong H^m_{\mathpzc{dR, c}}(\fib)\otimes_{\mathbb R}\mathcal O_\para$,
therefore taking global sections we obtain an $\mathcal O_\para(\para)$-linear map
\begin{equation} \label{phimap}
\Phi:\ H^m_{\mathpzc{dR,c}}(\fib)\otimes_{\mathbb R}\mathcal O_\para(\para)\ \cong\ \Gamma\!\bigl(\para , \mathbf{R}^m\varphi_!\Omega^\bullet_{\mani/\para}\bigr)
\ \xrightarrow{\ \Gamma(\para,\mathcal H^0(\int_{\mani / \para}))\ }\  \Gamma ( \para , O_\para ).
\end{equation}
This suggests that $\Phi_\iota$ 
must be independent of the chosen embedding $\iota$, as we shall prove.

\begin{lemma}
  \label{lem:indy}

  Let $\mani_{/\para} = \fib \times \para_{/\para}$ be a trivial family of
  supermanifolds with oriented fiber of dimension $m|n$. For any two embeddings of
  the underlying even submanifold,
  $\iota \maps \mani_{/ \para}^\mathpzc{ev}\hookrightarrow\mani_{/\para}$ and
  $\iota' \maps \mani_{/ \para}^\mathpzc{ev}\hookrightarrow\mani_{/\para}$, we have
  $\Phi_{\iota}=\Phi_{\iota'}$. In other words, the morphism
  $\Phi_{\iota}$ is independent of the choice of the embedding of the underlying
  even manifold $\mani_{/\para}^{\mathpzc{ev}}$ and any such choice defines the
  map
  \begin{equation}
    \mathbf{\Phi}_{\mani^\mathpzc{ev}/\para} : \Gamma\!\bigl(\para , \mathbf{R}^m\pr_{\para \, !} \Omega^\bullet_{\mani/\para}\bigr) \longrightarrow \Gamma (\para, \mathcal{O}_\para),
  \end{equation}
where $\pr_{\para} \maps \fib \times \para \to \para.$
\end{lemma}

\begin{proof}
  Fix a closed, compactly supported $m$-form $\omega \in
  \Omega^m_{\mani/\para,\c}(\mani)$. Since the family is trivial, we have the
  canonical embedding $\iota_\can \maps \mani^\ev_{/\para} \inclusion
  \mani_{/\para}$. It therefore suffices to prove, for any other embedding $\iota
  \maps \mani^\ev_{/\para} \inclusion \mani_{/\para}$, we have
  $\int_{\mani^\ev/\para} \iota^* \omega = \int_{\mani^\ev/\para} \iota_\can^*
  \omega$.

  By Proposition \ref{prop:trivproj}, we can choose a projection $p_\can \maps
  \mani_{/\para} \to \mani_{/\para}$ with respect to $\iota_\can$. This splits the
  short exact sequence
  \begin{equation}
    0 \longrightarrow K^\bullet \longrightarrow \Omega^\bullet_{\mani / \para, \c}(\mani) \stackrel{\iota_\can^*}{\longrightarrow} \Omega^\bullet_{\mani^\ev / \para, \c}(\mani^\ev) \longrightarrow 0 .
  \end{equation}
  Here, the subcomplex $K^\bullet$ is the kernel of $\iota_\can^*$, which is is acyclic. 
  We can thus decompose the closed
  $m$-form $\omega$ into $\omega = p^*_\can \omega_\ev + d \eta$, where
  $\omega_\ev \in \Omega^m_{\mani^\ev/\para, \c}(\mani^\ev)$ is closed and
  $d \eta \in K^m$ is exact. By construction,
  $\int_{\mani^\ev_{/\para}} \iota^*_\can \omega = \int_{\mani^\ev_{/ \para}} \omega_\ev$.
  On the other hand,
  \begin{eqnarray}
    \int_{\mani^\ev{/\para}} \iota^* \omega = \int_{\mani^\ev{/\para}} \left( \iota^* p_\can^* \omega_\ev + \iota^* d \eta \right)  = 
    \int_{\mani^\ev{/\para}} \xi^* \omega_\ev + \int_{\mani^\ev{/\para}} d \iota^* \eta  = 
    \int_{\mani^\ev{/\para}} \omega_\ev .
  \end{eqnarray}
  In the above, $\xi = p_\can  \iota$ is a diffeomorphism of
  $\mani^\ev_{/\para}$, by Prop. \ref{prop:changeemb}. Said diffeomorphism is
  orientation preserving because it is the identity on the underlying topological
  space, and it thus preserves the integral.
\end{proof}

Next, we would like to construct an integral form $Y_\iota \in
\Gamma ( \mani, \Ber^0_{\mani/\para} ) $ that is \emph{Poincar\'e dual} to the embedding $\iota \maps
\mani^\ev_{/\para} \inclusion \mani_{/\para}$, in the sense that integration over the
even submanifold is given by the Poincar\'e pairing with $Y_\iota$. That is, we want
\begin{equation}
  \int_{\mani^\ev{/\para}} {\iota^* \omega} = \int_{\mani{/\para}} Y_\iota \cdot
  \omega ,
\end{equation}
for all compactly supported $\omega$. Because the Poincar\'e pairing is injective,
such $Y_\iota$ is unique. But we need to show it exists.

Just like above, a good candidate is singled out by the Poincar\'e-Verdier dual side of the morphism
$\mathbf{R}^m \varphi_!\Omega^\bullet_{\mani/\para} \rightarrow \mathcal O_\para$ (that was related to the map $\Phi_{\iota}$, as we discussed). 
This is given by $\mathbf{R}^0 \varphi_* \Be^\bullet_{\mani / \para}$, that for a trivial family is isomorphic to $H^0_{\mathpzc{Sp}} (\fib) \otimes \mathcal{O}_\para$.
Taking global sections, one finds $H^0_{\mathpzc{Sp}} (\fib) \otimes \Gamma (\para , \mathcal{O}_\para),$ which suggests that the integral form that we are after 
is the avatar of $[1] \in H^0_{\mathpzc{dR}} (\fib)$ in Spencer cohomology, namely it is given by Equation \eqref{eq:generator} in Theorem
\ref{thm:PL} (see also Theorem \ref{thm:pairing-with-Y-global} above), and -- once again -- independent of the chosen embedding of 
the underlying even manifold. \\

To see this without appealing to known result, let $U_\lambda$ be a coordinate chart where our embedding takes the form
$\iota^* x^a = x^a$, $\iota^* \theta^\alpha = 0$ (note that these coordinates always
exist: see the discussion before equation \ref{cancoord}). In these coordinates, define the 
integral form
\begin{equation}
  \label{eq:canonical}
  Y_\lambda \defeq \ber^\c \otimes \, \theta^1 \theta^2 \cdots \theta^n \pd{x^1} \wedge
  \pd{x^2} \wedge \cdots \wedge \pd{x^m} ,
\end{equation}
where $\ber^c = [\partial_{\theta^1} \ldots \partial^{\theta_n}\otimes dx^1 \cdots dx^m]$ is the coordinate
Berezinian. It is easy to check that $Y_{\lambda}$ defines a globally defined cocycle, 
\emph{i.e.}\ $Y_\lambda = Y_{\iota} \in \Gamma (\fib , \Be^0_{\fib})$ -- in fact its class $[Y_\iota]\otimes 1_\para$ is the generator of $\Gamma (\para, \mathbf{R}^0 \pr_{\para *} \Be^\bullet_{\mani / \para} )  =  H^0_\mathpzc{Sp} (\fib)\otimes \Gamma (\para, \mathcal{O}_\para)$ -- 
and therefore a globally defined functional 
\begin{equation}
\Gamma_{\mathpzc{c}} (\mani, \Omega^m_{\mani / \para})\ni \omega \longmapsto \int_{\mani{/ \para}} Y_\iota \cdot \omega \in \Gamma (\para, \mathcal{O}_\para).
\end{equation}
Notice, though, that this global existence does not, by itself, imply independence of $[Y_\iota]$ from the choice of
the embedding $\iota$: we prove this in the following theorem.

\begin{theorem}
  \label{thm:dual}
  Let $\mani_{/\para} = \fib \times \para_{/\para}$ be a trivial family of
  supermanifolds of relative dimension $m|n$, with oriented fiber. If
  $(\mani_{/ \para}^\mathpzc{ev}, \iota)$ is an underlying even manifold with
  embedding $\iota \maps \mani_{/ \para}^\mathpzc{ev}\hookrightarrow\mani_{/\para}$,
  then there exists a unique integral form $Y_{\iota} \in
  \Ber^0_{\mani/\para}(\mani)$ such that
  \begin{equation} \label{eq:equalpair}
    \int_{\mani{/\para}} Y_\iota \cdot \omega  = \int_{\mani^\ev/ \para} {\iota^*\omega},
  \end{equation}
  for every compactly supported $m$-form $\omega \in
  \Omega^m_{\mani/\para,\c}(\mani)$. Moreover, $Y_\iota$ is closed, and its class
  $[Y_{\iota}] \in H^0_{\mathpzc{Sp}} (\fib) \otimes \Gamma (\para, \mathcal{O}_\para)$ is independent of the choice of embedding
  $\iota \maps \mani_{/ \para}^\mathpzc{ev}\hookrightarrow\mani_{/\para}$.
  \end{theorem}

\begin{proof}
  Cover $\mani_{/\para}$ with relative coordinate charts
  $\{U_\lambda \}_{\lambda \in \Lambda}$ adapted to the embedding $\iota$. That is,
  over each $U_\lambda$, we have $\iota^* x^a = x^a$, $\iota^* \theta^\alpha = 0$,
  where $x^a|\theta^\alpha$ is the coordinate system on this chart. Choose a
  partition of unity $\{ \varphi_\lambda \in \Oh(\mani)\}_{\lambda \in \Lambda}$
  subordinate to $\{U_\lambda\}_{\lambda \in \Lambda}$. For each $\lambda$, we write
  $U^{\ev}_\lambda$ for the underlying even submanifold of $U_\lambda$, which has
  coordinates $x^a$. The collection of charts $\{U^\ev_\lambda\}$ covers
  $\mani^\ev_{/\para}$.

It is an easy computation to verity that the integral form $Y_\lambda \in
  \Ber^0_{\mani/\para}(U_\lambda)$ given in coordinates by Eq. \ref{eq:canonical}
  satisfies the desired property over $U_\lambda$, for $\omega$ with compact support
  in $U_\lambda$:
  \begin{equation}
    \int_{U^{\ev}_\lambda} \iota^* \omega = \int_{U_\lambda} Y_\lambda \cdot \omega.
  \end{equation}
  Hence, we define $Y_\iota = \sum_{\lambda \in \Lambda} \varphi_\lambda Y_\lambda$,
  and note that, for all compactly supported $\omega$, we have
  \begin{eqnarray}
    \int_{\mani/\para} Y_\iota \cdot \omega = \sum_{\lambda \in \Lambda} \int_{U_\lambda} \varphi_\lambda Y_\lambda \cdot \omega 
                                             =  \sum_{\lambda \in \Lambda} \int_{U^\ev_\lambda} \iota^* (\varphi_\lambda \omega) 
                                            =  \int_{\mani^\ev/\para} \iota^* \omega ,
  \end{eqnarray}
  as desired. This establishes the existence of $Y_\iota$. Uniqueness follows from
  the injectivity of the Poincar\'e pairing.
  To check that $Y_\iota$ is closed, let $\alpha \in
  \Omega^{m-1}_{\mani/\para,\c}(\mani)$ be a compactly supported $(m-1)$-form. Then
  $Y_\iota \cdot \alpha$ is a compactly supported integral form in degree $m-1$.
  Stokes' theorem for integral forms says that
  \begin{equation}
    \int_{\mani/\para} \delta ( Y_\iota \cdot \alpha) = 0 .
  \end{equation}
  Expanding the integrand using the Leibniz rule, $\delta (Y_\iota \cdot \alpha) =
  \delta Y_\iota \cdot \alpha + Y_\iota \cdot d \alpha$, this says that
  $\int_{\mani/\para} \delta Y_\iota \cdot \alpha = -\int_{\mani/\para} Y_\iota \cdot
  d \alpha$. The right hand side of this equation vanishes by Stokes' theorem for
  differential forms, since $\int_{\mani/\para} Y_\iota \cdot d \alpha =
  \int_{\mani^\ev/\para} \iota^* d \alpha = 0$. Hence for all $\alpha$, we have:
  \begin{equation}
    \int_{\mani/\para} \delta Y_\iota \cdot \alpha = 0 .
  \end{equation}
  By the nondegeneracy of the Poincar\'e pairing, we conclude $\delta Y_\iota = 0$.

Finally, let $\iota,\iota'$ be two embeddings of the even manifold.
By Lemma~\ref{lem:indy}, the right-hand side of \eqref{eq:equalpair}
is the same for $\iota$ and for $\iota'$, hence
\begin{equation}
\int_{\mani/\para} (Y_\iota-Y_{\iota'})\cdot \omega \;=\; 0
\end{equation}
for all $\omega \in \Gamma_\mathpzc{c} (\mani, \Omega^m_{\mani / \para}). $ Again by non-degeneracy of the pairing, $Y_\iota-Y_{\iota'}$ is $\delta$-exact, so
$[Y_\iota]=[Y_{\iota'}]\in H^0_{\mathpzc{Sp}} (\fib) \otimes \Gamma (\para, \mathcal{O}_\para)$.
\end{proof}

\noindent
We call the integral 0-form $Y_\iota$ from the above theorem the \define{Poincar\'e
  dual} of the underlying even manifold
$\iota \maps \mani^\ev_{/\para} \inclusion \mani_{/\para}$. It plays a key role in
the next section, where we apply it to supergravity.

\section{Poincar\'e duality in supergravity}
\label{sec:supergeosupergra}

We now turn from mathematics, our beloved home, to the foreign land of physics, in
order to visit an exotic region known by the name of geometric supergravity. While
there, we will spend some of our mathematical capital to make some ideas in this
formalism more precise. Our techniques have value here because in geometric
supergravity, the Lagrangian is a form on a supermanifold, called superspace, while
the action is defined by integrating over an ordinary manifold, spacetime, embedded
in the supermanifold. Following ideas of Castellani, Catenacci, and Grassi
\cite{CCG}, we can thus define the action as an integral over the entire superspace
through the use of a Poincar\'e dual integral form called a `picture changing
operator.' This idea allowed Castellani \emph{et al.\ }to compare different
formulations of supergravity in 3d, but their work raises some questions which our
techniques may help to answer.

Our techniques enter because, in $d$-dimensional geometric supergravity, the
Lagrangian $\Lag(\rheo{\varphi})$ is a $d$-form on a trivial family of supermanifolds
$\mani_{/\para}$ of dimension $d|s$, depending on the fields $\rheo{\varphi}$ on
$\mani_{/\para}$. We would like to integrate this over spacetime to define an action,
but this differential form cannot be integrated over $\mani_{/\para}$. Because this
is a trivial family, however, we have the canonical embedding of the even
submanifold, $\iota_{\can} \maps \mani^{\ev}_{/\para} \inclusion \mani_{/\para}$, so
we can define the action of the supergravity theory as the integral of the pullback:
\begin{equation}
  S(\rheo{\varphi}) = \int_{\mani^{\ev}/\para} \iota^{*}_{\can} \L(\rheo{\varphi}) \in \Oh(\para) .
\end{equation}
As we have seen, there are many other embeddings besides the canonical one. A key
idea of geometric supergravity is that we should vary the action with respect to
these embeddings as well as with respect to the fields:
\begin{equation}
  S(\rheo{\varphi}, \iota) = \int_{\mani^{\ev}/\para} \iota^{*} \L(\rheo{\varphi}) .
\end{equation}

Castellani \emph{et al.\ }\cite{CCG} took this a step further, by first
transferring the dependence on the embedding $\iota$ to its Poincar\'e dual integral
form, as we do in Theorem \ref{thm:dual},
$S(\rheo{\varphi}, \iota) = \int_{\mani/\para} Y_{\iota} \cdot \L(\rheo{\varphi})$,
and then allowing the integral form, and hence the embedding, to depend on the
fields:
\begin{equation}
  \label{eq:mother}
  S(\rheo{\varphi}) = \int_{\mani / \para} \Y(\rheo{\varphi}) \cdot \L(\rheo{\varphi}) .
\end{equation}
An integral form $\Y(\rheo{\varphi})$ in the same cohomology class as
$Y_{\iota_{\can}}$ is called a \define{picture changing operator}. By choosing
different picture changing operators, one gets equivalent, but superficially
different, actions for the supergravity theory. Castellani \emph{et al.\ }\cite{CCG}
used this idea to argue that the component, geometric, and superspace formulations of
3d supergravity all have actions of the above form, Eq. \ref{eq:mother}. 
While the basic idea of rewriting the superspace action using integral forms goes back to
to~\cite{CCG}, our goal in this section is to put this technique on a
fully geometric and cohomological footing, so that the comparison of formulations becomes a
mathematically controlled statement rather than a formal manipulation. Concretely, we strengthen
and clarify many points that were left implicit in the physics literature. 

First, we provide a first rigorous definition picture changing operators, and treat them functorially in families. This provides a
canonical framework in which ``changing the PCO'' is a cohomological statement, and allows us to
state and prove general invariance results under the precise hypotheses that are needed.

Moreover, we work systematically with sheaves of fields (and their constraints), and we make explicit the
morphisms relating superspace, geometric, and component descriptions. In particular, we compare
the three formulations at the level of \emph{spaces of fields} (and not merely at the
level of their action functionals), which is essential for a precise identification of ``equivalent''
descriptions.

Crucially, we provide proofs for statements that are often used informally in the physics literature, including
the cohomological equivalence of standard and supersymmetric PCO constructions. Along the way,
we isolate the exact assumptions under which such equivalences hold and clarify some steps where
additional hypotheses are required.



Finally, we will give a framework, broader than geometric
supergravity, where picture changing operators can be used to prove an equivalence of
Lagrangian field theories defined over families of supermanifolds, and highlight how
3d supergravity fits into this framework.


\subsection{3d supergravity basics}
\label{sec:basics}

We will compare three different approaches to 3d, $\cN = 1$ supergravity: the
component, superspace, and geometric formulations. The component and superspace
approaches are the most well known. In the component approach, the fields are defined
on spacetime, an ordinary manifold, and their physical meaning is clear\footnote{At
  least to a physicist.}, but supersymmetry is not manifest. In the superspace
approach, the `superfields' are defined on an extension of spacetime to a
supermanifold called superspace. While their physical meaning is more obscure,
supersymmetry is implemented via diffeomorphisms of superspace and is geometric in
nature. The geometric formulation, also known as the rheonomic formulation, attempts
to find a happy medium between the component and superspace versions of supergravity,
in which supersymmetry is given by diffeomorphisms on superspace, while the fields on
superspace, also called superfields, correspond to the component fields on spacetime
in a one-to-one fashion.

Before we dig into the details of these approaches, a word on our setting. In this
section, we will work locally, considering only a single, contractible, relative
coordinate chart, but we will retain an arbitrary supermanifold $\para$ as the base
of our family. In other words, \define{spacetime} means the trivial family
$\R^{3} \times \para_{/\para}$ in this section, and \define{superspace} means
$\R^{3|2} \times \para_{/\para}$, where the odd dimension 2 is the dimension of a
relevant spinor representation. To save space, we will shorten our notation to
$\R^{3}_{/\para}$ and $\R^{3|2}_{/\para}$ for spacetime and superspace, respectively.

The geometric and superspace formulations of supergravity both use so-called
superfields: these are fields defined on superspace, $\R^{3|2}_{/\para}$. In the
cases of interest to us, such fields are sections of vector bundles over
$\R^{3|2}_{/\para}$, subject to a natural nondegeneracy condition. In other words,
\define{superfields} are sections of subsheaves of
$\Oh_{\R^{3|2} \times \para}$-modules. We will denote a sheaf of superfields with an
uppercase script letter, such as $\Geo$.

As we mentioned about, the point of superfield formulations is to provide a manifest,
geometric meaning to supersymmetry. So, we can expect Lie superalgebra of
supersymmetries, which we denote $\p$, to act on the space of superfields $\Geo$ in a
nice, geometric fashion, which will leave the equations of motion of supergravity
invariant.

However, the physical intuition for superfields is generally less clear. So, every
theory of supergravity using superfields comes with a second formulation, using
so-called component fields. In contrast to superfields, these are fields defined on
spacetime, $\R^3_{/\para}$. In cases of interest to us, such fields are sections of
vector bundles over $\R^3_{/\para}$, subject to a natural nondegeneracy condition. In
other words, \define{component fields} are sections of subsheaves of
$\Oh_{\R^3 \times \para}$-modules. We will denote a sheaf of component fields with an
uppercase calligraphic letter, such as $\GeoComp$.

Both $\Geo$ and $\GeoComp$ are sheaves on the same topological space,
$\R^3 \times |\para|$. To compare the two spaces of fields, we have a map of sheaves,
\begin{equation}
  \kappa \maps \Sup \to \SupComp .
\end{equation}
In general, $\kappa$ is a differential operator. It is never an isomorphism.

We want our superfield and component theories to be equivalent, however. In
particular, we want an isomorphism between the spaces of fields. In order to achieve
this, we impose \emph{constraints} on the superfields. For the geometric formulation,
these are the infamous `rheonomic constraints'. For the superspace formulation, these
are called `conventional constraints'. In either case, we end up with a subsheaf of
the space of superfields, $\Conv \subseteq \Geo$, consisting of the superfields which
satisfy the constraints.

The key property is that the constrained superfields should be isomorphic to the
component fields via the component map, $\kappa$. Putting all three sheaves together
into a commutative triangle, we see the general shape of superfield formalisms:
\begin{figure}[H]
  \[
    \begin{tikzcd}[column sep=0.3em]
      \mbox{superfields}  & \Geo \ar[drrrrrrr, "\kappa \quad" below] &&&&&&& \Conv \ar[lllllll, hook'] \ar[d, "\iso"]  & \mbox{on } \R^{3|2}_{/\para} , \\
      \phantom{\mbox{abcde}} \mbox{component fields}  &                        &&&&&&& \GeoComp & \mbox{on } \R^{3}_{/\para} .  \\
    \end{tikzcd}
  \]
  \caption{A commutative triangle of sheaves on $\R^3 \times |\para|$.}
  \label{fig:shape}
\end{figure}

\noindent
Both of the superfield formulations of 3d supergravity we will discuss, the geometric
and superspace formulations, will have a commutative triangle of sheaves like this
one. However, in the superspace formulation, we will need to upgrade from sheaves to
stacks in order to account for gauge equivalence. In the geometric formulation, this
is not needed, which suggests that it is somehow more rigid.

\begin{rem}[Notational convention]
  \label{rem:notation}
  As an aid, we indicate the sheaf where a given field $\varphi$ lives with some
  decoration on the symbol.
  \begin{itemize}
    \item $\varphi$, with no decoration, is a superfield, i.e., a global section of
          $\Geo$;
    \item $\rheo{\varphi}$ in boldface is a constrained superfield, i.e., a global
          section of $\Conv$;
    \item $\comp{\varphi}$ with a subscript 0 is a component field, i.e., a global
          section of $\GeoComp$.
  \end{itemize}
\end{rem}

To conclude, let us preview the three approaches for 3d, $\cN = 1$ supergravity we
shall describe. In the following table, we give the main features of the component,
geometric, and superspace formulations, respectively. In this table we will see the
main fields of all three formulations: the spin connection $\omega$, the coframe $e$,
and the gravitino $\psi$, all decorated as in Remark \ref{rem:notation}. Besides
these fields, the component and geometric formulations include auxiliary fields,
either a scalar, $\scal$, or a 2-form, $B$.
\begin{center}
  \renewcommand{\arraystretch}{1.2}
  \begin{tabular}{l|ccc}
    \hline
    \multicolumn{4}{|c|}{\textbf{\boldmath{3d $\cN = 1$ supergravity}}} \\
    \hline
    \hline
    Formulation & component & geometric & superspace \\
    \hline
    Model space & $\R^{3}$ & $\R^{3|2}$ & $\R^{3|2}$ \\
    Fields & $\comp{\omega}$, $\comp{e}$, $\comp{\psi}$ & $\rheo{\omega}$, $\rheo{e}$, $\rheo{\psi}$  & $\conv{\omega}$, $\conv{e}$, $\conv{\psi}$ \\
    Aux fields & $\comp{\scal}$ or $\comp{B}$ & $\rheo{B}$ & --- \\
    Constraints & --- & rheonomic & conventional \\
    Lagrangian $\Lag$ & differential 3-form & differential 3-form & integral 3-form \\
    PCO $\Y$ & --- & integral 0-form & --- \\[0.5ex]
    Action & $\displaystyle \int_{\R^3} \Lag $  & $\displaystyle \int_{\R^{3|2}} \Y \cdot \Lag$ & $\displaystyle \int_{\R^{3|2}} \Lag $ \\[2ex]
    \hline
  \end{tabular}
\end{center}

\subsection{The super-Poincar\'e algebra}
\label{sec:superalg}

To talk about supergravity, as with all supersymmetric theories in physics, we first
need to fix the supersymmetry algebra, a Lie superalgebra that plays an essential
role in all that follows. In our case, this is the 3d, $\cN = 1$ super-Poincar\'e
algebra, and it is best introduced in stages, beginning with the supersymmetric
analogue of infinitesimal translations. Spinors enter our story here, so we first
recall a tiny bit of representation theory.

\begin{defn}
  The \define{Lorentz algebra} $\so(2,1)$ is the Lie algebra preserving the
  Minkowski inner product on $\R^{2,1}$. We call $V \defeq \R^{2,1}$ the
  \define{vector representation} of $\so(2,1)$.
\end{defn}

\noindent
It is a famous fact that the Lorentz algebra $\so(2,1)$ is isomorphic to
$\mathfrak{sl}(2,\R)$, the Lie algebra of traceless $2 \times 2$ real matrices
\cite{Lie, DSpinors}. Taking this isomorphism as given, we get another representation of
$\so(2,1)$.

\begin{defn}
  The \define{spinor representation} $S$ of $\so(2,1)$ is the defining representation
  $\R^{2}$ of $\mathfrak{sl}(2,\R)$, viewed as an $\so(2,1)$ representation via the
  isomorphism $\so(2,1) \iso \mathfrak{sl}(2,\R)$.
\end{defn}

\noindent
The spinor representation $S$ of $\so(2,1)$ is quite famous in physics---in the
physics literature, elements of $S$ are known as 3d \define{Majorana spinors}
\cite{Figueroa}.

With these two representations of $\so(2,1)$ in hand, we introduce the translation
part of our supersymmetry algebra. For this step, we note one final fact: under the
isomorphism $\so(2,1) \iso \mathfrak{sl}(2,\R)$, the vector representation turns out
to be the symmetric square of the defining representation,
$V \iso \bigvee^{2} \R^{2}$ \cite[Ch 2]{DFSuper}. In other words, as $\so(2,1)$
representations, the vector representation is the symmetric square of the spinor
representation:
\begin{equation}
  \textstyle V \iso \bigvee^{2} S .
\end{equation}

\begin{defn}
  The 3d, $\cN = 1$ \define{supertranslation algebra} is the Lie superalgebra $\t$
  with even part given by the vector representation of $\so(2,1)$, and odd part given
  by the spinor representation:
  \begin{equation}
    \t_{0} \defeq V, \quad \t_{1} \defeq S .
  \end{equation}
  This becomes a Lie superalgebra when equipped with a Lie bracket:
  \begin{equation}
    [-,-] \maps \t \otimes \t \to \t .
  \end{equation}
  By definition, the only nonzero part of this Lie bracket takes a pair of spinors and produces a vector, using the symmetrizing map:
  \begin{equation}
    \renewcommand{\arraycolsep}{1pt}
    \begin{array}{cccc}
      [-,-] \maps & S \otimes S & \to & V, \\
                  & \psi \otimes \phi & \mapsto & \psi \vee \phi ,
    \end{array}
  \end{equation}
  where we have made use of the aforementioned isomorphism, $V \iso \bigvee^{2} S$.
\end{defn}

\noindent
Note that the above bracket does indeed make $\t$ into a Lie superalgebra: it
satisfies the Jacobi identity trivially, because all nested brackets vanish, and it
is graded-skew symmetric because it is symmetric on the spinors, which are odd.

We constructed the Lie superalgebra $\t$ in a way that makes the following fact
immediate.

\begin{prop}
  \label{prop:equivariant}
  The supertranslation algebra $\t$ is an $\so(2,1)$ representation, and the Lie
  bracket on $\t$ is $\so(2,1)$-equivariant.
\end{prop}

\noindent
For the aid of our physics readers, we also give the definition of $\t$ in terms of a
basis. Let $\{Q_{\alpha}\}$ be a basis for the spinors $S$, and $\{P_{a}\}$ a basis
for the vectors $V$. In terms of gamma matrices, the bracket on $\t$ is given by the
formula
\begin{equation}
  \label{eq:susycharges}
  [Q_{\alpha}, Q_{\beta}] = 2 \gamma^{a}_{\alpha \beta} P_{a} ,
\end{equation}
where $\gamma^{a}$ is a real 3d gamma matrix.

With this rather lengthy definition of the supertranslation algebra out of the way,
the full supersymmetry algebra is easy to define: it is the super-Poincar\'e algebra.

\begin{defn}
  The 3d, $\cN = 1$ \define{super-Poincar\'e algebra} $\p$ is defined as the
  semidirect product of the Lorentz algebra $\so(2,1)$ and the supertranslation
  algebra $\t$,
  \begin{equation}
    \p \defeq \so(2,1) \ltimes \t .
  \end{equation}
\end{defn}

\noindent
Note that this semidirect product makes sense: by construction, $\t$ is a
representation of $\so(2,1)$, and the Lie bracket on $\t$ is $\so(2,1)$-equivariant,
as noted in Proposition \ref{prop:equivariant}.

As the reader will have noted, this section is all about the representation theory of
the Lorentz algebra $\so(2,1)$. Before we move on, we will briefly make note of
several $\so(2,1)$-equivariant maps that will be useful later, particularly in
constructing Lagrangians for 3d supergravity.

First, note that the action of $\so(2,1)$ on a representation is an equivariant map,
where $\so(2,1)$ acts on itself via the adjoint action. We will use the vector and
spinor representations in this way later:
\begin{equation}
  \rho_{V} \maps \so(2,1) \otimes V \to V, \quad \rho_{S} \maps \so(2,1) \otimes S \to S.
\end{equation}

Second, both the vector and spinor representations are self-dual, and thus have an
invariant pairing. The invariant pairing on the vector representation is simply the
Minkowski metric, which we denote by $g$:
\begin{equation}
  g \maps V \otimes V \to \R .
\end{equation}
The spinors have an invariant, antisymmetric pairing, given by contracting with a
volume form on $S = \R^{2}$ preserved by $\mathfrak{sl}(2,\R) \iso \so(2,1)$. We
denote this pairing by $\ip{-,-}$ and call it the \define{spinor pairing}:
\begin{equation}
  \ip{-,-} \maps S \otimes S \to \R .
\end{equation}

Third, as we have already noted, the Lie bracket is equivariant:
\begin{equation}
  [-,-] \maps S \otimes S \to V .
\end{equation}

Fourth, we can use the invariant pairings to dualize $V$ and $S$ in the Lie bracket
and obtain a new equivariant map, called the \define{Clifford action}:
\begin{equation}
  \gamma \maps V \otimes S \to S .
\end{equation}
As befitting an action, we will write $\gamma(v) \psi$ for the image of
$v \otimes \psi \in V \otimes S$ under $\gamma$. Explicitly, we take $\gamma$ to be
the unique map such that the following identity holds:
\begin{equation}
  g(v, [\psi, \phi]) = 2 \ip{\psi, \gamma(v) \phi} ,
\end{equation}
for all $v \in V$, $\psi, \phi \in S$. The factor of 2 makes this definition
consistent with the $\gamma$ matrix in Eq. \ref{eq:susycharges}, so
$\gamma_a \defeq \gamma(P_a)$ really is a $\gamma$ matrix in the usual sense.

Fifth, it is a famous fact that the special orthogonal Lie algebra consists of
skew-symmetric matrices. Hence, there is an isomorphism between the adjoint
representation and the exterior square of the vector representation. We denote this
by $\sharp$ and call it the \define{musical morphism}, in honor of the fact that it
comes from using the Minkowski metric $g$ to ``raise an index'' on an element of
$\so(2,1)$:
\begin{equation}
  \textstyle (-)^\sharp \maps \so(2,1) \isoto \bigwedge^{2} V .
\end{equation}

Sixth, because $V$ has a volume form and a metric, we have the Hodge duality
isomorphism $V \iso \bigwedge^2 V$. Combining this with the musical morphism, we have
an isomorphism between the vector representation and the adjoint representation,
which we call the \define{Hodge star} and write as a superscript:
\begin{equation}
  (-)^\star \maps V \isoto \so(2,1) .
\end{equation}

Finally, the Lorentz algebra preserves the canonical volume form on $V$ coming from
the Minkowski metric and a choice of orientation:
\begin{equation}
  \textstyle \vol \maps \bigwedge^{3} V \to \R .
\end{equation}
The volume form is normalized so that, for $\{P_a\}$ an oriented orthonormal basis
of $V$, we have $\vol(P_0, P_1, P_2) = 1$. More generally, the components of $\vol$
in such a basis are given by the Levi-Civita alternating symbol,
$\vol(P_a, P_b, P_c) = \epsilon_{abc}$, with $\epsilon_{012} = 1$.

We collect these maps in the following table.

\begin{table}[H]
  \begin{tabular}{lr@{$\maps$}l}
    \hline
    Description & \multicolumn{2}{c}{Map} \\
    \hline
    Vector representation   & $\rho_{_{V}}$ & $\so(2,1) \otimes V \longrightarrow V$ \\
    Spinor representation   & $\rho_{_{S}}$ & $\so(2,1) \otimes S \longrightarrow S$ \\
    Minkowski inner product & $g$          & $V \otimes V \longrightarrow \R$ \\
    Spinor pairing          & $\ip{-,-}$   & $S \otimes S \longrightarrow \R$ \\
    Lie bracket             & $[-,-]$      & $S \otimes S \longrightarrow V$ \\
    Clifford action         & $\gamma$     & $V \otimes S \longrightarrow S$ \\
    Musical morphism        & $(-)^\sharp$     & $\so(2,1) \isoto \bigwedge^{2} V$ \\
    Hodge star              & $(-)^\star$   & $V \isoto \so(2,1)$ \\
    Volume form             & $\vol$       & $\bigwedge^{3} V \longrightarrow \R$ \\
    \hline
  \end{tabular}
  \caption{The $\so(2,1)$-equivariant maps.}
  \label{tab:maps}
\end{table}

\noindent
In the coming sections, we will use the super-Poincar\'e algebra in the construction
of 3d, $\cN = 1$ supergravity in the component, geometric and superspace formulations.

\subsection{3d supergravity -- the geometric formulation}
\label{sec:geom}

The component and superspace approaches to supergravity are the most well known. In
the component approach, the fields are defined on spacetime, an ordinary manifold,
and their physical meaning is typically clear, but supersymmetry is not manifest. In
the superspace approach, the fields are defined on an extension of spacetime to a
supermanifold called superspace. While their physical meaning is more obscure,
supersymmetry is implemented via diffeomorphisms of superspace and is geometric in
nature.

The key idea of the geometric formulation of supergravity is to find a happy medium
between the component and superspace formulations, in which supersymmetry is given by
diffeomorphisms on superspace, while the fields on superspace correspond to the
component fields on spacetime in a one-to-one fashion. This works via the so-called
`rheonomic parametrization': for any component field $\comp{\varphi}$ on spacetime
$\R^3_{/\para}$, the rheonomic parametrization yields a unique extension
$\rheo{\varphi}$ to superspace $\R^{3|2}_{/\para}$, with the property that
$\iota^{*}_{\can} \rheo{\varphi} = \comp{\varphi}$.

Finally, the reader may wonder why we confine ourselves to trivial families when
talking about supergravity, and why the canonical embedding $\iota_{\can}$ plays a
distinguished role. This is because, to the best of our understanding, physicists
working in superspace only work (albeit implicitly) with trivial families, and use
the canonical embedding to single out the actual, physical spacetime.

\subsubsection{\textbf{Fields}}
\label{sec:fields}

After these algebraic preliminaries, we are finally ready to introduce geometric
supergravity. We begin by introducing the fields.

\begin{defn}
  \label{def:fields}
  We define the \define{geometric superfields of 3d, $\cN = 1$ supergravity} to be
  the following:
  \begin{itemize}
    \item an $\so(2,1)$-valued 1-form
          $\omega \in \Omega^{1}(\R^{3|2}_{/\para}, \so(2,1))$, the
          \define{superfield spin connection};
    \item a vector-valued 1-form $e \in \Omega^{1}(\R^{3|2}_{/\para}, V)$,
          the \define{superfield coframe};
    \item a spinor-valued 1-form $\psi \in \Omega^{1}(\R^{3|2}_{/\para}, S)$,
          the \define{superfield gravitino};
    \item a 2-form $B \in \Omega^{2}(\R^{3|2}_{/\para})$, the \define{2-form
          gauge superfield}.
  \end{itemize}
  In addition, we require the $\t$-valued 1-form given by the sum
  $E = e + \psi$ to be nondegenerate in the sense that
  $E$ trivializes the tangent bundle. That is, we have an isomorphism of
  $\Oh(\R^{3|2} \times \para)$-modules:
  \begin{equation}
    E \maps \T \R^{3|2}_{/\para} \isoto \Oh(\R^{3|2} \times \para) \otimes \t,
  \end{equation}
  where $\T \R^{3|2}_{/\para}$ denotes the module of vertical vector fields on our
  family. This supertranslation-valued 1-form
  $E \in \Omega^{1}(\R^{3|2}_{/\para}, \t)$ is called the
  \define{supervielbein}, and plays an important role in the superspace formulation
  of supergravity we will introduce later on.
\end{defn}

\noindent
Because we are working in a single chart, we can work directly with forms on
superspace. A nontrivial topology on superspace would force us to think about the
principal bundle that encodes gauge transformations between coordinate patches, but
we consider this a distraction to our aims in this section.

We have defined our fields globally, but they naturally form a sheaf, which we denote
$\Geo_{\R^{3|2}/\para}$. For an open set $U \subseteq \R^{3|2} \times \para$, the
sections are:
\begin{equation}
  \Geo_{\R^{3|2}/\para}(U) = \{ (\omega, e, \psi, B) \mbox{ over } U \, : \,  E = e + \psi \mbox{ nondegenerate} \} .
\end{equation}
Moreover, because $\Geo_{\R^{3|2}/\para}$ is a sheaf of differential forms, we can
apply the operation of pullback to sections of this sheaf. For instance, given an
even submanifold $\iota \maps \R^3_{/\para} \inclusion \R^{3|2}_{/\para}$, we get a
map $\iota^*$ of sheaves:
\begin{equation}
  \iota^* \maps \Geo_{\R^{3|2}/\para} \to \GeoComp_{\R^3/\para} ,
\end{equation}
where $\GeoComp_{\R^3/\para}$ is the sheaf of `component fields' we will define
later; the latter is also a sheaf of differential forms, but now on spacetime,
$\R^3_{/\para}$. This operation of pullback will play a critical role when we
introduce the `rheonomic parametrization' in Section \ref{sec:constraints}.

This collection of four fields is fairly unwieldy, so one might want to organize them
in a way that makes their conceptual meaning clearer. This is possible within the
framework of Cartan geometry, which describes geometries by comparison with a
suitable homogeneous space \cite{Sharpe}. In our case, the super-Poincar\'e-valued
1-form $A \defeq \omega + e + \psi$ turns out to be the local
data for the Cartan connection of a Cartan geometry modeled on the homogeneous
supermanifold called \define{super-Minkowski spacetime},
\begin{equation}
  \mathbb{M}^{2,1|2} \defeq \mathcal{P}/\Spin(2,1) .
\end{equation}
Here, $\mathcal{P}$ denotes the \define{super-Poincar\'e group}, the connected Lie
supergroup whose Lie superalgebra is the super-Poincar\'e algebra, $\p$, and
$\Spin(2,1) \iso \SL(2,\R)$ denotes the connected spin group. Super-Minkowski spacetime is the
supergeometric analogue of Minkowski spacetime $\mathbb{M}^{2,1}$, an analogy we can
make visible by identifying $\M^{2,1}$ with the homogeneous space $P / \SO(2,1)$, for
$P$ the ordinary Poincar\'e group.

One can take this idea of Cartan geometry even further by moving to what is called
higher Cartan geometry: the 1-form \emph{and} 2-form fields
$\A = \omega + e + \psi + B$ turn out to be the local data for a `higher Cartan
connection' valued in an `$L_{\infty}$-algebra'. It turns out that there is an
$L_{\infty}$-algebra that extends the super-Poincar\'e algebra $\p$. This
$L_\infty$-algebra, $\sugra$, is defined on a cochain complex with two terms:
\begin{equation}
  \R \stackrel{d}{\longrightarrow} \p .
\end{equation}
Here, $\R$ is in degree $-1$ and $\p$ is in degree 0. We regard $\A$ as a
$\sugra$-valued form of total degree 1, which foreshadows the fact that $\A$ is a
kind of connection. For much more, see the article of Sati, Schreiber, and Stasheff
on $L_{\infty}$-algebra-valued connections \cite{Sati:2009}.

In the physics literature, the $L_{\infty}$-algebra $\sugra$ is presented in terms of
a so-called `free differential algebra'. Translating between the physics perspective,
where many examples have been computed, and the mathematical perspective, where
powerful theorems connect $L_{\infty}$-algebras to other topics, has already been
quite fruitful, and much remains to be done! For more on this topic from the physics
side, see the textbook of Castellani, D'Auria, and Fr\'e \cite{CDF}. From the
mathematical side, see the work of Urs Schreiber and collaborators, starting with the
$n$Lab entry on the D'Auria--Fr\'e formulation of supergravity \cite{nLab}.

Of course, in order to apply Cartan geometry or higher Cartan geometry to our
setting, one should spell out how to adapt these frameworks to work with families of
supermanifolds. For Cartan geometry, this was done in Eder's PhD thesis \cite{Eder}.
For higher Cartan geometry, parts of the formalism have been developed in a high
level of generality in the setting of higher topos theory \cite{Schreiber} and modal
homotopy type theory \cite{Cherubini}, though we are unaware of a treatment that
specifically works with families. We leave that to future work.

\subsubsection{\textbf{Interlude on the supervielbein and related fields}}
\label{sec:interlude}

In what follows, including in the Lagrangian we introduce in the next section, we
will need a few facts about supervielbein $E$, its inverse $E^{-1}$, and some related
fields. First recall that the supervielbein $E = e + \psi$ is the $\t$-valued 1-form
obtained from the sum of the superfield coframe $e$ and the gravitino $\psi$.
Regarded as a map of $\Oh(\R^{3|2} \times \para)$-modules, we demand that it be an
isomorphism:
\begin{equation}
  E \maps \T \R^{3|2}_{/\para} \isoto \Oh(\R^{3|2} \times \para) \otimes \t,
\end{equation}
where $\T \R^{3|2}_{/\para}$ denotes the module of vertical vector fields on our
family.

So, we want to work with a map, $E$, linear over the supercommutative algebra
$\Oh(\R^{3|2} \times \para)$. Before getting to specifics, we recall a few
conventions to minimize signs when working with such maps.

Let $A$ be a supercommutative ring, and $M$ a left $A$-module. Because $A$ is
supercommutative, we can regard $M$ as a right $A$-module, where the right action of
$a \in A$ on $m \in M$ is defined in terms of the left as follows:
\begin{equation}
  ma \defeq (-1)^{|a||m|} am .
\end{equation}
Here $|a|, \, |m| \in \Z_2$ denote the parity of the homogeneous elements $a$ and
$m$. Thus every left module is naturally also a right module, and in fact the
category of left modules is equivalent to that of right modules. Since it involves a
sign, however, we will try to avoid using this equivalence, and thereby avoid
introducing the signs. The cost of doing so will be that, at times, we will need to
distinguish carefully between left and right modules. Here is our convention:

\begin{convention}
  In any bilinear map pairing a left module with a right module, the element of the
  left module should appear on the left and the right module on the right.
\end{convention}

\noindent
An example of this convention can be seen in the pairing between vector fields and
1-forms, Eq.\ \ref{eq:pairing}, which we recall now:
\begin{eqnarray}
\xymatrix{
( - , - ) \maps \mathcal{T}_{\mani / \para} \otimes_{\mathcal{O}_{\mani}} \Omega^1_{\mani / \para} \ar[r] & \mathcal{O}_{\mani},
}
\end{eqnarray}
Here, we view the tangent bundle $\T_{\mani / \para}$ as a left $\Oh_\mani$-module,
and its dual $\Omega^1_{\mani / \para}$ as a right $\Oh_\mani$-module. Doing so
reduces signs when pulling out coefficients, since we have,
\begin{equation}
  (fv, \omega g) = f(v, \omega)g ,
\end{equation}
for all vector fields $v$, 1-forms $\omega$, and structure sheaf sections $f$ and
$g$.

We will pursue this convention relentlessly, even going so far as to modify the order
in which we write function evaluation and composition. Specifically, for the map $E$
of interest to us now, we denote the value of $E$ on a vector field $v$ by $vE$.
Moreover, in composing $E$ with its inverse $E^{-1}$, $EE^{-1}$ denotes the operation
of first applying $E$ and then $E^{-1}$, and the reverse for $E^{-1} E$, contrary to
the more usual order.

With these conventions laid out, let us analyze $E$ and $E^{-1}$ more carefully.
While we have ruthlessly avoided a choice of basis, such a choice serves us well now.
So, let $\{P_a\}_{a \in \{0,1,2\}}$ be a basis for the vectors $V$, and
$\{Q_{\alpha}\}_{\alpha \in \{+,-\}}$ a basis for the spinors $S$. We use upper case
Latin indices $A, B$ for indices that can run over both the vector and spinor
indices, $A, B \in \{0,1,2, +, -\}$, and write $\{u_A\}_{A \in \{0,1,2,+,-\}}$ for
the basis of $\t = V \oplus S$ such that $u_a = P_a$ and $u_\alpha = Q_\alpha$. We
will also need the dual basis $\{u^A\}$ of $\t^*$, which is the unique basis such
that $(u_A, u^B) = \delta^B_A$. Note that this duality relation implies that
$\{u^a\}$ and $\{u^\alpha\}$ form the bases of $V^*$ and $S^*$ dual to $\{P_a\}$ and
$\{Q_\alpha\}$, respectively, so we also write $u^a = P^a$ and $u^\alpha = Q^\alpha$.

Since $E$ is $\t$-valued, we can find 1-forms $E^A$ such that $E = E^A u_A$. This
means that for any vector field $v \in \T \R^{3|2}_{/\para}$, we have:
\begin{equation}
  vE = (v, E^A) u_A ,
\end{equation}
where as usual $(-,-)$ denotes the canonical pairing between vector fields and
1-forms. For the inverse, we can find vector fields $E_A$ such that
$E^{-1} = u^A E_A$. This means that for any element $X \in \t$ in the
supertranslation algebra, we have:
\begin{equation}
  X E^{-1} = (X, u^A) E_A ,
\end{equation}
where we extend this to $\Oh(\R^{3|2} \times \para) \otimes \t$ by linearity, and now
$(-,-) \maps \t \otimes \t^* \to \R$ denotes the canonical pairing between $\t$ and
its dual $\t^*$. It is a quick computation to check that $E^{-1}$ is the inverse of
$E$ if and only if the vector fields $\{E_A\}$ and 1-forms $\{E^A\}$ constitute dual
bases of $\T\R^{3|2}_{/\para}$ and $\Omega^1(\R^{3|2}_{/\para})$, respectively:
\begin{equation}
  (E_A, E^B) = \delta^B_A .
\end{equation}

As we know, the supervielbein is the sum $E = e + \psi$ of the superfield coframe $e$
and the superfield gravitino $\psi$. In terms of components, this means:
\begin{equation}
  e^a = E^a, \quad \psi^\alpha = E^\alpha .
\end{equation}
In parallel, we shall define the following distinguished vector fields:
\begin{equation}
  e_a \defeq E_a, \quad \psi_\alpha \defeq E_\alpha .
\end{equation}
Of course, we can use the Minkowski metric and spinor pairing to raise the indices of
$e_a$ and $\psi_\alpha$, respectively, and thus we define:
\begin{itemize}
  \item a $V$-valued vector field, $e^{-1} \defeq g^{ab} P_a e_b $, the \define{superfield
        frame};
  \item an $S$-valued vector field,
        $\psi^{-1} \defeq \epsilon^{\alpha \beta} Q_\alpha \psi_\beta$, the
        \define{inverse superfield gravitino}; here
        $\epsilon^{\alpha \beta} = \ip{Q^\alpha, Q^\beta}$ are the components of the
        inverse spinor pairing, and the order matters, since $\ip{-,-}$ is
        antisymmetric.
\end{itemize}

Both $e^{-1}$ and $\psi^{-1}$ will play a role in what follows. Specifically, we need
to build a pair of 3-vector fields from them; that is, we want to combine $e^{-1}$
and $\psi^{-1}$ into certain sections of $\bigwedge^3 \T\R^{3|2}_{/\para}$. We give
these 3-vector fields in the following definitions.

\begin{defn}
  Apply the volume form $\vol \maps \bigwedge^3 V \to \R$ to the
  $\bigwedge^{3} V$-valued 3-vector $e^{-1} \wedge e^{-1} \wedge e^{-1}$ to obtain
  the \define{volume 3-vector}:
  \begin{equation}
    \vol(e^{-1}) \defeq \vol(e^{-1} \wedge e^{-1} \wedge e^{-1}) \in \bigwedge^3 \T \R^{3|2}_{/\para}.
  \end{equation}
\end{defn}

To clarify the volume 3-vector, let us expand it in a basis. Recall that when
$\{P_a\}$ is an oriented orthonormal basis of $V$, the volume form has components
given by the Levi-Civita alternating symbol, $\vol(P_a, P_b, P_c) = \epsilon_{abc}$,
where $\epsilon_{012} = 1$. Thus the volume 3-vector becomes
\begin{equation}
  \vol(e^{-1}) = \frac{1}{3!} \epsilon^{abc} e_a \wedge e_b \wedge e_c = - e_0 \wedge e_1 \wedge e_2.
\end{equation}
The minus sign is due to the fact that $\epsilon^{012} = -1$, thanks to raising the
indices with the Minkowski metric.

\begin{defn}
  Use the bracket to produce a $V$-valued 2-vector $[\psi^{-1}, \psi^{-1}]$, and then
  contract with $e^{-1}$ using the metric to obtain the \define{supersymmetric 3-vector}:
  \begin{equation}
    \textstyle \Theta \defeq g(e^{-1}, [\psi^{-1}, \psi^{-1}]) \in \bigwedge^3 \T\R^{3|2}_{/\para} .
  \end{equation}
\end{defn}

Again, let us expand in terms of a basis. Recall that
$[Q_\alpha, Q_\beta] = 2\gamma^a_{\alpha\beta} P_a$, where $\gamma^a$ is the 3d real
gamma matrix. Using the spinor pairing to raise the spinor indices, we see that
$\Theta$ has the form:
\begin{equation}
  \Theta = 2 \gamma^{a \alpha \beta} e_a \wedge \psi_\alpha \wedge \psi_\beta.
\end{equation}

\subsubsection{\textbf{The Lagrangian}}
\label{sec:lag}

Geometric supergravity is a Lagrangian field theory, whose equations of motion are
PDE coming from a variational principle. As we have mentioned, the Lagrangian will be
a 3-form on superspace. It will be expressed using the curvatures of the geometric
superfields, which we now define.

\begin{defn}
  \label{def:curvatures}
  By the \define{curvatures} of the geometric superfields, we mean the following
  differential forms:
  \begin{itemize}
    \item The \define{spin curvature} of the superfield spin connection,
          $\curv \defeq d\omega + \frac{1}{2} [\omega, \omega]$, which is an
          $\so(2,1)$-valued 2-form;
    \item the \define{torsion} of the coframe,
          $T \defeq d_{\omega} e + \frac{1}{2}[\psi,\psi]$, which is a
          vector-valued 2-form;
    \item the \define{gravitino field-strength}, $\rho \defeq d_{\omega} \psi$, which
          is a spinor-valued 2-form;
    \item the \define{3-form curvature}, $H \defeq dB - \frac{1}{2} g(e, [\psi, \psi])$, which is
          a 3-form.
  \end{itemize}
  In the above equations, we have adopted the following conventions:
  \begin{itemize}
    \item $d_{\omega}$ denotes the covariant exterior derivative built from the spin
          connection $\omega$. For a $p$-form $\alpha$ valued in a representation $U$
          of $\so(2,1)$, $d_{\omega} \alpha$ is the $U$-valued $(p+1)$-form given by
          the familiar formula
          \begin{equation}
            d_{\omega} \alpha \defeq d \alpha + \omega \wedge \alpha ,
          \end{equation}
          where $d$ is the usual exterior derivative, and the wedge product involves
          the wedge of forms and the action of $\so(2,1)$ on $U$;
    \item in $[\psi, \psi]$, we bracket the spinor values of $\psi$ and wedge the
          forms, to produce a vector-valued 2-form;
    \item in $g(e, [\psi, \psi])$, we apply the Minkowski inner product $g$ to the
          wedge of the vector-valued 1-form $e$ and 2-form $[\psi,\psi]$ to produce a
          real-valued 3-form.
 \end{itemize}
\end{defn}

The origin of these formulas may seem obscure, but it is clarified if we combine the
1-form fields into a Cartan connection given by the super-Poincar\'e-valued 1-form
$A = \omega + e + \psi$, in which case the super-Poincar\'e-valued 2-form
$F = \curv + T + \rho$ is the curvature 2-form of this Cartan connection. In
fact, it is given by the familiar formula,
\begin{equation}
  F = dA + \frac{1}{2} [A,A],
\end{equation}
as the reader can check. Here the bracket $[A,A]$ of forms now uses the Lie bracket
on the super-Poincar\'e algebra, $\p$.

Generalizing to a higher Cartan connection, we could derive the expression for the
3-form curvature $H$ as well. Alas, the precise details of how this works are beyond
the scope of this article. For more, see Sati \emph{et al.\ }\cite{Sati:2009}.

With the curvatures of the geometric superfields in hand, we are ready to write down
our Lagrangian. For this definition, recall from Section \ref{sec:interlude} that we
write $e^{-1}$ for the $V$-valued vector field coming from the inverse
supervielbein $E^{-1}$:
\begin{equation}
  e^{-1} \defeq g^{ab} P_a E_b .
\end{equation}

\begin{defn}
  \label{def:Lagrangian}
  The \define{Lagrangian of 3d, $\cN = 1$ geometric supergravity} is the following
  3-form on $\R^{3|2}_{/\para}$:
  \begin{equation}
    \Lag(\omega, e, \psi, B) \defeq \vol(\curv^{\sharp} \wedge e) + \langle \psi, d_{\omega} \psi \rangle + \alpha (fH - \frac{1}{2} f^{2} \vol (e \wedge e \wedge e)) .
  \end{equation}
  Here, $\alpha \in \R$ is a constant that will be fixed later. In this Lagrangian,
  we have made ample use of the $\so(2,1)$-equivariant maps from Table
  \ref{tab:maps}, as we now describe:
  \begin{itemize}
    \item $\curv^{\sharp}$ denotes the $\bigwedge^{2}V$-valued 2-form that comes from
          applying the musical morphism to the spin curvature $\curv$;
    \item both $\curv^{\sharp} \wedge e$ and $e \wedge e \wedge e$ are
          $\bigwedge^{3} V$-valued 3-forms, and we apply the volume form on $V$ to
          make them real-valued;
    \item we apply the spinor pairing to the wedge of the spinor-valued 1-form $\psi$
          and the spinor-valued 2-form $d_{\omega} \psi$ to obtain a real-valued
          3-form $\ip{\psi, d_{\omega} \psi}$.
  \end{itemize}
  Finally, the 0-form $f$, called the \define{geometric scalar}, is shorthand for the
  component of $H$ proportional to $\vol(e \wedge e \wedge e)$; more precisely:
  \begin{itemize}
    \item $f \defeq (\vol(e^{-1}), H)$, where $\vol(e^{-1})$ is the volume 3-vector
          of the superfield frame $e^{-1}$, and $(-,-)$ is the duality pairing
          between 3-vectors and 3-forms. Fixing $\{P_a\}$ an oriented orthonormal
          basis of $V$, this reads $f = -(e_0 \wedge e_1 \wedge e_2, H)$, since we
          checked in Section \ref{sec:interlude} that
          $\vol(e^{-1}) = -e_0 \wedge e_1 \wedge e_2$ in such a basis.
  \end{itemize}
\end{defn}

\noindent
We can regard the Lagrangian as a map of sheaves:
\begin{equation}
  \Lag \maps \Geo_{\R^{3|2}/\para} \to \Omega^{3}_{\R^{3|2}/\para} .
\end{equation}
Similarly, the geometric scalar is a map of sheaves:
\begin{equation}
  f \maps \Geo_{\R^{3|2}/\para} \to \Omega^0_{\R^{3|2}/\para}
\end{equation}

\subsubsection{\textbf{The rheonomic constraints}}
\label{sec:constraints}

As we mentioned earlier, the fields of geometric supergravity should be in one-to-one
correspondence with the component fields on spacetime. Indeed, we define the
`component fields' of 3d, $\cN = 1$ supergravity exactly as above, but on
$\R^{3}_{/\para}$. We distinguish these fields from their geometric counterparts by a
subscript 0.

\begin{defn}
  \label{def:comp}
The \define{geometric component fields of 3d, $\cN = 1$ supergravity} are
\begin{itemize}
  \item an $\so(2,1)$-valued 1-form
        $\comp{\omega} \in \Omega^{1}(\R^{3}_{/\para}, \so(2,1))$,
        the \define{spin connection};
  \item a vector-valued 1-form, $\comp{e} \in \Omega^{1}(\R^{3}_{/\para}, V)$,
        the \define{coframe};
        \item a spinor-valued 1-form $\comp{\psi} \in \Omega^{1}(\R^{3}_{/\para}, S)$,
        the \define{gravitino};
  \item a 2-form $\comp{B} \in \Omega^{2}(\R^{3}_{/\para})$, the \define{2-form
        gauge field};
\end{itemize}
In addition, we require that the coframe field $\comp{e}$ is nondegenerate in the
sense that it trivializes the tangent bundle of spacetime. In other words, we have an
isomorphism of $\Oh(\R^{3} \times \para)$-modules:
\begin{equation}
  \comp{e} \maps \T \R^{3}_{/\para} \isoto \Oh(\R^{3} \times \para) \otimes V .
\end{equation}
The geometric component fields of supergravity naturally form a sheaf, $\GeoComp_{\R^{3}/\para}$.
\end{defn}

\noindent
Although we would like geometric superfields to correspond in a one-to-one fashion
with the geometric component fields, they do not---the space of geometric superfields
is much larger. More precisely, thanks to the canonical even submanifold
$\iota_{\can} \maps \R^{3}_{/\para} \inclusion \R^{3|2}_{/\para}$, we have a map of
sheaves given by the pullback of differential forms,
$\iota^{*}_{\can} \maps \Geo_{\R^{3|2}/\para} \to \GeoComp_{\R^{3}/\para}$. This map
is far from being one-to-one, as the next example demonstrates.

\begin{example}
  Let us compare the superfield spin connection with the spin connection. More
  precisely, let us find the kernel of the map
  \begin{equation}
    \label{eq:wantkernel}
    \iota^*_\can \maps \Omega^1(\R^{3|2}_{/\para}, \so(2,1)) \to \Omega^1(\R^3_{/\para}, \so(2,1)) .
  \end{equation}
  To compute this, let us write $A = \Oh(\R^3 \times \para)$ for the superalgebra of
  smooth functions on spacetime, $\R^3_{/\para}$. Fixing coordinates
  $x,y,z|\eta,\theta$ of superspace $\R^{3|2}_{/\para}$, we have the superalgebra
  isomorphism:
  \begin{equation}
    \Oh(\R^{3|2} \times \para) \iso A[\eta, \theta] .
  \end{equation}
  Under this identification, the pullback
  $\iota^*_\can \maps \Oh(\R^{3|2} \times \para) \to \Oh(\R^3 \times \para)$ becomes
  the quotient map $A[\eta, \theta] \to A$ that mods out by the ideal
  $J \defeq (\eta, \theta)$. Moreover, as $A$-modules,
  \begin{eqnarray}
    \Omega^1(\R^{3|2}, \so(2,1)) & \iso & A[\eta, \theta] \langle dx,dy,dz | d\eta, d\theta \rangle \otimes \so(2,1) , \\
    \Omega^1(\R^{3}, \so(2,1)) & \iso & A \langle dx,dy,dz \rangle \otimes \so(2,1) .
  \end{eqnarray}
  The pullback $\iota^*_\can$ in Eq. \ref{eq:wantkernel} annihilates
  $\eta, \theta, d\eta$, and $d\theta$. Hence, we conclude its kernel is the
  $A$-submodule:
  \begin{equation}
    \ker(\iota^*_\can) \iso J\langle d\eta, d\theta \rangle \otimes \so(2,1) .
  \end{equation}
\end{example}

As this example shows, to have any hope at all of identifying the geometric fields
with the component fields, we must impose \emph{constraints}. Indeed, all
supersymmetric field theories defined in terms of superfields impose some system of
constraints in order to make the space of superfields smaller and match it to some
related space of component fields. In geometric supergravity, these constraints are
called the `rheonomic constraints'.

\begin{defn}
  \label{def:rheonomic}
  A geometric superfield $\varphi = (\omega, e, \psi, B)$ is \define{rheonomic} if it
  satisfies the following equations, the \define{rheonomic constraints}:
  \begin{eqnarray}
    d\Lag(\varphi) & = & 0 , \\
    T(\varphi) & = & 0 , \\
    H(\varphi) & = & f(\varphi) \vol(e \wedge e \wedge e) , \\ \label{eq:H}
    \rho(\varphi) & = & \rho_{ab} \, e^a \wedge e^b + \frac{3}{2} f(\varphi) \, \gamma(e) \wedge \psi \label{eq:rho} .
  \end{eqnarray}
  Here, $f(\varphi)$ is the geometric scalar. The rheonomic superfields naturally
  form a subsheaf, $\Rheo \subseteq \Geo_{\R^{3|2}/\para}$. We write sections of this
  subsheaf in boldface, $\rheo{\varphi} \in \Rheo(U)$.
\end{defn}

\noindent
These equations differ from their counterparts in Castellani \emph{et al.\
}\cite{CCG}, but they are equivalent thanks to the closure of the Lagrangian, which
those authors also impose. Indeed, requiring the Lagrangian to be closed is standard
in the geometric approach to any supersymmetric theory that includes auxiliary fields
\cite{CDF}. For us, the 2-form $B$ plays such a role.

The attentive reader may complain that the above constraint equations depend on the
constant $\alpha \in \R$ that appears in the definition of the Lagrangian. In fact,
the constraints imply $\alpha = 6$.

\begin{prop}
  Unless $\alpha = 6$, the rheonomic constraints have no solution.
\end{prop}

\begin{proof}
  The vanishing of the $e \wedge e \wedge e \wedge \psi$ components of
  $d\Lag(\varphi)$, together with Eq. \ref{eq:rho}, yield $\alpha = 6$.
\end{proof}

\noindent
Henceforth, of course, we will assume $\alpha = 6$. We next derive the standard form
of the rheonomic constraints using the Bianchi identities. Indeed, the equations in
the following proposition are Eqs. 3.11--3.17 from Castellani \emph{et al.\
}\cite{CCG}. We will need them later when we compare the geometric formulation of
supergravity to the superspace formulation.

\begin{prop}
  \label{prop:rheo}
  If $\varphi = (\omega, e, \psi, B)$ is a rheonomic superfield, then the following
  equations hold:
  \begin{eqnarray}
    T & = & 0, \label{eq:torsionless}\\
    H & = & f \vol(e \wedge e \wedge e), \label{eq:H_param}\\
    \rho & = & \rho_{ab}\,e^a \wedge e^b + \frac{3}{2} \gamma(e) \wedge \psi , \label{eq:rho_param}\\
    \curv & = & \curv_{ab} \, e^a \wedge e^b + \ip{\Theta_c , \psi} \, e^c + \frac{3}{2} f \, [\psi,\psi]^\star , \label{eq:Rab_param}\\
    df & = & (e_a f) \, e^a  + \ip{\psi, \Xi}. \label{eq:df_param}
  \end{eqnarray}
  Here, the coefficients $\Theta_c$ and $\Xi$ are $\so(2,1) \otimes S$- and
  $S$-valued 0-forms, respectively, and have the form:
  \begin{equation}
    \label{eq:theta_Xi_param}
    \ip{\Theta^{ab}{}_{c}, Q_\alpha}
    = 2 \ip{\rho^{[a}{}_{c}, \gamma^{b]} Q_\alpha} - \ip{\rho^{ab}, \gamma_cQ_\alpha},
    \qquad
    \Xi = -\frac{1}{3!} \,\epsilon_{abc} \gamma^a\rho^{bc}.
  \end{equation}
  In the above expressions, indices are raised and lowered by the Minkowski metric,
  $g$.
\end{prop}

\noindent
As noted, the proof uses the Bianchi identities, which we recall now:
\begin{align*}
\text{Bianchi of the torsion:}\qquad
& d_\omega T \;-\; \curv \wedge e \;-\; [\psi, \rho] \;=\;0,
\\[2pt]
\text{Bianchi of the gravitino field strength:}\qquad
& d_\omega \rho \;-\; \curv \wedge \psi \;=\;0,
\\[2pt]
\text{Bianchi of the spin curvature:}\qquad
&d_\omega \curv \;=\; 0,
\\[2pt]
\text{Bianchi of the 3-form curvature:}\qquad
&dH \;+\; \frac{1}{2} g(d_\omega e, [\psi,\psi]) \; + \; g(e, [\psi, \rho]) \;=\; 0.
\end{align*}

\begin{proof}
  We need only prove the last two equations, since the first three are rheonomic
  constraints. Indeed, the Bianchi of the 3-form curvature $H$ fixes the fermionic
  component of $df$ (\emph{i.e.}, $\Xi$ in $df= (e_af) e^a + \ip{\psi, \Xi}$), while
  the Bianchi of the gravitino field strength $\rho$, together with the $3d$ Clifford
  identities, fixes the mixed components of $\curv$, \emph{i.e.}, $\Theta_c$ in
  $\curv = \curv_{ab} \, e^a \wedge e^b + \ip{\Theta_c , \psi} \, e^c + \frac{3}{2} f \, [\psi,\psi]^\star$.
\end{proof}

\noindent
We have now seen a number of the standard ingredients of rheonomy from the physics
literature: the closure of the Lagrangian, constraint equations on the curvature
tensors, and the use of Bianchi identities. In the physics literature, one often
speaks of ``solving the Bianchi identities'', as if they were not identities. This is
because, in the physics context, they play the role of a consistency condition, as we
now remark upon.

\begin{remark}[On the role of Bianchi identities in the rheonomic
  approach]\label{rem:bianchi_role}
  Strictly speaking, if one starts from genuine geometric data (connections and
  potentials) and \emph{defines} the corresponding curvatures, then the Bianchi
  identities are automatic and do not need to be imposed as additional assumptions.

  In the rheonomic literature, however, one often proceeds in the opposite direction:
  one \emph{postulates} constraints on the curvature components in the basis
  $(e, \psi)$ (e.g., by declaring the mixed $e \wedge \psi$ and $\psi \wedge \psi$
  components to be covariant functions of a minimal set of ``inner'' fields) and only
  afterwards asks whether such a parametrization is actually realizable by underlying
  geometric superfields. In this viewpoint, the Bianchi identities play the role of
  \emph{integrability/consistency conditions} for the proposed constraints: they
  constrain and often uniquely determine the undetermined coefficients and mixed
  components (such as $\Theta$ and $\Xi$ in here), and ensure that the resulting
  expressions indeed come from bona fide curvatures.

  Accordingly, when we say that ``one has to solve the Bianchi identities'' in
  passing from a partial condition such as $dL=0$ to the full rheonomic constraints,
  what is meant is not that Bianchi identities are extra axioms, but rather that they
  are the necessary compatibility relations which \emph{complete} the reconstruction.
\end{remark}

\noindent
Finally, we have promised that the rheonomic constraint will give a correspondence
between geometric superfields and component fields. To accomplish this, however, we
need to impose a constraint on the component fields, coming from the torsion
$T = d_\omega e + \frac{1}{2} [\psi,\psi]$.

\begin{defn}
  The \define{supergravity torsion equation} is
  \begin{equation}
    d_{\comp{\omega}} \comp{e} + \frac{1}{2} [\comp{\psi}, \comp{\psi}] = 0
  \end{equation}
  for component fields $\comp{\omega}, \comp{e}, \comp{\psi}$. We write
  $\GeoComp_{\cl} \subseteq \GeoComp_{\R^3/\para}$ for the subsheaf of component
  fields satisfying the supergravity torsion equation, which we call the
  \define{classical subsheaf}.
\end{defn}

\noindent
The supergravity torsion equation is the equation of motion of the spin connection,
$\comp{\omega}$, in the component action we will introduce later. It says that spin
connection is not the Levi-Civita connection, but instead has torsion
$-\frac{1}{2} [\comp{\psi},\comp{\psi}]$ coming from the gravitino $\comp{\psi}$. By
standard techniques of differential geometry, we can then solve for the spin
connection in terms of the coframe and the gravitino, although we refrain from doing
so.

Moreover, if our component field is given by the pullback of a geometric superfield,
$\comp{\varphi} = \iota^*_\can \varphi$, then the supergravity torsion equation is
simply $i^*_\can T(\varphi) = 0$. Thus, the rheonomic subsheaf lands in the classical
subsheaf upon pullback, $\iota^*_\can(\Rheo) \subseteq \GeoComp_\cl$.

\begin{defn}
  We say that a subsheaf $\Rheo \subseteq \Geo_{\R^{3|2}/\para}$ defines a
  \define{rheonomic parametrization} if the restriction of the pullback to the
  canonical even submanifold
  \begin{equation}
    \iota^{*}_{\can} \maps \Rheo \to \GeoComp_{\cl} ,
  \end{equation}
  is an isomorphism of sheaves.
\end{defn}

\noindent
The above definition is simply the sheaf theoretic way of saying that, for any
classical component field $\comp{\varphi} \in \GeoComp_{\cl}(U)$, there is a unique
rheonomic superfield $\rheo{\varphi} \in \Rheo(U)$ such that
$\iota^*_\can \rheo{\varphi} = \comp{\varphi}$. The correspondence
$\comp{\varphi} \mapsto \rheo{\varphi}$ is what Csatellani--D'Auria--Fr\'e refer to
as the `rheonomic extension mapping'~\cite{CDF}.

\begin{prop}
  \label{prop:param}
  Let $\Rheo$ be the subsheaf of the geometric superfields $\Geo_{\R^{3|2}/\para}$
  defined by the rheonomic constraints, Def. \ref{def:rheonomic}. Then $\Rheo$
  defines a rheonomic parametrization. In other words, the restriction of the
  pullback
  \begin{equation}
    \iota^*_\can \maps \Rheo \to \GeoComp_\cl
  \end{equation}
  is an isomorphism of sheaves.
\end{prop}

\noindent
We will not use this result in what follows, so we only sketch how the proof should
work.

\begin{proof}[Sketch of proof.]
  We only argue that $\iota^*_\can$ is injective. Let
  $\rheo{\varphi} = (\rheo{\omega}, \rheo{e}, \rheo{\psi}, \rheo{B})$ be a rheonomic
  superfield, and $\comp{\varphi} = \iota^*_\can \rheo{\varphi}$ the corresponding
  component field. We want to show that $\comp{\varphi}$ determines $\rheo{\varphi}$.
  We do so by deriving a system of first-order PDE satisfied by $\rheo{\varphi}$,
  derived from the rheonomic constraints. We then note that $\comp{\varphi}$ provides
  the initial data to solve this system.

  Before writing our system of PDE, note that the torsion constraint $T = 0$ implies
  we can solve for $\rheo{\omega}$ in terms of $\rheo{e}$ and $\rheo{\psi}$, so we
  need only consider the triple $(\rheo{e}, \rheo{\psi}, \rheo{B})$.

  To obtain our PDE, fix a basis $\{P_a | Q_\alpha\}$ for the supertranslation
  algebra. Let $\rheo{e}_a$ and $\rheo{\psi}_\alpha$ denote the vector fields dual to
  the 1-forms $\rheo{e}^a$ and $\rheo{\psi}^\alpha$, respectively. We can use the
  constraint equations to obtain expressions for the covariant derivative of
  $\rheo{\varphi}$ along $\rheo{\psi}_\alpha$, as follows:
  \begin{eqnarray}
    i_{\rheo{\psi}_\alpha} d_{\rheo{\omega}} \rheo{e} & = &  \frac{1}{2} [\rheo{\psi}, Q_\alpha], \\
    i_{\rheo{\psi}_\alpha} d_{\rheo{\omega}} \rheo{\psi} & = & -\frac{3}{2} f \gamma(\rheo{e}) Q_\alpha, \\
    i_{\rheo{\psi}_\alpha} d \rheo{B} & = & -\frac{1}{2} [e, [\psi, Q_\alpha]] .
  \end{eqnarray}
  This system is consistent thanks to the torsion constraint, which relates the
  commutator $[\rheo{\psi}_\alpha, \rheo{\psi}_\beta]$ to the vector fields
  $\rheo{e}_a$, along with the Bianchi identities.

  By nondegeneracy of the supervielbein, the vector fields $\rheo{\psi}_\alpha$ span
  the odd part of the tangent space, transverse to the canonical even submanifold
  $\iota_\can \maps \R^3_{/\para} \inclusion \R^{3|2}_{/\para}$. By inspection, the
  above equations determine the derivatives of $\comp{\varphi}$ normal to the
  canonical even submanifold, and hence allow us to solve for $\rheo{\varphi}$ in
  terms of $\comp{\varphi}$.
\end{proof}

\noindent
Note that Prop. \ref{prop:param} implies that there is a commutative triangle of
sheaves:
\begin{equation}
  \begin{tikzcd}
    \Geo_{\R^{3|2}/\para} \ar[dr, "\iota^*_\can \quad" below] & \Rheo \ar[l, hook'] \ar[d, "\iso"] \\
    & \GeoComp_\cl \\
  \end{tikzcd}
\end{equation}
This is Fig. \ref{fig:shape} in the context of geometric supergravity.

\subsubsection{\textbf{Supersymmetry transformations}}

The supersymmetry transformations in geometric supergravity are intimately related to
the rheonomic constraints. To get a taste of this, let
$\varepsilon = \varepsilon^\alpha \rheo{E}_\alpha$ be a vector field along the odd
directions. For $\rheo{\varphi}$ a rheonomic superfield, and
$\comp{\varphi} = \iota^*_\can \rheo{\varphi}$ the corresponding component field, we
define a \define{supersymmetry transformation} along $\varepsilon$ to be:
\begin{equation}
  \delta_\varepsilon \comp{\varphi} \defeq \iota^*_\can( \Lie^\omega_\varepsilon \rheo{\varphi} ) ,
\end{equation}
where
$\Lie^\omega_\varepsilon \defeq d_\omega i_\varepsilon + i_\varepsilon d_\omega$ is
the \define{covariant Lie derivative}.

Let us see how this works for the coframe $\comp{e}$ and the gravitino $\comp{\psi}$.
For the coframe, we have
\begin{equation}
  \Lie^\omega_\varepsilon \rheo{e} = d_{\rheo{\omega}} i_\varepsilon \rheo{e} + i_\varepsilon d_{\rheo{\omega}} \rheo{e}.
\end{equation}
The first term vanishes because $i_\varepsilon \rheo{e} = 0$. The second term
becomes, via the torsion constraint,
$-\frac{1}{2} i_\varepsilon [\rheo{\psi}, \rheo{\psi}]$. Since
$i_\varepsilon \rheo{\psi} = \varepsilon^\alpha Q_\alpha$, we conclude
$\Lie_\varepsilon \rheo{e} = -[\varepsilon^\alpha Q_\alpha, \rheo{\psi}]$. Or,
pulling back to the canonical even submanifold:
\begin{equation}
  \delta_\varepsilon \comp{e} = -[\varepsilon^\alpha Q_\alpha, \comp{\psi}] .
\end{equation}
For the gravitino, we immediately compute
$\Lie^\omega_\varepsilon \rheo{\psi} = d_\omega i_\varepsilon \psi + i_\varepsilon d_\omega \psi = d_\omega( \varepsilon^\alpha Q_\alpha ) + i_\varepsilon \rho$.
From the rheonomic constraint on $\rho$, we read off
$i_\varepsilon \rho = \frac{3}{2} f \gamma(e) \varepsilon^\alpha Q_\alpha$. Hence,
the transformation law becomes:
\begin{equation}
  \delta_\varepsilon \comp{\psi} = d_{\comp{\omega}} ( \varepsilon^\alpha Q_\alpha ) + \frac{3}{2} \comp{f} \gamma(\comp{e}) \varepsilon^\alpha Q_\alpha.
\end{equation}

\subsection{Picture changing operators}
\label{sec:pco3d}

In the last few sections, we have introduced almost all of 3d, $\cN = 1$ geometric
supergravity, with one major omission: we have not given an action for the theory. In
order to do so, one final ingredient is needed: a `picture changing operator', or
PCO, which enables us to convert our Lagrangian from a differential 3-form to an
integral 3-form, and thereby obtain an integral over superspace. This integral will
be the action.

By giving the action in this form, as advocated by Castellani, Catenacci, and Grassi
\cite{CCG}, we are able to prove precise comparison theorems among the geometric,
superspace, and component formulations of 3d, $\cN = 1$ supergravity. While this
basic technique is due to Castellani \emph{et al.}, and those authors give some of
the same results, we improve upon their work in at least two ways. First, we give a
precise definition of `picture changing operator', making this concept amenable to
mathematical analysis. Second, we prove precise relationships between the spaces of
fields of all three formulations, rather than just their actions.

Let us briefly recap 3d, $\cN = 1$ geometric supergravity as we know it so far. The
fields are the geometric superfields: the superfield spin connection $\omega$, the
superfield coframe $e$, the superfield gravitino $\psi$, and the 2-form gauge
superfield $B$, all described in Def. \ref{def:fields}. The fields form a sheaf
$\Geo_{\R^{3|2}/\para}$. From these fields, we derive the curvatures in Def.
\ref{def:curvatures}, of which we recall the spin curvature,
$\curv = d \omega + \frac{1}{2} [\omega, \omega]$, and the 3-form curvature,
$H = dB - g(e, [\psi, \psi])$. In terms of the fields and their curvatures, we
construct the Lagrangian in Def. \ref{def:Lagrangian}:
\begin{equation}
  \Lag(\omega, e, \psi, B) \defeq \vol(\curv^{\sharp} \wedge e) + \langle \psi, d_{\omega} \psi \rangle + a (fH - \frac{1}{2} f^{2} \vol (e \wedge e \wedge e)) .
\end{equation}
Finally, we defined the rheonomic constraints in Def. \ref{def:rheonomic}. For us,
this includes the condition that the Lagrangian is closed:
\begin{equation}
  d\Lag(\varphi) = 0 ,
\end{equation}
where for brevity we have written $\varphi$ for the tuple $(\omega,e,\psi,B)$. We
write $\Rheo \subseteq \Geo_{\R^{3|2}/\para}$ for the subsheaf of rheonomic fields,
and distinguish local sections of this subsheaf by writing them in boldface,
$\rheo{\varphi} \in \Rheo_{\Lag}(U)$.

We now need to give the action. To obtain an integrand on superspace from the
differential 3-form $\Lag$, we pick one more piece of data, a `picture changing
operator' or PCO, in order to convert our differential form into an integral form.
The definition of PCO makes use of the canonical integral form
$Y_{\iota_\can} \in \Ber^{0}(\R^{3|2})$, defined by Poincar\'e duality with the
canonical even submanifold
$\iota_{\can} \maps \R^{3}_{/\para} \inclusion \R^{3|2}_{/\para}$.

\begin{defn}
  \label{def:pco}
  A \define{picture changing operator} or \define{PCO} is a pair $(\Y, \alpha)$ of
  maps of sheaves,
  \begin{equation}
    \Y \maps \Rheo \to \Ber^{0}_{\R^{3|2}/\para} , \quad \alpha \maps \Rheo \to \Ber^{-1}_{\R^{3|2}/\para},
  \end{equation}
  such that $\Y$ is cohomologous to the canonical integral form $Y_{\iota_{\can}}$
  via $\alpha$:
  \begin{equation}
    \Y = Y_{\iota_{\can}} + \delta \alpha .
  \end{equation}
  Here, the latter equation means that for any local section
  $\rheo{\varphi} \in \Rheo(U)$, we have
  \begin{equation}
    \Y(\rheo{\varphi}) = Y_{\iota_{\can}} + \delta \alpha(\rheo{\varphi}) .
  \end{equation}
\end{defn}

\noindent
We will generally suppress $\alpha$ and speak of the PCO $\Y$. In order to give the
action of geometric supergravity, we need to fix a PCO, but we will see momentarily
that all choices are equivalent.

\begin{defn} The \define{action of 3d, $\cN = 1$ geometric supergravity} is the integral\footnote{We ignore issues of convergence.}
  \begin{equation}
    \label{eq:action}
    S_{\geom}(\rheo{\varphi}) \defeq \int_{\R^{3|2}/\para} \Y(\rheo{\varphi}) \cdot \Lag(\rheo{\varphi}) ,
  \end{equation}
  where $\Lag$ is the Lagrangian, $\rheo{\varphi}$ is an rheonomic superfield, and
  $\Y$ is a choice of picture changing operator.
\end{defn}

\noindent
Thanks to the closure of the Lagrangian on the rheonomic superfields, all choices
of PCO give the same action.

\begin{prop}
  The action
  $S_{\geom}(\rheo{\varphi}) = \int_{\R^{3|2}/\para} \Y(\rheo{\varphi}) \cdot \Lag(\rheo{\varphi})$
  is independent of the choice of picture changing operator $\Y$.
\end{prop}
\begin{proof}
  By definition, if $\Y$ is a PCO, any other PCO has the form
  $\Y' = \Y + \delta \beta$, where
  $\beta \maps \Rheo \to \Ber^{-1}_{\R^{3|2}/\para}$ is a map of
  sheaves. Observe that
  \begin{equation}
    \int_{\R^{3|2}/\para} \delta \beta(\rheo{\varphi}) \cdot \Lag(\rheo{\varphi}) = -\int_{\R^{3|2}/\para} \beta(\rheo{\varphi}) \cdot d\Lag(\rheo{\varphi}) = 0
  \end{equation}
  since $\Lag(\rheo{\varphi})$ is closed on rheonomic superfields. Thus we
  immediately conclude that
  \begin{equation}
    \int_{\R^{3|2}/\para} \Y(\rheo{\varphi}) \cdot \Lag(\rheo{\varphi}) = \int_{\R^{3|2}/\para} \Y'(\rheo{\varphi}) \cdot \Lag(\rheo{\varphi}) ,
  \end{equation}
  so $S_{\geom}(\rheo{\varphi})$ is independent of the choice of PCO.
\end{proof}

\noindent
The above fact may seem somewhat formal, even trivial, but it is quite useful:
following the central idea of Castellani \emph{et al.} \cite{CCG}, we use it to show
the geometric action is equal to both the component action and the superspace action,
simply by making different choices of picture changing operator.

First, we need to define the action we wish to compare to the geometric action.

\begin{defn}
  The \define{component Lagrangian of 3d, $\cN = 1$ supergravity} is the following
  3-form on $\R^3_{/\para}$:
  \begin{equation}
    \Lag_{\compo}(\comp{\omega}, \comp{e}, \comp{\psi}, \comp{B}) \defeq \vol(\comp{\curv}^{\sharp} \wedge \comp{e}) + \langle \comp{\psi}, d_{\omega} \comp{\psi} \rangle + 6 (\comp{f}\comp{H} - \frac{1}{2} \comp{f}^{2} \vol (\comp{e} \wedge \comp{e} \wedge \comp{e})) .
  \end{equation}
  where $(\comp{\omega}, \comp{e}, \comp{\psi}, \comp{B})$ are the component
  fields, as defined in Def. \ref{def:comp}, and
  \begin{itemize}
    \item $\comp{\curv} = d \comp{\omega} + \frac{1}{2} [\comp{\omega}, \comp{\omega}]$
          is the curvature of the spin connection $\comp{\omega}$;
    \item $\comp{H} = d\comp{B} - \frac{1}{2} g(\comp{e}, [\comp{\psi}, \comp{\psi}])$
          is the 3-form curvature of the gauge 2-form $\comp{B}$;
    \item $\comp{f}$ is the unique $0$-form such that
          $\comp{H} = \comp{f} \vol(\comp{e} \wedge \comp{e} \wedge \comp{e})$.
  \end{itemize}
  The \define{component action of 3d, $\cN = 1$ supergravity} is the integral
  \begin{equation}
    S_{\compo}(\comp{\varphi}) \defeq \int_{\R^{3}/\para} \Lag_{\compo}(\comp{\varphi}) .
  \end{equation}
\end{defn}

As is typical in the literature on geometric supergravity, the component Lagrangian
is chosen to match the full Lagrangian. That is, by construction, the pullback of
$\Lag$ to the canonical even submanifold is the component Lagrangian:
\begin{equation}
  \iota^{*}_{\can} \Lag(\varphi) = \Lag_{\compo}(\comp{\varphi}) ,
\end{equation}
where we take the component field to be the pullback,
$\comp{\varphi} = \iota^{*}_{\can} \varphi$.

\begin{thm}
  \label{thm:geomiscomp}
  Let $\rheo{\varphi}$ be an rheonomic superfield, and let
  $\comp{\varphi} = \iota^{*}_{\can} \rheo{\varphi}$ be the component field obtained
  by pullback. Then the component action of $\comp{\varphi}$ equals the geometric
  action of $\rheo{\varphi}$:
  \begin{equation}
    S_{\compo}(\comp{\varphi}) = S_{\geom}(\rheo{\varphi}) .
  \end{equation}
\end{thm}

\noindent
This equivalence with the component action comes from the most obvious choice of PCO
we can make: the \define{canonical PCO} $\Y_{\can}$, also called the `spacetime PCO'
by some authors. This is the constant PCO with no dependence on the fields,
$\Y_{\can}(\varphi) \defeq Y_{\iota_\can}$. With this choice of PCO, the following calculation
serves as the proof of the theorem:
\begin{eqnarray}
  S_{\geom}(\rheo{\varphi}) & = & \int_{\R^{3|2}/\para} \Y_{\can} \cdot \Lag(\rheo{\varphi}) \\
                            & = & \int_{\R^{3}/\para} \iota^{*}_{\can} \Lag(\rheo{\varphi}) \\
                            & = & \int_{\R^{3}/\para} \Lag_{\compo}(\comp{\varphi}) \\
                            & = & S_{\compo}(\comp{\varphi}) .
\end{eqnarray}
In going from the first line to the second, we used Theorem \ref{thm:dual}. Apart
from a little extra mathematical care, this is precisely the argument in Castellani
\emph{et al.\ }\cite{CCG}.

Next, we use the same idea to prove the geometric action equals the superspace
action. The latter is a beautiful action for 3d, $\cN = 1$ supergravity given in the
superspace formalism \cite{RRvN}. In contrast to the canonical PCO, the picture
changing operator for the superspace formalism depends nontrivially on the fields.

\subsection{3d supergravity -- the superspace formulation}

In the superspace formalism, the fields are the same as in the geometric formalism,
except the 2-form gauge superfield is omitted.

\begin{defn}
  \label{def:susyfields}
  We define the \define{superspace superfields of 3d, $\cN = 1$ supergravity} to be
  the following:
  \begin{itemize}
    \item an $\so(2,1)$-valued 1-form
          $\omega \in \Omega^{1}(\R^{3|2}_{/\para}, \so(2,1))$, the
          \define{superfield spin connection};
    \item a vector-valued 1-form $e \in \Omega^{1}(\R^{3|2}_{/\para}, V)$,
          the \define{superfield coframe};
    \item a spinor-valued 1-form $\psi \in \Omega^{1}(\R^{3|2}_{/\para}, S)$,
          the \define{superfield gravitino}.
  \end{itemize}
  In addition, we require the $\t$-valued 1-form given by the sum
  $E = e + \psi$ to be nondegenerate in the sense that
  $E$ trivializes the tangent bundle. That is, we have an isomorphism of
  $\Oh(\R^{3|2} \times \para)$-modules:
  \begin{equation}
    E \maps \T \R^{3|2}_{/\para} \isoto \Oh(\R^{3|2} \times \para) \otimes \t,
  \end{equation}
  where $\T \R^{3|2}_{/\para}$ denotes the module of vertical vector fields on our
  family. As we already noted earlier, the $\t$-valued 1-form
  $E \in \Omega^{1}(\R^{3|2}_{/\para}, \t)$ is called the \define{supervielbein}.
\end{defn}

The superspace superfields naturally form a sheaf $\Sup_{\R^{3|2}/\para}$; there is a
map of sheaves $\tau$ that simply forgets $B$:
\begin{equation}
  \renewcommand{\arraycolsep}{1pt}
  \begin{array}{cccl}
    \tau \maps & \Geo_{\R^{3|2}/\para} & \to & \Sup_{\R^{3|2}/\para} \\
     & (\omega, e, \psi, B) & \mapsto & (\omega, e, \psi) ,
  \end{array}
\end{equation}
where $(\omega, e, \psi, B)$ denotes a local section of $\Geo_{\R^{3|2}/\para}$, and
similarly for $(\omega, e, \psi)$ in $\Sup_{\R^{3|2}/\para}$. We call $\tau$ the
\define{truncation map}.

It may seem quite surprising that we can discard the field $B$ and still obtain an
equivalent theory, or at least a morally equivalent theory. The secret to doing so
lies in the constraints of the superspace formalism, which reduce the superspace
superfields to a triple consisting of a metric, a gravitino, and a scalar field, all
defined on spacetime rather than superspace. This being done, we can identify the
scalar with the 3-form $\comp{H}$, since in three-dimensions scalars and 3-forms are
dual. In this way, we recover $B$ in the guise of its curvature $H$, or at least the
part of its curvature on spacetime, $\comp{H}$.

To define the constraints, we recall that the curvature $\curv$, the torsion $T$ and the gravitino field
strength $\rho$ are defined to be (Def. \ref{def:curvatures}):
\begin{equation}
  \curv = d \omega + \frac{1}{2} [\omega, \omega], \quad T = d_\omega e + \frac{1}{2} [\psi,\psi], \quad \rho = d_\omega \psi .
\end{equation}
It is easiest to write the constraints in terms of components, so let $\{P_a\}$ be an
oriented orthonormal basis for $V$, $\{Q_\alpha\}$ a basis for $S$, and write
$\{u_A\}$ for the basis of $\t$ obtained as the union of $\{P_a\}$ and
$\{Q_\alpha\}$. The supervielbein then decomposes as $E = E^A u_A$, where $\{E_A\}$
is a basis for 1-forms. Hence, $\{E^A \wedge E^B\}$ forms a basis for 2-forms, and we
can expand the 2-forms $\curv$, $T$ and $\rho$ as follows:
\begin{equation}
 \curv = \frac{1}{2} E^A \wedge E^B \curv_{AB}, \quad T = \frac{1}{2} E^A \wedge E^B \,  T_{AB}, \quad \rho = \frac{1}{2} E^A \wedge E^B \, \rho_{AB} .
\end{equation}

\begin{defn}
  \label{def:susyconstraints}
  The \define{conventional constraints} for 3d, $\cN = 1$ supergravity are the
  following equations for the curvature, the torsion and gravitino field strength,
  for all vector indices $b \in \{0,1,2\}$ and spinor indices
  $\alpha, \beta \in \{+, -\}$:
  \begin{equation}
    \label{eq:conv}
    \curv_{\alpha \beta} = 0, \quad T_{\alpha \beta} = 0, \quad \rho_{\alpha \beta} = 0, \quad T_{\alpha b} = 0.
  \end{equation}

  Equivalently, we can write these constraints without indices, using the superfield
  frame $e^{-1}$ and inverse superfield gravitino $\psi^{-1}$:
   \begin{align}
     (\psi^{-1} \wedge \psi^{-1}, \curv) & =  0 ,  \label{eq:convcurv} \\
     (\psi^{-1} \wedge \psi^{-1}, T) & =  0 , \label{eq:convtors} \\
     (\psi^{-1} \wedge \psi^{-1}, \rho) & =  0 , \label{eq:convgrav} \\
     (e^{-1} \wedge \psi^{-1}, T) & =  0 . \label{eq:unconvtors}
    \end{align}
  As always, $(-,-)$ denotes the pairing between polyvector fields and differential
  forms.
\end{defn}

\noindent
We write $\Conv \subseteq \Sup_{\R^{3|2}/\para}$ for the subsheaf of superspace
superfields satisfying the conventional constraints. We write $\conv{\varphi}$ for a
local section of $\Conv$, and call $\conv{\varphi}$ a \define{conventional
  superfield}.

Some readers may have seen a version of one of the constraints that reads
$T_{\alpha \beta}^a = \gamma^a_{\alpha \beta}$, whereas we have set this tensor to
zero. This is consistent, however, as we have defined the torsion tensor $T$ to
include the term $\frac{1}{2}[\psi,\psi]$, and this term cancels the $\gamma$ matrix
when we pass to components.

We adapted the constraint equations \ref{eq:conv} from the invaluable lecture notes
of Ruiz Ruiz and van Nieuwenhuizen \cite{RRvN}. We note that those authors use the
term ``conventional constraints'' only for equations \ref{eq:convcurv},
\ref{eq:convtors}, \ref{eq:convgrav}, but we have adopted this terminology for all
constraints. For more on torsion constraints in supergravity, see the paper of Lott
\cite{Lott}. See Merkulov for a detailed mathematical treatment of the 3d case
\cite{Merkulov}.

A well-known consequence of the conventional constraints is that the curvature,
torsion and gravitino field strength can be expressed in terms of only two
superfields, the so-called `irreducible superfields'. These are:
\begin{itemize}
  \item a scalar superfield, $\scal \in \Omega^0(\R^{3|2}_{/\para})$;
  \item a tensor superfield,
        $\tens \in \Omega^0(\R^{3|2}_{/\para}, \bigvee^3 S)$.
\end{itemize}
The irreducible superfields are derived by combining the conventional constraints
with the Bianchi identity \cite{RRvN}. For our purposes here, we only need the
scalar.

\begin{defn}
  Fixing a conventional superfield $\conv{\varphi}$, the \define{superspace scalar}
  is the unique 0-form $\scal(\conv{\varphi}) \in \Omega^0(\R^{3|2}_{/\para})$ such
  that, for all vector indices $a,b,c \in \{0,1,2\}$ of an oriented, orthonormal
  basis $\{P_a\}$ of $V$, we have:
  \begin{equation}
    g_{ad} T^{d}_{bc} = \epsilon_{abc} \scal(\rheo{\varphi}) .
  \end{equation}
  where $g$ is the Minkowski metric and $\epsilon_{abc}$ is the Levi-Civita
  alternating symbol \cite{CCG}. More invariantly, this equation reads:
  \begin{equation}
    g\left(\conv{e}^{-1}, (\conv{e}^{-1} \wedge \conv{e}^{-1}, T)\right) = \vol(\conv{e}^{-1} \wedge \conv{e}^{-1} \wedge \conv{e}^{-1}) \scal(\conv{\varphi}) .
  \end{equation}
\end{defn}

We need only one more ingredient to define the Lagrangian in the superspace
formulation, which is a section of the Berezinian canonically associated to any
supervielbein. To define this in turn, we need the idea of a `good basis' of
the supertranslation algebra $\t$.

\begin{defn}
  A \define{good basis} $\{P_a | Q_\alpha\}$ of $\t$ consists of a basis
  $\{Q_\alpha\}_{\alpha \in \{+, -\}}$ of the spinors $S$ such that
  $\ip{Q_+, Q_-} = 1$, and an oriented, orthonormal basis $\{P_a\}_{a \in \{0,1,2\}}$
  of the vectors $V$ defined in terms of the brackets of $\{Q_\alpha\}$ as follows:
  \begin{eqnarray}
    P_0 & = & \frac{1}{2} \left([Q_+, Q_+] + [Q_-, Q_-] \right)\\
    P_1 & = & \frac{1}{2} \left([Q_+, Q_+] - [Q_-, Q_-] \right)\\
    P_2 & = & [Q_+, Q_-]
  \end{eqnarray}
\end{defn}

\begin{rem}
  The expressions for the basis $\{P_a\}$ above come from using the isomorphism
  $V \iso \bigvee^2 S$ to view elements of $V$ as $2 \times 2$ symmetric matrices,
  and then identifying the components of the matrix with the usual coordinates of
  spacetime as follows:
  \begin{equation}
    \begin{pmatrix} t + x & y \\ y & t - x \end{pmatrix} .
  \end{equation}
  In this way, the quadratic form $-\det$ on the space of $2 \times 2$ matrices gives
  the usual Minkowski metric.
\end{rem}

\begin{defn}
  Let $E = e + \psi$ be a supervielbein. The \define{superspace Berezinian}
  $\ber(E) \in \Ber^3(\R^{3|2}_{/\para})$ is the integral 3-form defined by using the
  supervielbein to give an ordered basis for the 1-forms
  $\Omega^1(\R^{3|2}_{/\para})$, and then passing to the induced element of the
  Berezinian module,
  $\Ber^3(\R^{3|2}_{/\para}) \defeq \Ber(\Omega^1(\R^{3|2}_{/\para}))$.

  Specifically, for a good basis $\{P_a | Q_\alpha\}$ of $\t$, we define:
  \begin{equation}
    \ber(E) \defeq [\psi_+ \psi_- \otimes e^0 e^1 e^2 ]
  \end{equation}
\end{defn}

\begin{prop}
  The superspace Berezinian $\ber(E)$ is independent of the choice of good basis for
  the supertranslation algebra $\t$. Since the Lorentz group $\Spin(2,1)$ acts
  transitively on good bases, this implies that $\ber(E)$ is Lorentz-invariant.
\end{prop}

\begin{proof}
  The Lorentz group $\Spin(2,1) \iso \SL(2, \R)$ indeed acts transitively on good
  bases, since a good basis $\{P_a | Q_\alpha\}$ is determined by the basis
  $\{Q_\alpha\}$ such that $\ip{Q_+, Q_-} = 1$, and $\SL(2, \R)$ acts transitively on
  this set. So, let $g \in \Spin(2,1)$, and note that
  \begin{equation}
    \ber(\rho_V(g) e | \rho_S(g) \psi) = \frac{\det(\rho_V(g))}{\det(\rho_S(g))} \ber(e | \psi) = \ber(e | \psi) ,
  \end{equation}
  where $\rho_V$ and $\rho_S$ are the representations of $\Spin(2,1)$ on $V$ and $S$
  respectively, both of which act with unit determinant,
  $\det(\rho_V(g)) = \det(\rho_S(g)) = 1$.
\end{proof}

\noindent
We are now ready for the action, which completes our construction of 3d, $\cN = 1$
supergravity in the superspace formalism.

\begin{defn}
  \label{def:susyaction}
  The \define{superspace Lagrangian of 3d, $\cN = 1$ supergravity} is the integral 3-form
  \begin{equation}
    \Lag_\susy(\conv{\varphi}) = \ber(\conv{\varphi}) \scal(\conv{\varphi}) ,
  \end{equation}
  where $\conv{\varphi} = (\conv{\omega},\conv{e},\conv{\psi})$ is a conventional
  superfield.

  The \define{superspace action of 3d, $\cN = 1$ supergravity} is the integral of the
  Lagrangian over superspace:
  \begin{equation}
    S_\susy(\conv{\varphi}) = \int_{\R^{3|2}_{/\para}} \Lag_\susy(\conv{\varphi}) .
  \end{equation}
\end{defn}

\subsubsection{\textbf{Reduction to component fields}}

Although we will not actually need the reduction to component fields for the
superspace formalism, we take a moment to sketch how this reduction should work.

In the superspace formalism, the relation between superfields and component fields is
less obvious than in the geometric formalism. The main complication is that a proper
mathematical treatment uses the language of stacks, rather than sheaves. This is
because the 3d supergravity has gauge symmetry. Of course, the reader should note,
geometric supergravity \emph{also} has gauge symmetry, and yet the language of stacks
was not needed when we discussed the reduction to component fields in the geometric
formalism. This suggests that the rheonomic constraints are much more `rigid', in
some suitable sense, than the conventional constraints.

Complications aside, we can recover the coframe and gravitino, along with an
auxiliary scalar field, from the superspace superfields.

\begin{defn}
  The \define{superspace component fields} of 3d, $\cN = 1$ supergravity are the
  following differential forms on spacetime, $\R^3_{/\para}$:
  \begin{itemize}
    \item a vector-valued 1-form, $\comp{e} \in \Omega^1(\R^3_{/\para}, V)$, the
          \define{coframe};
    \item a spinor-valued 1-form, $\comp{\psi} \in \Omega^1(\R^3_{/\para}, S)$, the
          \define{gravitino};
    \item a 0-form, $\comp{\scal} \in \Omega^0(\R^3_{/\para})$, the \define{auxiliary
          scalar}.
  \end{itemize}
  As always, we require that the coframe field $\comp{e}$ is nondegenerate in the
  sense that it trivializes the tangent bundle of spacetime. In other words, we have
  an isomorphism of $\Oh(\R^{3} \times \para)$-modules:
  \begin{equation}
    \comp{e} \maps \T \R^{3}_{/\para} \isoto \Oh(\R^{3} \times \para) \otimes V .
  \end{equation}
  The superspace component fields of supergravity naturally form a sheaf,
  $\SupComp$.
\end{defn}

Stacks enter when we try to construct a map $\kappa \maps \Sup \to \SupComp$. Here we
just give the idea. Perhaps surprisingly, $\kappa$ only depends on the inverse
superfield gravitino $\psi^{-1}$ \cite{ZP}. Fixing coordinates
$x^M = x^m | \theta^\mu$ and a basis $\{Q_\alpha\}$ of the spinors, we expand
$\psi^{-1}$ in components:
\begin{equation}
  \psi^{-1} = \epsilon^{\beta\alpha} Q_{\beta} \psi_{\alpha M} \pd{x^M} .
\end{equation}
We then expand each component $\psi_{\alpha M} \in \Oh(\R^{3|2} \times \para)$ in the
$\theta^\mu$. In this expansion, some components are extracted and used to construct
the coframe $\comp{e}$, gravitino $\comp{\psi}$, and the auxiliary scalar
$\comp{\scal}$. The other component fields in this expansion can be eliminated using
gauge transformations.

Stated more invariantly, we thus have a map of stacks:
\begin{equation}
  \kappa \maps [\Sup/\Gauge] \to [\SupComp/\GaugeComp] .
\end{equation}
Here, $\Gauge$ is the sheaf of `gauge transformations' of the superspace superfields,
$\Sup$, and $\GaugeComp$ is the sheaf of gauge transformations of the superspace
component fields, $\SupComp$, and the brackets tell us to `construct a stack' from
the data indicated.

Now, using the constraints together with the Bianchi identity, it is possible to
express $\conv{\omega}$ and $\conv{e}$ in terms of $\conv{\psi}$. This implies
$\kappa$ induces an equivalence of stacks between the conventional superfields and
the superspace component fields:
\begin{equation}
  \kappa \maps [\Conv/\Gauge] \isoto [\SupComp/\GaugeComp] .
\end{equation}
Of course, physicists know gauge theory, not the language of stacks, so they write
this differently. A mathematical treatment would be interesting to see.

Similar to the geometric formalism, the superspace superfields, conventional
superfields, and superspace component fields fit together into a commutative
triangle. Now, however, it is a triangle of stacks. This is the version of Fig.
\ref{fig:shape} for the superspace formalism.

\begin{claim}
  \label{claim:suptriangle}
  We have the following commutative triangle of stacks on $\R^3 \times |\para|$,
  where $\kappa$ is the map to component fields, and the vertical map is an
  equivalence of stacks:
  \begin{equation}
    \begin{tikzcd}
      \left[ \Sup/\Gauge \right] \ar[dr, "\kappa \quad" below] & \left[ \Conv/\Gauge \right] \ar[l, hook'] \ar[d, "\simeq"] \\
                                  & \left[ \SupComp/\GaugeComp \right] \\
    \end{tikzcd}
  \end{equation}
\end{claim}

\begin{proof}[Hints for a proof.]
  The analysis of the superspace formalism for 3d supergravity is known to parallel
  that of 3d super-Yang--Mills theory \cite{RRvN}. For super-Yang--Mills theory,
  Deligne and Freed prove that the stack of constrained super-Yang--Mills fields is
  equivalent to the stack of component fields \cite{DFSuper}. We expect their proof
  could be adapted to supergravity.
\end{proof}

\subsection{The supersymmetric PCO}
\label{sec:susypco}

We now compare the superspace formulation of supergravity with the geometric
formulation via a suitable choice of picture changing operator, much as we did for
the component formulation. This proved to be quite a bit more subtle than for the
component case. 

To begin, we introduce the 0-form part of the putative supersymmetric picture
changing operator:

\begin{defn} 
  The \define{supersymmetric picture changing operator} is the sheaf map
  \begin{equation} 
    \Y_\susy \maps \Rheo \to \Ber^0_{\R^{3|2}/\para},
  \end{equation}
  given by the following formula:
  \begin{equation} \label{susypco}
    \Y_{\susy}(\rheo{\varphi}) = \ber(\rheo{E}) \otimes g(\rheo{e}^{-1}, [\rheo{\psi}^{-1}, \rheo{\psi}^{-1}]) ,
  \end{equation}
  where $\rheo{\varphi} = (\rheo{\omega}, \rheo{e}, \rheo{\psi}, \rheo{B})$ is an
  rheonomic superfield, $\rheo{E} = \rheo{e} + \rheo{\psi}$ is the associated
  supervielbein, $\rheo{E}^{-1} = \rheo{e}^{-1} + \rheo{\psi}^{-1}$ its inverse,
  $\ber(\rheo{E})$ is the Berezinian determined by $\rheo{E}$, and
  $g(\rheo{e}^{-1}, [\rheo{\psi}^{-1}, \rheo{\psi}^{-1}])$ is a 3-vector field. Since
  $\Y_\susy (\rheo \varphi)$ only depends on the supervielbein $\rheo E$, we will
  often write $\Y_\susy (\rheo E)$.
\end{defn}


\noindent
In an appendix to their paper, Castellani, Catenacci, and Grassi \cite{CCG} attempt
to prove that $\Y_\susy$ truly defines a PCO (as in Def. \ref{def:pco}) along these lines, in two steps:
\begin{itemize}
  \item First they argue $\Y_\susy(\rheo{E}_\flat)$ is cohomologous to $Y_{\iota_\can}$; 
  \item They argue that $\Y_\susy(\rheo E)$ is cohomologous to
        $\Y_\susy(\rheo{E}_\flat)$, where $\rheo E$ is arbitrary, while
        $\rheo{E}_\flat$ is a specific supervielbein which is flat in the sense of
        Cartan geometry.
\end{itemize}
The first point is true, as we prove below. The flaw lies in the argument for the
second point, for which the authors posit the existence of a diffeomorphism $\Xi$,
isotopic to the identity, such that $\Xi^* \rheo{E} = \rheo{E}_\flat$. However, no
such diffeomorphism exists unless $\rheo{E}$ is itself flat. The Cartan curvature is
the obstruction. \\

In the following we will prove that $\Y_\susy$ does indeed define a PCO in the sense of Definition \ref{def:pco}.\\
First, we remark that $\Y_\susy(\rheo{E})$ cannot be too far from $Y_{\iota_\can}$. Indeed, it is
closed, as can be checked by a detailed calculation \cite{CCG}. Since
$H^0_\Sp(\R^{3|2}_{/\para})$ is one-dimensional, we can take $[Y_{\iota_\can}]$ as its
generator, and thus we must have
\begin{equation} \label{eqn:Ysusyeqn}
  \Y_\susy(\rheo{E}) = c(\rheo{E}) Y_{\iota_\can} + \delta \alpha_\susy(\rheo{E}) ,
\end{equation}
for some integral $-1$-form $\alpha_\susy(\rheo{E})$, and some constant
$c(\rheo{E}) \in \Oh(\para)$ that depends only on $\rheo{E}$.

\begin{defn}
  The \define{flat supervielbein} $\rheo{E}_\flat$ is the left-invariant coframe of
  super-Minkowski spacetime; in coordinates, it is given by
  \begin{equation}
    \rheo{E}^a_\flat = dx^a - \frac{1}{2} \theta^\alpha \gamma^a_{\alpha \beta} d \theta^\beta, \quad \rheo{E}^\alpha_\flat = d \theta^\alpha .
  \end{equation}
\end{defn}

\noindent
For $\rheo{E}_\flat$ we know $c(\rheo{E}_\flat) = 1$, as we shall prove now.

\begin{prop} \label{prop:flat}
  $\Y_\susy(\rheo{E}_\flat)$ is cohomologous to $Y_{\iota_\can}$.
\end{prop}

\begin{proof}
  Indeed, consider
  \begin{equation}\label{eq:Omega0}
    \alpha_0(\rheo{\varphi}_{\mathpzc{susy}}) \defeq
    \frac12\,x^a\,(\gamma_{ab})^{\alpha\beta}\,\epsilon^{bcd}\;
    \ber(\rheo{E}_{\mathpzc{flat}})\otimes\bigl((\rheo{e}^{-1}_{\mathpzc{flat}})_c\wedge (\rheo{e}^{-1}_{\mathpzc{flat}})_d\wedge (\rheo{\psi}^{-1}_{\mathpzc{flat}})_\alpha\wedge (\rheo{\psi}^{-1}_{\mathpzc{flat}})_\beta\bigr),
  \end{equation}
  \begin{equation}\label{eq:Omega1}
    \alpha_1 (\rheo{\varphi}_{\mathpzc{susy}})\defeq
    \frac{1}{4}\,(\theta^\gamma \gamma^a_{\gamma\delta}\theta^\delta)\,(\gamma_{ab})^{\alpha\beta}\,\epsilon^{bcd}\;
    \ber(\rheo{E}_{\mathpzc{flat}})\otimes\bigl((\rheo{e}^{-1}_{\mathpzc{flat}})_c\wedge (\rheo{e}^{-1}_{\mathpzc{flat}})_d\wedge (\rheo{\psi}^{-1}_{\mathpzc{flat}})_\alpha\wedge (\rheo{\psi}^{-1}_{\mathpzc{flat}})_\beta\bigr).
  \end{equation}
  Then a lengthy calculation shows that
  $\Y_\susy(\rheo{E}_\flat) = Y_{\iota_\can} + \delta( \alpha_0 + \alpha_1 )$.
\end{proof}

To settle the case of generic $\rheo{E}$, we first observe that Eq.\
\refeq{eqn:Ysusyeqn} implies that
\begin{equation}
  \label{cEinpractice}
  c(\rheo{E}) = \int_{\R^{3|2} / \para} \Y_\susy (\rheo E) \cdot p^* \varpi,
\end{equation}  
where $\varpi$ is a compactly supported 3-form which integrates to 1 on
$\mathbb{R}^3_{ / \para}$, and $p^*\varpi$ is its pullback to $\R^{3|2}_{/ \para}$
via the map $p \maps \R^{3|2}_{/\para} \to \R^{3}_{/\para}$ given by $p^* x^a = x^a$.
Since $p^*\varpi$ contains no $d\theta$-components, the contraction with $p^*\varpi$
kills every term of the 3-vector
$\Theta (\rheo E) \defeq g(\rheo{e}^{-1}, [\rheo{\psi}^{-1}, \rheo{\psi}^{-1}]) $ in
$\Y_{\susy} (\rheo E)$ containing a component along an odd vector. Thus
\eqref{cEinpractice} reads off precisely the coefficient of $\theta^1\theta^2$ in the
purely bosonic component of $\Theta(\rheo E)$.

\begin{prop} Let $\rheo{E}$ be a vielbein field satisfying the supergravity torsion constraints 
\begin{equation} \label{eqn:sugralocus}
[\rheo{E}_\alpha , \rheo{E}_\beta] = \kappa (\gamma^a)_{\alpha \beta} \rheo{E}_a.
\end{equation} 
Then $c(\rheo {E})$ is independent of $\rheo E$. In particular, $\Y_\susy(\rheo{E})$ is a PCO for any $\rheo{E}$ satisfying Eq.\ \eqref{eqn:sugralocus}.
\end{prop}

\begin{proof}
  Working in normal coordinates\footnote{Note that this is not a restriction on
    curvature.}, the bosonic component of $\rheo E_\alpha$ satisfies
  \begin{equation}
    (\rheo E_\alpha)_{0}=-\frac12\,\theta^\nu T_{\nu\alpha}{}^{a}\rheo E_a+\mathcal{O}(\theta^2).
  \end{equation}
  Therefore, the relevant component of $\Theta(\rheo E)$ for $c(\rheo E)$ begins at
  order $\theta^2$, and we obtain for the 3-vector
  \begin{align}
    \label{eq:Theta-30-explicit}
    \Theta(\rheo{E})^{\theta\theta}_{3,0} & =
                                            2(\gamma^{c})^{\alpha\beta}\,\rheo E_c\wedge (\rheo E_\alpha)_{0}\wedge (\rheo E_\beta)_{0}
                                            \nonumber \\
                                          &=  \frac12\,\theta^\nu\theta^\mu\,
                                            (\gamma^{c})^{\alpha\beta}\,T_{\nu\alpha}{}^{a}\,T_{\mu\beta}{}^{b}\;
                                            \rheo E_c\wedge \rheo E_a\wedge \rheo E_b.
      \end{align}
  where we have denoted $\Theta (\rheo E)^{\theta \theta}_{3,0}$ the part of the
  3-vector which has no fermionic components. Substituting equation
  \eqref{eqn:sugralocus} in the form of
  $T_{\alpha\beta}{}^{a}=\kappa(\gamma^{a})_{\alpha\beta}$ into
  \eqref{eq:Theta-30-explicit} gives
  \begin{equation}
    \Theta(\rheo E)^{\theta \theta}_{3,0}
    =
    \frac{\kappa^2}{2}\,\theta^\nu\theta^\mu\,
    (\gamma^{c})^{\alpha\beta}(\gamma^{a})_{\nu\alpha}(\gamma^{b})_{\mu\beta}\;
    \rheo E_c\wedge \rheo E_a\wedge \rheo E_b.
  \end{equation}
  The tensor
  $(\gamma^{c})^{\alpha\beta}(\gamma^{a})_{\nu\alpha}(\gamma^{b})_{\mu\beta}$ is
  Spin$(2,1)$-equivariant, and its totally antisymmetric part in $(c,a,b)$ is
  therefore proportional to $\varepsilon^{cab}\varepsilon_{\nu\mu}$. Consequently,
  there exists a constant $\lambda\in\R$ depending only on $\kappa$ (and on the fixed
  gamma-matrix conventions) such that
\begin{equation}
\label{eq:Theta-universal-form}
\Theta(\rheo E)^{\theta \theta}_{3,0}=\lambda\;\theta^1\theta^2\;\rheo E_1\wedge \rheo E_2\wedge \rheo E_3. 
\end{equation}
By Eq.\ \eqref{cEinpractice}, $c (\rheo E) = \lambda$. Finally, by Proposition \ref{prop:flat}, $c (\rheo E_\flat ) = 1$, 
thus $c (\rheo E) = \lambda = 1$, for any $\rheo E$ satisfying the supergravity constraints.
\end{proof}

The constancy of $c(\rheo E)$ proved above relies crucially on working on the
supergravity torsion constraint locus of Eq.\ \eqref{eqn:sugralocus} of the chosen
supergravity theory (and on nondegeneracy of the supervielbein): in fact the
statement $c(\rheo E)=1$ should be read as a property of the \emph{constraint locus
  of the chosen theory}, not as a formal identity for arbitrary supervielbeins.

If one leaves this setting, there are several (rather ``exotic'') mechanisms by which
the scalar $c(\rheo E)$ may fail to be constant, or even cease to be well-defined.
For example, if $T_{\alpha\beta}{}^{a}$ is allowed to vary (\emph{e.g.}\ super-Weyl
covariant formulations before gauge fixing, or models with compensators in which
$T_{\alpha\beta}{}^{a}=W\,(\gamma^{a})_{\alpha\beta}$ with a dynamical superfield
$W$), then the leading $\theta^2$ coefficient controlling $c(\rheo E)$ typically
acquires dependence on the additional fields, and $c(\rheo E)$ may vary accordingly.
Also, even with the standard constraints, the space of admissible backgrounds may
have multiple connected components, (\emph{e.g.}, due to orientation/spin-structure
choices, or other global data). Then $c(\rheo E)$ is constant on each component, but
a priori could take different constant values on different components (fixed by the
chosen normalization on each component).

In the present paper we stay within the minimal constraint locus of the chosen theory, where these pathologies
are excluded and the normalization fixed by Proposition~\ref{prop:flat} yields $c(\rheo E)=1$.

Now that we have a PCO, to make use of $\Y_\susy$, we need to know how the
rheonomic constraints in the geometric formulation relate to the conventional
constraints in the superspace formulation. We show that, up to a shift in the spin
connection, these constraints coincide. To write the shift, recall that the geometric
scalar $f$ is a component of the 3-form curvature $H$.

\begin{prop}
  \label{prop:shift}
  Let ${\varphi} = ({\omega}, {e}, {\psi}, {B})$ be a geometric superfield, with
  geometric scalar $f$. Define the $\so(2,1)$-valued 1-form
  \begin{equation}
    \shift \defeq -3f e^\star \in \Omega^1(\R^{3|2}_{/\para}, \so(2,1)) ,
  \end{equation}
  where $e^\star$ is the $\so(2,1)$-valued 1-form obtained from the vector-valued $e$
  using the Hodge star. If
  $(\rheo{\omega}, \rheo{e}, \rheo{\psi}, \rheo{B}) \in \Rheo(\R^{3|2}_{/\para})$ is
  rheonomic, then
  $(\rheo{\omega} + \shift, \rheo{e}, \rheo{\psi}) \in \Conv(\R^{3|2}_{/\para})$ is a
  conventional superfield.

  Sheafifying, this says the following square of sheaves on $\R^3 \times |\para|$
  commutes:
  \begin{equation}
    \begin{tikzcd}
      \Geo \ar[r, "\trun + \shift"] & \Sup \\
      \Rheo \ar[u, hook] \ar[r, "\trun + \shift"] & \Conv \ar[u, hook]
    \end{tikzcd}
  \end{equation}
  where we have written $\trun + \shift$ for the map
  $({\omega}, {e}, {\psi}, {B}) \mapsto ({\omega} + \shift, {e}, {\psi})$.
\end{prop}

\begin{proof}
  Let $\curv$ denote the curvature of $\omega$. Then
  \begin{equation}
    \curv' = \curv + d_\omega \shift + \frac{1}{2}[\shift, \shift]
  \end{equation}
  is the curvature of the shifted connection, $\omega + \shift$. We need to check
  $\curv'$ satisfies the conventional constraint $\curv'_{\alpha\beta} = 0$. Since
  the term $[\shift,\shift]$ has no $\psi \wedge \psi$ component, only
  $d_\omega \shift$ is relevant. Note that
  \begin{equation}
    d_\omega \shift = -3df \wedge e^\star - 3 f (d_\omega e)^\star .
  \end{equation}
  The first term is irrelevant, as it has no $\psi \wedge \psi$ component. The second
  term, on the other hand, is $-\frac{3}{2} f [\psi, \psi]^\star$, thanks to the
  torsion constraint. By Prop. \ref{prop:rheo}, the $\psi \wedge \psi$ component of
  $\curv$ is $\frac{3}{2} f [\psi, \psi]^\star$, which precisely cancels with the
  contribution from $d_\omega \shift$. Hence $\curv'_{\alpha\beta} = 0$, as desired.
\end{proof}

\noindent
For the next proposition, recall that the superspace scalar $\scal$ is a component of
the torsion $T = d_\omega e + \frac{1}{2} [\psi,\psi]$, where $(\omega, e, \psi)$
satisfy the conventional constraints.

\begin{prop}
  \label{prop:scalars}
  Let $\shift$ be the $\so(2,1)$-valued 1-form from Prop. \ref{prop:shift}. If
  $\rheo{\varphi}$ is a rheonomic superfield, then the geometric scalar of
  $\rheo{\varphi}$ and the superspace scalar of $\trun(\rheo{\varphi}) + \shift$
  agree, up to a constant multiple:
  \begin{equation}
    -3f(\rheo{\varphi}) = \scal(\trun(\rheo{\varphi}) + \shift)  .
  \end{equation}
  In other words, the following triangle of sheaves on $\R^{3} \times |\para|$
  commutes:
  \begin{equation}
    \begin{tikzcd}
      \Rheo \ar[dr, "-3f \quad \quad" below] \ar[r, "\trun + \shift"] & \Conv \ar[d, "\scal"] \\
      & \Omega^0_{\R^{3|2}/\para} .
    \end{tikzcd}
  \end{equation}
\end{prop}

\begin{proof}
  Let $(\rheo{\omega}, \rheo{e}, \rheo{\psi}, \rheo{B})$ be rheonomic. By Prop.
  \ref{prop:shift}, $(\rheo{\omega} + \shift, \rheo{e}, \rheo{\psi})$ is
  conventional. To get our hands on the superspace scalar of this conventional
  superfield, we need to compute the torsion of the shifted spin connection,
  $\rheo{\omega} + \shift$:
  \begin{eqnarray}
    T & = & d_{\rheo{\omega + \shift}} \rheo{e} + \frac{1}{2} [\rheo{\psi}, \rheo{\psi}] \\
      & = & d_{\rheo{\omega}} \rheo{e} + \shift \wedge \rheo{e} + \frac{1}{2} [\rheo{\psi}, \rheo{\psi}] \\
      & = & -3f \rheo{e}^* \wedge \rheo{e} ,
  \end{eqnarray}
  where, in the last line, we have used the torsion constraint and the definition of
  $\shift$. In a basis, we have thus shown:
  \begin{equation}
    g_{bd} T^d_{ab} = -3 \epsilon_{abc} f .
  \end{equation}
  This component of the torsion is the superspace scalar, implying the proposition.
\end{proof}

\noindent
We can now compute the pairing of $\Y_\susy$ with the Lagrangian.

\begin{prop}
  If $\rheo{\varphi} = (\rheo{\omega}, \rheo{e}, \rheo{\psi}, \rheo{B})$ is
  rheonomic, then the pairing of the geometric Lagrangian $\Lag(\rheo{\varphi})$ with
  the supersymmetric PCO $\Y_\susy(\rheo{\varphi})$ has the following form:
  \begin{equation}
    \Y_\susy(\rheo{\varphi}) \cdot \Lag(\rheo{\varphi}) = \ber(\rheo{E}) f(\rheo{\varphi}) ,
  \end{equation}
\end{prop}

\noindent
The proof is a calculation. The equations on the curvatures from Prop.
\ref{prop:rheo} rule out all terms in the Lagrangian that could pair nontrivially
with the 3-vector $g(\rheo{e}^{-1}, [\rheo{\psi}^{-1}, \rheo{\psi}^{-1}])$ in the
supersymmetric PCO, except for terms proportional to
$f g(\rheo{e}, [\rheo{\psi}, \rheo{\psi}])$. The result follows up to a constant of
proportionality, which we ignore.

\begin{thm}
 Let $\rheo{\varphi}$ be an rheonomic superfield, and
  $\trun(\rheo{\varphi}) + \shift$ the corresponding conventional superfield. Then
  the geometric action of $\rheo{\varphi}$ equals the superspace action of
  $\trun(\rheo{\varphi}) + \shift$:
  \begin{equation}
   S_\geom(\rheo{\varphi}) = S_\susy(\trun(\rheo{\varphi}) + \shift) .
  \end{equation}
\end{thm}

\begin{proof}
  As with the proof of Theorem \ref{thm:geomiscomp}, the proof is by choice of
  picture changing operator. Naturally, we use the supersymmetric PCO, $\Y_\susy$:
  \begin{eqnarray}
    S_\geom(\rheo{\varphi}) & = & \int_{\R^{3|2}/\para} \Y_\susy(\rheo{\varphi}) \cdot \Lag(\rheo{\varphi}) \\
    & = & \int_{\R^{3|2}/\para} \ber(\rheo{E}) f(\rheo{\varphi}) \\
    & = & \int_{\R^{3|2}/\para} \ber(\rheo{E}) \scal(\trun(\rheo{\varphi}) + \shift) \\
    & = & S_\susy(\trun(\rheo{\varphi}) + \shift) .
  \end{eqnarray}
  This completes the proof.
\end{proof}

\subsection{The framework of picture changing operators}
\label{sec:framework}

Everything we have done above generalizes immediately to classical field theories on
supermanifolds in any dimension. Let $\mani_{/\para} = \fib \times \para_{/\para}$ be
a trivial family of supermanifolds with oriented fiber of dimension $m|n$.

\begin{defn}
  A \define{geometric Lagrangian field theory} on $\mani_{/\para}$ consists of
  \begin{itemize}
    \item A sheaf $\Geo$ of superfields on $\mani_{/\para}$ consisting of
          differential forms twisted by vector bundles on $\mani_{/\para}$;
    \item A sheaf $\GeoComp$ of component fields on the even submanifold
          $\mani^\ev_{/\para}$ consisting of differential forms twisted by vector
          bundles on $\mani^\ev_{/\para}$, such that
          $\iota^*_\can (\Geo) = \GeoComp$;
    \item A rheonomic parametrization $\Rheo \subseteq \Geo$, in the sense that
          $\iota^*_\can \maps \Rheo \to \GeoComp$ is an isomorphism;
    \item A Lagrangian given as a differential $m$-form, $\Lag \maps \Geo \to \Omega^m_{\mani/\para}$.
  \end{itemize}
\end{defn}

\noindent
For any geometric Lagrangian field theory, we can define picture changing operators
exactly as we did for 3d, $\cN = 1$ supergravity. We can define the component
Lagrangian to be $\iota^*_\can \Lag$, and the component action as its integral.

\begin{thm}
  \label{thm:framework}
  If $\Lag$ is closed on $\Rheo$, then the geometric action
  \begin{equation}
    S(\rheo{\varphi}) = \int_{\mani/\para} \Y(\rheo{\varphi}) \cdot \Lag(\rheo{\varphi})
  \end{equation}
  is independent of the choice of picture changing operator $\Y$, and equivalent to
  the component action.
\end{thm}

\appendix

\section{Cartan Calculus on Integral Forms}
\label{app:Cartan}

For the sake of notation, we will only work over a single supermanifold $\mani$ of dimension $m|n$ -- the reader can easily generalize this to work over families. First, we recall the definition of differential on the complex of integral forms given in the main text. For $\ber \in \Ber_\mani$ and homogeneous vector fields $Y_1,\dots,Y_k$, set $\delta \maps \bersheaf^{m-k}_{\mani} \to \bersheaf^{m-k+1}_{\mani}$ with
\begin{align}
&\delta\bigl(\ber\otimes Y_1\wedge\cdots\wedge Y_k\bigr)
=
\sum_{i=1}^k (-1)^{i+1}\,
(-1)^{|Y_i|\,(|\ber|+\sum_{j<i}|Y_j|)}
\,(\ber\cdot Y_i)\otimes Y_1\wedge\cdots\widehat{Y_i}\cdots\wedge Y_k
\label{eq:delta} \nonumber\\
&\quad+\sum_{1\le i<j\le k} (-1)^{i+j}\,
(-1)^{|Y_i|\sum_{l<i}|Y_l| + |Y_j|\sum_{l<j,l\neq i}|Y_l|}
\,\ber\otimes [Y_i,Y_j]\wedge Y_1\wedge\cdots\widehat{Y}_i\cdots\widehat{Y}_j\cdots\wedge Y_k,
\end{align}
where the hatted fields are omitted.

In order to define a notion of Lie derivative, we first need the following.
\begin{definition}[Insertion $\iota_X$ on cochains]
Define
\[
\iota_X: \bersheaf^p_\mani \longrightarrow \bersheaf^{p-1}_\mani
\]
on homogeneous $\ber \otimes P \in \bersheaf_\mani^p$ by
\begin{equation}\label{eq:iota}
\iota_X(\ber \otimes P)\;\defeq\;(-1)^{|X||\ber|}\,\ber\otimes (X\wedge P),
\end{equation}
for $ P\in \PV^{m-p}_\mani$.
In particular, $\iota_X$ \emph{decreases} the cohomological degree by $1$.
\end{definition}

\noindent
Moreover, we have a Schouten derivation on polyfields, which satisfies the following
obvious lemma.

\begin{lemma}[Schouten derivation rule]\label{lem:sch-der}
Let $X \in \mathcal{PV}^{-1}_X$ be homogeneous and let $P=Y_1\wedge\cdots\wedge Y_k \in \mathcal{PV}^{-k}_{\mani}$ be decomposable with homogeneous $Y_i$.
Then
\begin{equation}\label{eq:sch-formula}
[X,P]_{\mathpzc{Sch}}
=
\sum_{i=1}^k (-1)^{|X|\sum_{j<i}|Y_j|}
\,Y_1\wedge\cdots\wedge [X,Y_i]\wedge\cdots\wedge Y_k.
\end{equation}
\end{lemma}

\begin{proof} Just iterate the graded Leibniz rule
$[X,Y\wedge Z]_{\mathpzc{Sch}}=[X,Y]_{\mathpzc{Sch}}\wedge Z+(-1)^{|X||Y|}Y\wedge[X,Z]_{\mathpzc{Sch}}$.
\end{proof}
 
\begin{definition}[Lie derivative $\Lie_X$ on $\bersheaf_{\mani}^{p}$]
Define $\Lie_X\maps \bersheaf^{p}_\mani \to \bersheaf^{p}_\mani $ by
\begin{equation}\label{eq:L-def}
\Lie_X(\ber \otimes P)\;\defeq\;(\ber\cdot X)\otimes P\;+\;(-1)^{|X||\ber|}\,\ber\otimes [X,P]_{\mathpzc{Sch}}.
\end{equation}
\end{definition}
As defined, the Lie derivative satisfies the following formula


\begin{theorem}[Cartan magic formula]\label{prop:cartan}
For every homogeneous vector field $X$ one has an identity of endomorphisms of $\bersheaf^\bullet_\mani$:
\begin{equation}\label{eq:cartan}
\Lie_X \;=\;[\delta,\iota_X]\;=\;\delta \circ \iota_X+(-1)^{|X|}\iota_X \circ \delta.
\end{equation}
\end{theorem}

\begin{proof}
It suffices to check \eqref{eq:cartan} on homogeneous decomposable elements
\[
c=\ber \otimes (Y_1\wedge\cdots\wedge Y_k)\in \bersheaf^p_\mani,
\qquad k=m-p.
\]
Using \eqref{eq:iota},
\[
\iota_X(c)=(-1)^{|X||\ber|}\,\ber \otimes (X\wedge Y_1\wedge\cdots\wedge Y_k)\in \bersheaf^{p-1}(M).
\]
Apply $\delta$ using \eqref{eq:delta}. The resulting terms split into three families.

\medskip\noindent
{(1) Coefficient term hitting $X$.}
In the first sum of \eqref{eq:delta}, the contribution obtained by removing $X$ yields
\[
(-1)^{|X||\ber|}\,(\ber\cdot X)\otimes (Y_1\wedge\cdots\wedge Y_k),
\]
which is exactly the first term of $\Lie_X(c)$ in \eqref{eq:L-def}.

\medskip\noindent
{(2) Coefficient terms hitting some $Y_i$.}
All other coefficient terms in $\delta(\iota_X c)$ have the form
\[
(\pm)\,(-1)^{|X||\ber|}\,(\ber\cdot Y_i)\otimes\bigl(X\wedge Y_1\wedge\cdots\widehat{Y_i}\cdots\wedge Y_k\bigr).
\]
On the other hand, applying first $\delta$ to $c$ and then $\iota_X$ produces the same family:
\[
\iota_X(\delta c)=
(\pm)\,(-1)^{|X||\ber|}\,(\ber \cdot Y_i)\otimes\bigl(X\wedge Y_1\wedge\cdots\widehat{Y_i}\cdots\wedge Y_k\bigr),
\]
with the same Koszul sign $(\pm)$ (the only extra sign is already accounted for by the factor $(-1)^{|X||\ber|}$ in \eqref{eq:iota}).
It follows that these terms cancel in $\delta\,\iota_X(c)-(-1)^{|X|}\iota_X\,\delta(c)$.

\medskip\noindent
{(3) Bracket terms.}
The bracket part of $\delta(\iota_X c)$ contains terms involving $[Y_i,Y_j]$, always preceded by a wedge with $X$, and terms involving $[X,Y_i]$.\\
The first subfamily matches exactly the bracket terms of $\iota_X(\delta c)$ and cancels in the commutator.
The second subfamily has no counterpart in $\iota_X(\delta c)$ and survives, giving
\[
(-1)^{|X||\ber|}\,\ber \otimes [X,\,Y_1\wedge\cdots\wedge Y_k]_{\mathpzc{Sch}}.
\]
By Lemma~\ref{lem:sch-der} this equals the second term of $\Lie_X(c)$ in \eqref{eq:L-def}.
Combining (1)--(3) proves \eqref{eq:cartan}.
\end{proof}
Finally, by a similar reasoning, it is not hard to see that one also have the remaining Cartan calculus identities, 
\[
[\Lie_X, \iota_Y]\; = \; \iota_{[X,Y]}, \qquad [\delta, \Lie_X] \;= \; 0,
\]
for any vector fields $X$ and $Y.$\\

Finally, we study deformations of integral forms along families of super diffeomorpshims. Namely, we consider the following setting: 
we let $\{\Phi_t\}_{t\in[0,1]}$ be a smooth one-parameter family of (super)diffeomorphisms of a supermanifold $\mani$, with $\Phi_0=\mathrm{id}$.
Let $\xi_t$ be the time-dependent vector field generating the isotopy, defined by
\begin{equation}\label{xit}
\xi_t \defeq \dot\Phi_t\circ \Phi_t^{-1},
\qquad\text{so that}\qquad
\frac{d}{dt}\,\Phi_t = \xi_t\circ \Phi_t.
\end{equation}
Then, in this setting, the following holds, which proves that a Spencer classes are stable under deformations along super diffeomorphisms.
\begin{lemma} \label{lemm:trangs} Let $\sigma$ be any integral form in $\Ber^\bullet (\mani) $ and let $\sigma_t$ be the $t$-dependent family defined by
\(
\sigma_t \defeq \Phi_t^{*}\sigma
\)
for $\{ \Phi_t \}_{t\in[0,1]}$ a smooth one-parameter family of (super)diffeomormorphism $\mani$, with $\Phi_0=\mathrm{id}$. 
If $\delta \sigma = 0,$ then there exists $\omega \in \Ber^{\bullet -1}(\mani)$ such that 
\[
\sigma_1 - \sigma_0 = \delta \omega.
\]
In particular, $[\sigma_1] = [\sigma_0] \in H^\bullet_{\mathpzc{Sp}} (\mani).$ 
\end{lemma}
\begin{proof} We find an explicit expression for $\omega$. First, we observe that, in the above setting, the derivative of a pullback yields
\begin{equation}\label{eq:pullback-derivative}
\frac{d}{dt}\,\sigma_t
=
\frac{d}{dt}\,\Phi_t^{*}\sigma
=
\Phi_t^{*}(\mathcal L_{\xi_t}\sigma),
\end{equation}
where $\mathcal L_{\xi_t}$ is the Lie derivative on integral forms. \\
Thanks to Cartan's magic formula, Theorem \ref{prop:cartan}, we have 
\begin{equation}\label{eq:cartan}
\mathcal L_{\xi_t}
=
\delta \,\iota_{\xi_t}\pm\iota_{\xi_t}\,\delta.
\end{equation}
Since $\delta \sigma = 0$, combining \eqref{eq:pullback-derivative} and \eqref{eq:cartan} we obtain
\begin{equation}\label{eq:exact-variation}
\frac{d}{dt}\,\sigma_t
=
\Phi_t^{*}\big(\delta(\iota_{\xi_t}\sigma)\big)
=
\delta\big(\Phi_t^{*}\iota_{\xi_t}\sigma\big).
\end{equation}
Integrating \eqref{eq:exact-variation} from $t=0$ to $t=1$ gives the explicit \emph{homotopy/transgression formula}
\begin{equation}\label{eq:transgression}
\sigma_1-\sigma_0
=
\delta\Omega,
\qquad
\Omega
\defeq
\int_0^1 \Phi_t^{*}\!\big(\iota_{\xi_t}\sigma\big)\,dt,
\end{equation}
which concludes the verification.
\end{proof}

\section{Compactly Supported Poincar\'e Lemma}
\label{app:csintegral}

In this appendix, we will construct a homotopy to concretely compute the the compactly supported Poincar\'e lemma for integral forms. 

We will work with the trivial fiber bundle
$\pi \maps \mani \times \R^{0|1} \to \mani$, which has purely odd fibers. We start noticing that the
topology of the total space is such that
$ |\mani|$, since $\R^{0|1}$ has the
topology of a point. 
and that the Berezinian sheaf of $\mani \times \R^{0|1} $ reads
$\bersheaf_{\mani \times \R^{0|1}} = \bersheaf_\mani \boxtimes \bersheaf_{\R^{0|1}}$\footnote{This
  witnesses the product structure on the Koszul (co)homology defining the Berezinians
  of $\mani$ and $\R^{0|1}$. Concretely, one has
  $[dx_1 \ldots dx_p \otimes \partial_{\theta_1} \ldots \partial_{\theta_q}] \otimes [\partial_{\psi}] \stackrel{\sim}{\longmapsto} [dx_1 \ldots dx_p \otimes \partial_{\theta_1} \ldots \partial_{\theta_q} \otimes \partial_{\psi}]$,
  where
  $[dx_1 \ldots dx_p \otimes \partial_{\theta_1} \ldots \partial_{\theta_q}] = \ber_{\mani}$
  and $[\partial_{\psi}] = \ber_{\R^{0|1}}$ are generators for the
  $\mathcal{E}xt$-sheaves defining the Berezinians.}.
Furthermore, in the above setting, every integral form is locally finite. More precisely, it is a locally finite sum
\begin{eqnarray}
    \sigma = \sum_{k} \sigma^{(k)} ,
  \end{eqnarray}
  where for every nonnegative integer $k$, $\sigma^{(k)}$ has the form
  \begin{equation} \sigma^{(k)} = (\pi^* \sigma \otimes \ber_{\R^{0|1}}) \otimes f(x,\psi) \partial^k_{\psi} ,
  \end{equation}
  where $\sigma$ is some integral form on $\mani$, $\ber_{\R^{0|1}}$ is a constant Berezinian on $\R^{0|1}$, $f$ is a regular functions on $\mani \times \R^{0|1}$, and finally
  $\partial_{\psi}^k = \partial_{\psi} \wedge \ldots \wedge \partial_\psi$ is the $k$-th wedge power. We call $\sigma^{(k)}$ an integral form of \define{type $k$}. Given these preliminaries, we now focus on compactly supported
integral forms on the product supermanifold $\mani \times \R^{0|1}.$

\begin{definition}[The maps $\pi_*$ and $e_*$]
  Let $\mani$ be a supermanifold and let $\bersheaf^\bullet_{\mathpzc{c}} (\mani
  \times \R^{0|1})$ and $\Omega^\bullet_{\mathpzc{c}} (\mani)$ the $
  \R$-modules of compactly supported integral forms on $\mani \times
  \R$ and $\mani$ respectively. We let the push-forward
  \begin{eqnarray}
    \xymatrix@R=1.5pt{
    \pi_* \maps \bersheaf^\bullet_{\mathpzc{c}} (\mani \times \R^{0|1}) \ar[r] & \bersheaf^\bullet_{\mathpzc{c}} (\mani)\\
    \sigma \ar@{|->}[r] & \pi_* \sigma
                          }
  \end{eqnarray}
  be defined on type $k$ integral forms as follows:
  \begin{eqnarray}
    \left \{
    \begin{array}{lr}
      \pi_{*} (\sigma^{(k)}) = \sigma ( \frac{\partial}{\partial \psi} f(x, \psi) )\vert_{\psi = 0}, & \mbox{ for } k = 0 , \\
      \pi_{*} (\sigma^{(k)}) = 0, & \mbox{ for } k \geq 1 ,
    \end{array}
    \right.
  \end{eqnarray}
  for $\sigma^{(k)}$ as above and $\sigma \in
  \bersheaf^\bullet_{\mathpzc{c}} (\mani).$ Also, we define the map
  \begin{eqnarray}
    \xymatrix@R=1.5pt{
    e_* \maps \bersheaf^\bullet_{\mathpzc{c}} (\mani) \ar[r] & \bersheaf^\bullet_{\mathpzc{c}} (\mani \times \R^{0|1})\\
    \sigma \ar@{|->}[r] & \pi^* \sigma \otimes \ber_{\R^{0|1}}  \psi.}
  \end{eqnarray}
\end{definition}

\noindent
The following hold.
\begin{lemma} Let $\mani$ be a supermanifold. Then the map $\pi_* \maps \emph{\bersheaf}^\bullet_\mathpzc{c} (\mani \times \R) \to \emph{\bersheaf}^\bullet_\mathpzc{c} (\mani)$ and the map $e_* \maps \emph{\bersheaf}^\bullet_{\mathpzc{c}} (\mani) \to \emph{\bersheaf}^{\bullet }_{\mathpzc{c}} (\mani \times \R^{0|1}) $ are cochain maps, \emph{i.e.}\ they both commute with the differential $\delta$.
In particular, they induce maps in cohomology as follows
\begin{align}
& [\pi_*] \maps H^\bullet_{\mathpzc{Sp}, \mathpzc{c}} (\mani \times \R^{0|1}) \longrightarrow H^{\bullet}_{\mathpzc{Sp}, \mathpzc{c}} (\mani), \\
& [e_*] \maps H^{\bullet}_{\mathpzc{Sp}, \mathpzc{c}} (\mani) \longrightarrow H^{\bullet}_{\mathpzc{Sp}, \mathpzc{c}} (\mani \times \R^{0|1}).
\end{align}
\end{lemma}
\begin{proof} This is a direct computation.
\end{proof}

\begin{lemma} \label{homotopySP} Let $\mani$ be a supermanifold and let $\pi_*$ and $e_*$ be defined as above. The composition map $e_* \circ \pi_* $ is homotopic to $1$ on $\Ber^\bullet_\mathpzc{c} (\mani \times \R)$. More precisely, there exists a chain homotopy
  \begin{eqnarray}
    K \maps \emph{\bersheaf}^\bullet_{\mathpzc{c}} (\mani\times \R^{0|1}) \to \emph{\bersheaf}^{\bullet }_{\mathpzc{c}} (\mani \times \R^{0|1})
  \end{eqnarray}
  such that
$(1 - e_* \pi_*) (\sigma) = (\delta K + K \delta) (\sigma)$,
for any $\sigma \in \emph{\bersheaf}^{\bullet}_{\mathpzc{c}} (\mani \times \R^{0|1}).$ In particular, the map  $[\pi_*] \maps H^\bullet_{\mathpzc{Sp, c}} (\mani \times \R^{0|1}) \to H_{\mathpzc{Sp, c}}^{\bullet} (\mani)$ defines an isomorphism in cohomology, with inverse map $[e_*] \maps H^\bullet_{\mathpzc{Sp, c}} (\mani) \to H^{\bullet}_{\mathpzc{Sp, c}} (\mani \times \R^{0|1}),$ \emph{i.e.}
\begin{eqnarray}
\xymatrix{
H^\bullet_{\mathpzc{Sp,c}} (\mani \times \R^{0|1})  \ar@/^1pc/[rr]|{\pi_*}
 &&  \ar@/^1pc/[ll]|{e_*}  H^{\bullet }_{\mathpzc{Sp,c}} (\mani).
}
\end{eqnarray}
\end{lemma}
\begin{proof}
  We work locally on $\mani \times \R^{0|1}$ and we let $\sigma $ be a generic integral form of type $k$, which we write in the following form:
  \begin{align}
  \sigma = \ber f \otimes (\partial_\psi)^k F.
  \end{align} 
  Here $\ber = \ber (x, \psi)$ is a generating section of the Berezinian, $f = f (x, \psi) $ is a (compactly supported) section of the structure sheaf, $F =  \partial^I$ is a polyvector field for some multi-index $I$ in the local coordinates over $\mani$, and we have isolated the polyvector field in the coordinate $\psi$ on $\R^{0|1}$ for convenience.

 One computes:
\begin{align}
\delta (\sigma) = &- \sum_a (-1)^{|x_a||f| + |F|} \ber (\partial_a f) \otimes (\partial_{\partial_a} F) ( \partial_\psi)^k + \nonumber \\
& \qquad \qquad +  (-1)^{|f| + |F|} \ber (\partial_\psi f) \otimes F k ( \partial_\psi)^{k-1} \Theta(k - 1/2).
\end{align}
Notice that the last term does not appear when the integral form is of type 0: to account for this, we have multiplied it by the
Heaviside step function $\Theta(k-1/2)$, instead of listing the two cases.

We claim that
\begin{eqnarray}
K (\sigma) = (-1)^{\epsilon_\sigma} \ber \psi f \otimes \frac{1}{k+1} ( \partial_\psi)^{k+1} F,
\end{eqnarray}
for $\epsilon_{\sigma}$ a choice of sign depending on $\sigma$ to be fixed later. Notice that the map $K$ affects only the part of $\sigma$ which is related to $\R^{0|1}$. \\
On the one hand, one computes
\begin{align}
K \delta ( \sigma ) & = - \sum_{a} (-1)^{\epsilon_{\delta \sigma}} (-1)^{|x_a||f| + |F|} \ber \psi (\partial_a f) \otimes (\partial_{ \partial_a} F ) \otimes \frac{1}{k+1} ( \partial_\psi)^{k+1} \nonumber \\
& \quad - (-1)^{|f| + |F|} (-1)^{\epsilon_{\delta \sigma}} \ber \psi (\partial_\psi f ) \otimes F \otimes (\partial_\psi  )^k \Theta(k - 1/2).
\end{align}
On the other hand, one has
\begin{align}
\delta K (\sigma) & = - \sum_a (-1)^{\epsilon_\sigma} (-1)^{|x_a| (|f| + 1) + |F|} \ber \partial_a (\psi f ) \otimes \frac{1}{k+1}( \partial_\psi)^{k+1} (\partial_{ \partial_a } F)  \nonumber \\
& \quad - (-1)^{\epsilon_\sigma} (-1)^{|f| + |F| + 1} \ber \partial_\psi (\psi f) \otimes F ( \partial_\psi)^{k} \Theta(k - 1/2).
\end{align}
Setting $\epsilon_{\delta \sigma} = - \epsilon_{\sigma}$, one has
\begin{align}
(K \delta + \delta K)(\sigma) 
& = (-1)^{\epsilon_\sigma + |f| + |F|} \ber (f- \delta_{k, 0}\psi \partial_\psi f) \otimes F (\partial_\psi)^k.
\end{align}
It follows that, posing $\epsilon_{\sigma} = |f| + |F|$, $K$ defines a homotopy whereas $k \neq 0$. In particular, one has $1 - e_* \pi_* = K \delta + \delta K.$ \\
For completeness, notice that, having fixed the sign, the action of the homotopy $K$ on a generic integral form $\sigma = \ber f \otimes ( \partial_\psi)^k F$ as above reads
\begin{eqnarray}
K (\sigma) = (-1)^{|f| + |F|} \frac{1}{k+1}  \ber \psi f \otimes (\partial_\psi)^{k+1} F.
\end{eqnarray}
This concludes the proof.
\end{proof}

\noindent
It is worth emphasizing that the homotopy established in Lemma \ref{homotopySP} is ``algebraic" in the sense that it does not involve integration. This reflects the fact that Berezin integration in the odd directions is purely algebraic and fundamentally different from Lebesgue integration -- unsurprisingly so, given that odd directions are nilpotent.

The above Lemmas allow us to prove our main result, which mirrors part of Theorem \ref{thm:compsupp} in the main text of this paper. 

\begin{theorem}[Compactly Supported Poincar\'e Lemma---Integral Forms]
  \label{thm:CSPL}
  One has
  \begin{eqnarray}
    H^k_{\mathpzc{Sp, c}} (\R^{m|n}) \iso \left \{ \begin{array}{lll}
                                                              \R & &  k = m\\
                                                              0 & & k \neq m.
                                                            \end{array}
                                                                    \right.
  \end{eqnarray}
  In particular, $H^m_{\mathpzc{Sp,c}} (\R^{m|n})$ is generated by
  \begin{equation}
    [dx_1 \ldots dx_m \otimes \partial_{\theta_1} \ldots \partial_{\theta_n}]
    \theta_1 \ldots \theta_n \mathpzc{B}_m,
  \end{equation}
  where $\mathpzc{B}_m$ is a bump function which integrates to one on $\R^m$
  and $[dx_1 \ldots dx_m \otimes \partial_{\theta_1} \ldots \partial_{\theta_n}]$ is
  a generating section of the Berezinian $\emph{\bersheaf}^m (\R^{m|n})$.
\end{theorem}

\begin{proof}
  One starts from the cohomology of $\R^{m|0}$. In particular, one has that
  $H^k_{\mathpzc{Sp}, \mathpzc{c}} (\R^{m|0}) \iso H^k_{\mathpzc{dR},
  \mathpzc{c}}(\R^{m|0})$, as the Spencer complex is isomorphic to the de
  Rham complex for an ordinary manifold. By the ordinary Poincar\'e lemma for
  compactly supported differential forms, this cohomology is concentrated in degree
  $m$, where it is one-dimensional and spanned by the class of $dx^1 \cdots dx^m
  \mathpzc{B}_m$, for $\mathpzc{B}_m$ a bump function integrating to 1 on
  $\R^m$. By Lemma \ref{homotopySP}, repeated application of the isomorphism $[e_*]$
  yields the claim.
\end{proof}

\end{document}